\documentclass[format=acmsmall]{acmart}
\usepackage{acm-ec-26}
\usepackage{booktabs} 
\usepackage[ruled]{algorithm2e} 

\SetAlFnt{\small}
\SetAlCapFnt{\small}
\SetAlCapNameFnt{\small}
\SetAlCapHSkip{0pt}
\IncMargin{-\parindent}
\setcitestyle{authoryear}

	\usepackage{nicefrac}
	\usepackage{tikz}

 \usetikzlibrary{positioning,arrows,shapes,decorations,calc,fit,arrows.meta} 
	\usepackage{colortbl}
	\usepackage{booktabs}
	\usepackage{cleveref}
	\usepackage{array}
	\usepackage{xcolor}
	\usepackage{enumitem}
	\usepackage{tabularx}
	\usepackage{thm-restate}
	\usepackage{ragged2e}
    \usepackage{microtype}
    \usepackage{cleveref}
    \usepackage{mathrsfs} 
    \usepackage{caption}
    \usepackage{nicefrac}
    \usepackage{thm-restate}
    \usepackage{wrapfig}
    \usepackage{wrapstuff}
    \usepackage{pgfplots}
    \usetikzlibrary{calc}
\usepackage{etoolbox}
\makeatletter
\patchcmd\WF@putfigmaybe{\lower\intextsep}{}{}{\fail}%
\AddToHook{env/wrapfigure/begin}{\setlength{\intextsep}{0pt}}
\makeatother

\AtBeginDocument{%
\fancypagestyle{firstpagestyle}{
	\fancyhf{}%
	\fancyfoot[L]{}%
}
}

\title{Proportional Representation in Rank Aggregation}
\author{Patrick Lederer}
\email{p.lederer@uva.nl}
\affiliation{%
  \institution{ILLC, University of Amsterdam}
  \country{The Netherlands}
}

\usepackage{natbib}
\setcitestyle{nosort}

\newcounter{remarkcount}

\newcolumntype{L}[1]{>{\raggedright\let\newline\\\arraybackslash\hspace{0pt}}m{#1}}
\newcolumntype{C}[1]{>{\centering\let\newline\\\arraybackslash\hspace{0pt}}m{#1}}
\newcolumntype{R}[1]{>{\raggedleft\let\newline\\\arraybackslash\hspace{0pt}}m{#1}}

\usepackage{xcolor}
\usepackage{ifthen}
\definecolor{alt-1}{HTML}{FFB3B3}  
\definecolor{alt-2}{HTML}{FFD9A0}  
\definecolor{alt-3}{HTML}{FFF4A0}  
\definecolor{alt-4}{HTML}{B8F0C8}  
\definecolor{alt-5}{HTML}{A8D8F0}  
\definecolor{alt-6}{HTML}{C4B0F5}  
\definecolor{alt-7}{HTML}{F5B0D8}  
\newcommand{\alternative}[2][default]{%
	\tikz[transform shape,scale=0.85]{
		\def\colorparam{#1}
		\ifthenelse{\equal{#1}{default}}{
			\def\colorparam{alt-#2}
		}{}
		\clip (0,3pt) circle (6pt);
		\node[anchor=base, yshift=0.5pt, text height=6pt,inner sep=1pt,fill=\colorparam, circle,minimum width=12pt] {};
		\node[anchor=base, text height=6pt, inner sep=1pt,minimum width=12pt,scale=1.15] {$x_{#2}$};
	}
}
\newcommand{\borderalternative}[2][default]{%
	\tikz[transform shape,scale=0.85]{
		\def\colorparam{#1}
		\ifthenelse{\equal{#1}{default}}{
			\def\colorparam{alt-#2}
		}{}
		\node[anchor=base, yshift=0.5pt, text height=6pt,inner sep=1pt,fill=\colorparam, circle,minimum width=12pt,draw=black!70,very thin] {};
		\node[anchor=base, text height=6pt, inner sep=1pt,minimum width=12pt,scale=1.15] {$#2$};
	}
}
\newcommand{\weakorder}[2][noborder]{%
	\begin{tikzpicture}
		[yscale=0.5]
		\foreach \row [count=\y] in {#2} {
			\ifthenelse{\y>1 \OR \equal{#1}{noborder}}{
				\node [anchor=south, inner sep=0.5pt] at (0, -\y) {
					\tikz[xscale=0.37]{
						\foreach \val/\color [count=\x] in \row {
							\ifthenelse{\equal{\val}{\color}}
							{\def\color{default}} %
							{}
							\node [inner sep=0] at (\x, 0) {\alternative[\color]{\val}};
						};}
				};
			}{
				\node [anchor=south, inner sep=0.5pt] at (0, -\y) {
					\tikz[xscale=0.37]{
						\foreach \val/\color [count=\x] in \row {
							\ifthenelse{\equal{\val}{\color}}
							{\def\color{default}} %
							{}
							\node [inner sep=0] at (\x, 0) {\borderalternative[\color]{\val}};
						};}
				};
			}
		}
	\end{tikzpicture}
}
\newcommand{\votermultiplicity}[2]{%
	\begin{tikzpicture}[baseline]
		\node[anchor=base,inner sep=1pt] at (0,0) {#1};
		\node[anchor=north,inner sep=0] at (0,-0.2) {#2};
	\end{tikzpicture}
}
\pgfplotsset{
	shortl/.style={%
		legend image code/.code={
			\draw[##1,line width=1.6pt]
			plot coordinates {
				(0cm,0cm)
				(0.2cm,0cm)
			};%
		}
	},
}

\newcommand{\rmes}{\texttt{RMES}\xspace}
\newcommand{\psb}{\texttt{PSB}\xspace}
\newcommand{\fb}{\texttt{FB}\xspace}
\newcommand{\mes}{\texttt{MES}\xspace}

\begin{abstract}
    In rank aggregation, the task is to aggregate multiple weighted input rankings into a single output ranking. While numerous methods, so-called social welfare functions (SWFs), have been suggested for this problem, all of the classical SWFs tend to be majoritarian and are thus not acceptable when a proportional ranking is required. Motivated by this observation, we design SWFs that guarantee that every input ranking is proportionally represented  by the output ranking. 
    Specifically, our central fairness condition requires that the number of pairwise comparisons between candidates on which an input ranking and the output ranking agree is at least proportional to the weight of the input ranking. As our main contribution, we present a simple SWF called the Proportional Sequential Borda rule which satisfies this condition. Moreover, we introduce a more involved variant of this rule, the Flow-adjusting Borda rule, which satisfies a stronger fairness condition that applies to arbitrary groups of rankings. Many of our axioms and techniques are inspired by results in approval-based committee voting and participatory budgeting, where the concept of proportional representation has been studied in depth.
\end{abstract}

\begin{document}

\maketitle
\vspace{1cm}
\setcounter{tocdepth}{2} 
\tableofcontents
\newpage

\section{Introduction}

For booking flights or hotels, many users consult aggregator websites such as google flights, skyscanner, or booking.com. These websites allow users to get an overview of the available flights or hotels by offering several ways to sort the options. For instance, flights may be sorted by their price, travel time, or number of stop-overs, whereas hotels are commonly sorted by price or user rating. Furthermore, aggregator websites typically offer a recommended ranking that combines multiple criteria. More generally, it would be desirable to display customized rankings to the users by combining the rankings for different criteria according to user-specific weights.\footnote{Such user-specific weights could arise from various approaches. For instance, we could ask the users to specify the weights themselves, pre-configure different weights for different regions, or use a machine learning model to learn the users' weights based on their previous behavior.}

The problem of finding such aggregate rankings is commonly studied under the term \emph{rank aggregation} and has attracted significant attention in social choice theory \citep[e.g.,][]{ASS02a,BCE+14a} and beyond \citep[e.g.,][]{KLAV12a,SCS14a,WDF+24a}. In more detail, in rank aggregation, we get a profile of weighted rankings as input, where the weights are non-negative and add up to one, and we need to return a single output ranking. For example, this model precisely captures the situation where a user want to sort hotels to $60\%$ by their price and $40\%$ by their user rating. Moreover, rank aggregation has numerous further applications.
\begin{itemize}
    \item Closely related to booking websites, there are many ranking portals that produce a global ranking by aggregating rankings for different criteria. For example, a city ranking can be computed by aggregating, say, rankings for the quality of life, cost of living, and economic aspects, and a university ranking by aggregating rankings of the scientific impact, the success of their graduates on the job market, and student satisfaction.
    \item An increasingly popular application of rank aggregation is the evaluation of AI agents. These agents can be ranked according to their performance on different tasks or benchmarks, and the resulting rankings then need to be aggregated into an overall ranking. Recent research has advocated the use of social choice methods for this problem \citep[e.g.,][]{LLB+25a,LLK+25a}.
    \item While the previous examples focus on technical applications, rank aggregation also arises in classical voting settings. For instance, in university hiring committees, it is common to compute a ranking of the applicants based on the rankings reported by the committee members, and the job is offered to the applicants according to the final ranking.
\end{itemize}

To address such rank aggregation problems, numerous \emph{social welfare functions (SWFs)}, which map profiles of weighted input rankings to a single output ranking, have been suggested in social choice theory. Prominent examples of such SWFs include the Kemeny rule \citep{Keme59a,YoLe78a}, various types of scoring rules \citep{Smit73a,CRX09a,BBP23a,Lede23a}, and Condorcet-type rules \citep{Cope51a,FHN15a}. 
However, as observed by \citet{LPW24a}, all of these classic SWFs are heavily majoritarian: for example, when combining two inverse rankings with weights of $51\%$ and $49\%$, most of the classic SWFs will simply return the ranking with the larger weight instead of actually combining the rankings. 

While this behavior is desirable for some electorate settings, it disqualifies these methods from our intended applications. For instance, when computing an aggregated ranking of hotels with 60\% weight for the price ranking, 30\% for the user rating ranking, and 10\% for the location ranking, the output ranking should not simply copy the price ranking. Instead, a hotel should be able to compensate a low position in the price ranking with high positions in the other rankings. Similarly, when ranking AI models based on their performance on several tasks or applicants in a hiring committee, it seems desirable to faithfully follow the weights, so that the output ranking fairly represents the diverse interests across different tasks or academic departments.

Motivated by this observation, \citet{LPW24a} have initiated the study of proportional SWFs, aiming for methods that guarantee fair representation to all input rankings.\footnote{We note that proportionality is often interpreted as a fairness notion for the input rankings. In rank aggregation, there is another influential line of work that investigates fairness with respect to attributes of candidates \citep[e.g.,][]{KuRu20a,CDKS22a,WISB22a,PSK22a}. Specifically, the idea of these papers is that some candidates have protected attributes that should be fairly represented in the output ranking, regardless of the weights of the input rankings.} Specifically, these authors formalize proportional representation in terms of the number of pairs of candidates on which the input and output rankings agree: an input ranking with  weight $\alpha$ should agree at least with an $\alpha$ fraction of all pairwise comparisons of the output ranking. Moreover, \citet{LPW24a} suggest the Squared Kemeny rule to compute proportional rankings. However, while this rule is certainly more proportional than established SWFs, it does not satisfy the aforementioned fairness condition in general (see, e.g., \Cref{fig:placeholder} for an example). Consequently, the design of fully proportional SWFs remains as an open problem. 

\begin{figure}
    \centering
    \begin{tikzpicture}
		\node at (0,0) [draw=black!30] (step1) {	
			\votermultiplicity{30\%}{\weakorder{{1},{2},{3},{7},{6},{5},{4}}}
			\votermultiplicity{30\%}{\weakorder{{1},{7},{6},{5},{4},{3},{2}}}
            \votermultiplicity{30\%}{\weakorder{{2},{7},{5},{6},{4},{3},{1}}}
            \votermultiplicity{10\%}{\weakorder{{3},{4},{5},{6},{7},{2},{1}}}
		};

        \draw[line width = 0.2mm]
  ([yshift=-0.5cm,xshift=0.1cm]step1.north west) -- ([yshift=-0.5cm,xshift=-0.1cm]step1.north east);
		
		\draw[->] (1.9,-0.25) -- (2.7,-0.25);
		
		\node at (6.5,0) [draw=black!30] (step2) {	
			\votermultiplicity{\textcolor{white}{ller}Kemeny\textcolor{white}{ller}}{\weakorder{{1},{2},{7},{6},{5},{4},{3}}}
			\votermultiplicity{{Sq. Kemeny}}{\weakorder{{1},{2},{7},{5},{6},{4},{3}}}
			\votermultiplicity{Prop. Seq. Borda}{\weakorder{{7},{1},{2},{6},{3},{5},{4}}}
		};

        \draw[line width = 0.2mm]
  ([yshift=-0.5cm,xshift=0.1cm]step2.north west) -- ([yshift=-0.5cm,xshift=-0.1cm]step2.north east);
	\end{tikzpicture}
    \vspace{-7pt}
    \caption{Comparison of our Proportional Sequential Borda rule with the Kemeny and Squared Kemeny rules. On the left, we show a profile consisting of four input rankings and their respective weights. On the right, we show the output rankings chosen by three SWFs for this profile. In particular, the Kemeny rule chooses exactly the inverse of the right most input ranking $\succ_4$, so it is not proportional. Similarly, the Squared Kemeny ranking only agrees in one pairwise comparison with $\succ_4$, even though the proportional share of $\succ_4$ is $10\%\cdot {7\choose 2}\geq 2$. By contrast, the ranking chosen by the Proportional Sequential Borda rule, the method suggested in this paper, agrees with four pairwise comparisons of $\succ_4$ and is thus proportional. 
    }
    \label{fig:placeholder}
\end{figure}

\paragraph{Our Contribution.} 
In this paper, we design the first truly proportional SWFs by employing ideas from approval-based committee voting and participatory budgeting, two fields in which the concept of proportional representation has been studied extensively \citep[see, e.g.,][]{LaSk23a,RSM25a}. To explain our results, we define the \emph{utility} of an input ranking $\succ$ for an output ranking~$\rhd$ as the number of pairwise comparisons these rankings agree on, i.e., $u({\succ},\rhd)=|\{(x,y)\colon {x\succ y}\text{ and } x\rhd y\}|$. Since every ranking on $m$ candidates induces ${m\choose 2}$ pairwise comparisons between candidates, the proportionality axiom of \citet{LPW24a} formally requires that every input ranking with weight~$\alpha$ obtains a utility of at least $\lfloor\alpha{m\choose2}\rfloor$ from the output ranking. We call this condition \emph{unanimous proportional justified representation (uPJR)} because, when viewing rankings as approval ballots over pairs of candidates, it can be seen as a weakening of the well-known notion of proportional justified representation from approval-based committee voting \citep{SFF+17a}. Moreover, we introduce two further proportionality axioms called \emph{unanimous justified representation (uJR)} and \emph{strong proportional justified representation (sPJR)}, which respectively weaken and strengthen uPJR. 
As our main contribution, we design SWFs that satisfy these fairness conditions. Specifically, we will show the following results: 
\begin{itemize}
    \item As a warm-up, we will first analyze uJR, which demands that every input ranking with a weight of at least $1/{m\choose 2}$ obtains a non-zero utility from the output ranking. We show that the Squared Kemeny rule by \citet{LPW24a} severely fails this condition (\Cref{prop:SqK}), thus demonstrating the need of more proportional rules. Moreover, we also present a simple rule inspired by Chamberlin-Courant approval voting that satisfies uJR (\Cref{prop:CC}). 
    \item We next turn to uPJR and, inspired by similar results in approval-based committee voting and participatory budgeting \citep{PeSk20a,BFL+23a}, prove that uPJR is implied by a more structured fairness axiom which we call rank-priceability (\Cref{prop:RPimpliesuPJR}). Based on this insight, we design a simple SWF, the \emph{Proportional Sequential Borda rule (\psb)}, that satisfies rank-priceability and thus uPJR (\Cref{thm:propBordaRP}). Roughly, \psb repeatedly picks the Borda winner, deletes this candidate from the profile, and reduces the weight of each input ranking proportional to its contribution to the score of the Borda winner.
    \item We further analyze sPJR which extends the reasoning of uPJR to arbitrary groups of input rankings: if a group of input rankings has a total weight of $\alpha$, the output ranking should agree with at least $\lfloor\alpha{m\choose 2}\rfloor$ pairwise comparisons of these rankings. As we will show, \psb fails this condition because the guarantees provided by rank-priceability do not extend to groups of rankings. We therefore introduce a refinement of rank-priceability called pair-priceability and show that this notion implies sPJR (\Cref{prop:sPJR}). Moreover, we propose a variant of \psb called the Flow-adjusting Borda rule (\fb) that satisfies pair-priceability (\Cref{thm:flowBordaPR}) and thus sPJR. Notably, \fb only augments \psb by using a more sophisticated scheme for updating the weights of the input rankings.
    \item Inspired by an analogous result for the Squared Kemeny rule \citep[][Theorem~4.2]{LPW24a}, we examine the average utility our SWFs guarantee to a group of rankings, as a function of the total weight of the group. For both \psb and \fb, we show that the average utility of every group of rankings is at least linear in the weight of the group. In more detail, we prove that the average utility of every group of rankings with a total weight of $\alpha$ is at least $\frac{\alpha}{4}{m\choose 2}-\frac{3}{16}$. Hence, in addition to satisfying fairness axioms, both of our rules provide strong quantitative proportionality guarantees for all groups of rankings
(\Cref{thm:propBordaAR,thm:flowBordaAR}).
\end{itemize}

\paragraph{Related Work}

Rank aggregation is one of the oldest problems in social choice theory as even Arrow's impossibility theorem has originally been stated in this setting \citep{Arro51a}. Consequently, there is a large body of works on this topic and we refer to the handbooks by \citet{ASS02a} and \citet{BCE+14a} for a general overview of this area.

In the classic literature on SWFs, our paper is related to an influential stream of works that study scoring rules for rank aggregation. 
Specifically, positional scoring rules such as the Borda rule compute the output ranking by assigning scores to the candidates and sorting the candidates in decreasing order of their scores. Such SWFs have attracted significant attention and have, e.g., been repeatedly characterized \citep[e.g.,][]{Smit73a,Youn74b,NiRu81a}. Moreover, positional scoring rules can be modified to work sequentially by repeatedly adding the candidate with the highest score to the next best (or worst) position of the output ranking and then deleting it from the profile \citep[e.g.,][]{Smit73a,FBC14a,BBP23a}. The Proportional Sequential Borda rule and the Flow-adjusting Borda rule are closely related to this approach as we repeatedly add the Borda winner to the output ranking, but we carefully update the weights of the input rankings to obtain a proportional outcome.
Furthermore, our work is related to the Kemeny rule, another well-studied SWF \citep[e.g.,][]{Keme59a,YoLe78a,Youn88a,Cast13a}, as this rule maximizes the utilitarian social welfare in our setting.

While the aforementioned works are influential, they do not focus on proportionality. Indeed, the study of proportional rank aggregation has only recently been initiated by \citet{LPW24a}, who suggested the Squared Kemeny rule. Moreover, \citet{ALP+25b} have designed  committee voting rules that satisfy committee monotonicity and proportionality for solid coalitions (PSC), a well-established proportionality notion for committee elections \citep[e.g.,][]{Dumm84a,Tide95a,AzLe20a}. When interpreted as SWFs, their rules satisfy an approximation of uPJR and additionally guarantee that every prefix of the output ranking satisfies PSC. Further, in this paper we draw on ideas from approval-based committee voting \citep{LaSk23a} and participatory budgeting \citep{AzSh21a,RSM25a}. In particular, our proportionality axioms are inspired by notions from this literature \citep[e.g.,][]{ABC+16a,SFF+17a,PeSk20a,BFL+23a}, and the budgeting approach of our SWFs is reminiscent of the Method of Equal Shares \citep{PeSk20a,PPS21a}. 

Finally, we note that proportionality has been studied in numerous other models, some of which are related to rank aggregation. In particular, \citet{SLB+17a} and \citet{BrIs25a} study the problem of proportionally aggregating the voters' approval ballots into a ranking, but they focus on a prefix-based fairness concept that is logically unrelated to uPJR. Moreover, there are models of repeated decision making \citep[e.g.,][]{Lack20a,BHP+21a,LaMa23a,CGP24a}, where a candidate needs to be selected in each round and each group of voters should be fairly represented over all rounds. However, this literature again focuses on conceptually similar but mathematically different proportionality notions. Finally, \citet{MPS24a} study a general model of proportional decision making, which contains rank aggregation as a special case. However, applying these results to rank aggregation gives only very mild guarantees and, e.g., allows that rankings with weight less than $\frac{1}{m}$ are left without any representation.

\section{Preliminaries}

Let $C=\{x_1,\dots, x_m\}$ denote a set of $m$ candidates. A \emph{ranking} $\succ$ is a strict linear order over $C$, and we write rankings as comma-separated lists. For instance, ${\succ}=x_1,x_2,x_3$ means that $x_1$ is ranked ahead of $x_2$ and $x_2$ ahead of $x_3$. The set of all rankings over $C$ will be denoted by $\mathcal{R}$. Following \citet{LPW24a}, we define \emph{(ranking) profiles} $R$ as functions from $\mathcal{R}$ to $[0,1]$ such that $\sum_{{\succ}\in\mathcal{R}} R({\succ})=1$. Less formally, a ranking profile specifies for every ranking ${\succ}\in\mathcal{R}$ a weight $R({\succ})$ and the total weight sums up to $1$. These weights may be interpreted as the importance scores in a multi-criteria decision-making problem or as the shares of voters that report a given ranking. Furthermore, we say a function $S:\mathcal R\rightarrow [0,1]$ is a \emph{subprofile} of a profile $R$ if $S({\succ})\leq R({\succ})$ for all ${\succ}\in\mathcal{R}$. In a voting setting, a subprofile $S$ can be interpreted as an arbitrary group of voters. The size of a subprofile $S$ is defined by $|S|=\sum_{{\succ}\in\mathcal{R}} S({\succ})$, which means that $|S|\in [0,1]$ for every subprofile $S$. 

Given a ranking profile, our goal is to aggregate the input rankings into one output ranking. For this problem, we use \emph{social welfare functions (SWFs)}, which are functions that map every ranking profile to a single output ranking. To clearly distinguish between input rankings and output rankings, we will write $\succ$ for the former and $\rhd$ for the latter. The assumption that SWFs always choose a single output ranking will sometimes require tie-breaking as multiple rankings can be tied for the win. We hence assume that there is a fixed tie-breaking order over the rankings and note that this assumption does not affect our results. Indeed, all our results also hold when viewing our rules as functions that return (non-empty) sets of rankings, but this model introduces unnecessary notational complexity. 

\subsection{Proportionality Axioms}

The goal of this paper is to find output rankings that proportionally represent the input rankings: each input ranking should have an influence proportional to its weight on the output ranking. Following \citet{LPW24a}, we will formalize this idea by requiring that each input ranking with weight $\alpha$ should agree at least with an $\alpha$ fraction of the pairwise comparisons of the output ranking. Moreover, we note that we define our proportionality axioms as properties of rankings; an SWF $f$ satisfies a given axiom if its chosen ranking $f(R)$ satisfies the axiom for all profiles $R$.

To formalize our proportionality axioms, we define the \emph{utility} of a ranking~$\succ$ for another ranking~$\rhd$ by $u({\succ}, \rhd)=|\{(x,y)\in C^2\colon x\succ y \text{ and } x\rhd y\}|$. That is, the utility of an input ranking $\succ$ for an output ranking $\rhd$ is the number of pairs of candidates on which the rankings agree. Furthermore, we let $u({\succ}, x, X)=|\{y\in X\setminus \{x\}\colon x\succ y\}|$ denote the utility of a candidate $x$ within the set of candidates $X$ with respect to $\succ$. Alternatively, $u({\succ}, x, X)$ can also be interpreted as the Borda score of $x$ within the set $X$. This term will be crucial in our analysis because $u({\succ}, \rhd)=\sum_{i=1}^{m-1} u({\succ}, x_i, \{x_i,\dots, x_m\})$ for every input ranking $\succ$ and output ranking $\rhd=x_1,\dots, x_m$. We further note that the utility $u({\succ}, \rhd)$ is dual to the swap distance $\mathit{swap}({\succ},\rhd)=|\{(x,y)\in C^2\colon x\succ y\land y\rhd x\}|$ studied by \citet{LPW24a}, because $u({\succ}, \rhd)={m\choose 2}-swap({\succ}, \rhd)$ for all rankings ${\succ},\rhd\in\mathcal{R}$. 

We will now introduce our first fairness condition called unanimous proportional justified representation (uPJR), which requires that the utility of every ranking should be at least proportional to its weight. Both \citet{LPW24a} and \citet[][Section 8]{ALP+25b} have investigated this condition, but these authors only present SWFs that satisfy approximate versions of uPJR.

\begin{definition}[Unanimous Proportional Justified Representation]
    A ranking $\rhd$ satisfies \emph{unanimous proportional justified representation (uPJR)} for a profile $R$ if $u({\succ},\rhd)\geq \lfloor R({\succ})\cdot {m\choose 2}\rfloor$ for all ${\succ}\in\mathcal{R}$.  
\end{definition}

We note that uPJR can be seen as a weakening of proportional justified representation (PJR), a well-known fairness condition for approval-based committee voting \citep{SFF+17a}. In this setting, a set of voters $N=\{1,\dots, n\}$ report approval ballots $A_i\subseteq C$ over the candidates and the goal is to choose a subset of the candidates of predefined size $k$. Then, PJR requires that, if a group of voters $S$ is large enough to deserve $\ell$ seats and the voters in $S$ agree on $\ell$ approved candidates, the winning committee should contain at least $\ell$ candidates that are approved by voters in $S$. More formally, a set of candidates $W$ of size $k$ satisfies PJR for an approval profile $A$ if $|W\cap \bigcup_{i\in S} A_i|\geq \ell$ for every group of voters $S$ with $|S|\geq \frac{\ell n}{k}$ and $|\bigcap_{i\in S} A_i|\geq \ell$. 

Now, PJR can be naturally adapted to rank aggregation by associating each ranking $\succ$ with the set $A({\succ})=\{(x_i,x_j)\in C\times C\colon x_i\succ x_j\}$, which can be interpreted as an approval ballot over $C^2=\{(x_i,x_j)\in C\times C\colon x_i\neq x_j\}$. Hence, we can view the problem of rank aggregation as an instance of approval-based committee voting over the set $C^2$, with the additional constraint that the chosen committee must correspond to a valid ranking. Specifically, given the input ballots $A({\succ})$ with weights $R({\succ})$, we need to choose a transitive and asymmetric subset of $C^2$ of size $k={m\choose 2}$. Applying PJR to this instance of committee voting results in the following condition: a ranking $\rhd$ satisfies PJR for a profile $R$ if $|A(\rhd)\cap \bigcup_{{\succ}\in\mathcal{R}\colon S({\succ})>0} A({\succ})|\geq \ell$ for every subprofile $S$ of $R$ with $|S|\geq \ell/ {m\choose 2}$ and $|\bigcap_{{\succ}\in\mathcal{R}\colon S({\succ})>0} A({\succ})|\geq \ell$. Finally, uPJR arises from PJR by additionally requiring that $S$ only assigns positive weight to a single ranking. 

In addition to uPJR, we will consider two more proportionality axioms in this paper. The first one, unanimous justified representation (uJR), is a weakening of uPJR which requires that each ranking with a weight of at least $1/{m\choose 2}$ should get a non-zero utility. We observe that this axiom can be seen as the counterpart to justified representation (JR), another well-known fairness condition in approval-based committee voting \citep{ABC+16a}.\footnote{
uJR is equivalent to JR in rank aggregation. In particular, the latter axiom requires that $|A(\rhd)\cap \bigcup_{{\succ}\in\mathcal{R}\colon S({\succ})>0}A({\succ})|\geq 1$  for all subprofiles $S$ with $|S|\geq1/{m\choose 2}$ and $|\bigcap_{{\succ}\in\mathcal{R}\colon S({\succ})}A({\succ})|\geq1$. Now, if $S$ assigns positive weight to two different rankings $\succ_1$ and $\succ_2$, there is a pair of candidates $x_1,x_2$ such that $x_1\succ_1 x_2$ and $x_2\succ_2x_1$. Consequently, every output ranking agrees with either $\succ_1$ or $\succ_2$ in at least one pair of candidates. Hence, JR is trivial in rank aggregation unless $S$ assigns a positive weight to a single ranking and it reduces to uJR in this case.} 

\begin{definition}[Unanimous Justified Representation (uJR)]
    A ranking $\rhd$ satisfies \emph{unanimous justified representation (uJR)} for a profile $R$ if $u({\succ},\rhd)\geq 1$ for every ${\succ}\in\mathcal{R}$ with $R({\succ})\geq 1/ {m\choose 2}$. 
\end{definition}

Secondly, we will study a strengthening of uPJR which applies to arbitrary groups of rankings. Specifically, strong proportional justified representation (sPJR) requires that for every subprofile $S$ of size $\alpha$, the output ranking chooses at least $\lfloor\alpha\cdot{m\choose 2}\rfloor$ pairwise comparisons from the union of the rankings in $S$. We note that sPJR does not impose any cohesiveness condition on $S$ and is thus a more demanding axiom than~PJR.

\begin{definition}[Strong Proportional Justified Representation]
    A ranking $\rhd$ satisfies \emph{strong proportional justified representation (sPJR)} for a profile $R$ if $|A(\rhd)\cap \bigcup_{{\succ}\in\mathcal{R}\colon S({\succ})>0} A({\succ})|\geq \lfloor|S|\cdot {m\choose 2}\rfloor$ for all subprofiles $S$ of $R$. 
\end{definition}

We conclude this section with an example illustrating our proportionality axioms.

\begin{example}[Proportionality axioms]\label{ex:axioms}
\begin{wrapstuff}[r,type=figure,width=3cm]
	\centering
	\begin{tikzpicture}
\node at (0,0) [draw=black!30] (step1) {	
			\votermultiplicity{50\%}{\weakorder{{1},{2},{3},{4},{5}}}
			\votermultiplicity{40\%}{\weakorder{{4},{5},{1},{2},{3}}}
            \votermultiplicity{10\%}{\weakorder{{5},{4},{3},{2},{1}}}
		};

        \draw[line width = 0.2mm]
  ([yshift=-0.5cm,xshift=0.1cm]step1.north west) -- ([yshift=-0.5cm,xshift=-0.1cm]step1.north east);
    \end{tikzpicture}
    \vspace{-7pt}
	\caption{Example for our proportionality axioms.}
	\label{fig:psc_exp}
    \end{wrapstuff}
To illustrate our proportionality axioms, we consider the profile $R$ for $m=5$ candidates shown on the right. This profile consists of the three rankings ${\succ_1}=x_1,x_2,x_3,x_4,x_5$, ${\succ_2}=x_4,x_5,x_1,x_2,x_3$, and ${\succ_3}=x_5,x_4,x_3,x_2,x_1$. Since ${5\choose 2}=10$, $\succ_1$ deserves a utility of at least $5$, $\succ_2$ of at least $4$, and $\succ_3$ of at least $1$. Hence, the ranking ${\rhd_1}=x_1,x_2,x_3,x_4,x_5$ even violates uJR for $R$ because $u(\rhd_1, \succ_3)=0$. By contrast, the ranking ${\rhd_2}=x_1,x_2,x_3,x_5,x_4$ satisfies uJR as every input ranking agrees with at least one pairwise comparison in $\rhd_2$. However, $\rhd_2$ fails uPJR since $u(\rhd_2,\succ_2)=3$ even though $\succ_2$ deserves $4$ pairwise comparisons. Thirdly, the ranking $\rhd_3=x_1,x_3,x_5,x_4,x_2$ satisfies uPJR since $u(\rhd_3,{\succ_1})=6$, $u(\rhd_3,\succ_2)=4$, and $u(\rhd_3,\succ_3)=4$. However, it fails sPJR because the rankings $\succ_1$ and $\succ_2$ make up 90\% of the profile but $|A(\rhd_3)\cap (A({\succ_1})\cup A({\succ_2}))|=8$. Lastly, it can be verified that the ranking $x_1,x_4,x_5,x_2,x_3$ even satifies sPJR. 
\end{example}

\section{The Squared Kemeny Rule and uJR}

As a warm-up, we start by examining uJR. To this end, we first note that traditional SWFs, such as the Kemeny rule or the Borda rule, fail this condition because they are heavily majoritarian. We will thus focus on the Squared Kemeny rule which has been explicitly proposed by \citet{LPW24a} for computing proportional rankings. However, as we will show, this SWF severely fails uJR. We hence present another SWF inspired by Chamberlin-Courant approval voting \citep{ChCo83a,LaSk23a} that satisfies this axiom.

We first consider the Squared Kemeny rule, which chooses the ranking that minimizes the total squared swap distance to the input rankings. Formally, the \emph{Squared Kemeny rule ($\mathtt{SqK}$)} is defined by $\mathtt{SqK}(R)=\arg\min_{\rhd\in\mathcal R} \sum_{{\succ}\in\mathcal{R}} R({\succ}) \mathit{swap}({\succ},\rhd)^2=\arg\min_{\rhd\in\mathcal R} \sum_{{\succ}\in\mathcal{R}} R({\succ}) ({m\choose 2}-u({\succ},\rhd))^2$. 
We note that there can be multiple rankings that minimize this objective, so a full definition of the Squared Kemeny rule requires further tie-breaking. However, the tie-breaking will not matter for our result, so we omit these details. We will next prove that $\mathtt{SqK}$ fails uJR for all $m\geq 5$. Specifically, we will present a family of profiles $R$ such that $R({\succ})=\frac{m}{5}/{m\choose 2}$ for some ranking $\succ$ but $\mathtt{SqK}$ uniquely chooses the inverse ranking of $\succ$. This means that the Squared Kemeny rule does not even approximate uJR: for every $k\in\mathbb{N}$, there is a number of candidates $m$, a profile $R$, and a ranking $\succ$ such that $\succ$ deserves a utility of $k$ in $R$ but obtains a utility of $0$ from the ranking chosen by the Squared Kemeny rule. The full proof of \Cref{prop:SqK} is deferred to \Cref{app:SqK}.

\begin{restatable}{proposition}{SqK}\label{prop:SqK}
    For all $m\geq 5$, there is a profile $R$ and ranking $\succ$ such that $R({\succ})=\frac{m}{5}/{m\choose 2}$ and $u({\succ},\mathtt{SqK}(R))=0$.
\end{restatable}
\begin{proof}[Proof Sketch]
\begin{wrapstuff}[3,r,type=figure,width=3.6cm]
	\centering
    \vspace{-0.05cm}
	\begin{tikzpicture}
\node at (0,0) [draw=black!30] (step1) {	
			\votermultiplicity{10\%}{\weakorder{{1},{2},{3},{4},{5}}}
			\votermultiplicity{30\%}{\weakorder{{1},{5},{4},{3},{2}}}
            \votermultiplicity{30\%}{\weakorder{{2},{5},{4},{3},{1}}}
            \votermultiplicity{30\%}{\weakorder{{3},{4},{5},{2},{1}}}
		};

    \draw[line width = 0.2mm]
  ([yshift=-0.5cm,xshift=0.1cm]step1.north west) -- ([yshift=-0.5cm,xshift=-0.1cm]step1.north east);
    \end{tikzpicture}
    \vspace{-7pt}
	\caption{A profile where $\mathtt{SqK}$ fails uJR.}
	\label{fig:sqk}
    \end{wrapstuff}
    We consider the following four rankings to prove this proposition: $\succ_1$ is given by ${\succ_1}=x_1,x_2,\dots,x_m$, $\succ_2$ by ${\succ_2}=x_1,x_m,\dots,x_2$, $\succ_3$ is an arbitrary ranking that bottom-ranks $x_1$ and agrees with $\lfloor \frac{1}{2}{m-1\choose 2}\rfloor$ pairwise comparisons with $\succ_1$, and $\succ_4$ also bottom-ranks $x_1$ and arranges the candidates $x_2,\dots, x_m$ inversely to $\succ_3$. Further, let $R$ denote the profile given by $R({\succ_1})=\frac{m}{5}/ {m\choose 2}$ and $R({\succ_2})=R({\succ_3})=R({\succ_4})=\frac{1}{3}\cdot (1-\frac{m}{5}/{m\choose 2})$. For instance, when $m=5$, the profile in \Cref{fig:sqk} satisfies this construction. We show that the Squared Kemeny rule picks the ranking $\mathtt{SqK}(R)=x_m,\dots, x_1$ for this profile, thus leaving $\succ_1$ without any representation. While the proof for this claim is tedious, we note three high-level ideas. Firstly, it can be shown that the output ranking must generate a higher utility for $\succ_2$ than for $\succ_1$ because $\succ_1$ and $\succ_2$ as well as $\succ_3$ and $\succ_4$ order the candidates $x_2,\dots, x_m$ inversely to each other and $R({\succ_2})>R({\succ_1})$. Secondly, we show that the closer the output ranking $\rhd$ without $x_1$ is to $x_m,\dots, x_2$, the lower we must rank $x_1$ in $\rhd$. Roughly, the reason for this is that both $\succ_3$ and $\succ_4$ cause a high penalty if $\rhd$ without $x_1$ is close to $x_m,\dots,x_2$, so $x_1$ is used to decrease the swap distance to these rankings. Thirdly, we prove that if $x_1$ is ranked sufficiently low in $\rhd$, then it is optimal to order the candidates $x_2,\dots, x_m$ inverse to $\succ_1$ because the weight of $\succ_2$ is significantly larger than that of $\succ_1$. By combining these ideas, it follows that $\mathtt{SqK}(R)=x_m,\dots,x_1$, thereby proving that $\mathtt{SqK}$ fails uJR.
\end{proof}

In light of \Cref{prop:SqK}, one may think that involved SWFs are required to satisfy uJR. We will next refute this hypothesis by introducing a simple SWF inspired by Chamberlin-Courant approval voting. To this end, we recall that Chamberlin-Courant approval voting is an approval-based committee voting rule which chooses the committee that maximizes the number of voters who approve at least one selected candidate \citep{ChCo83a,LaSk23a}. Put differently, this rule maximizes the number of voters that have a utility of at least $1$. We adapt this idea to rank aggregation by defining the score function $s:\mathbb{N}_0\rightarrow\mathbb{R}$ by $s(x)=1$ if $x>0$ and $s(0)=0$. Then, the \emph{Chamberlin-Courant SWF} chooses a ranking $\rhd$ that maximizes $\sum_{{\succ}\in\mathcal{R}} R({\succ})\cdot s(u({\succ},\rhd))$, with ties broken arbitrarily. We will next show that this SWF satisfies uJR.

\begin{proposition}\label{prop:CC}
The Chamberlin-Courant SWF satisfies uJR.
\end{proposition}
\begin{proof}
    Fix a profile $R$ and let $\rhd$ denote the ranking chosen by the Chamberlin-Courant SWF. Hence, $\rhd$ maximizes $\sum_{{\succ}\in\mathcal{R}} R({\succ})\cdot s(u({\succ},\rhd))$ and therefore also $\sum_{{\succ}\in\mathcal{R}} R({\succ})\cdot (s(u({\succ},\rhd))-1)$. We next observe that it holds for all ${\succ},\rhd'\in\mathcal{R}$ that $s(u({\succ},\rhd'))=0$ if and only if $\succ$ orders the candidates inversely to $\rhd'$. This means that $\sum_{{\succ}\in\mathcal{R}} R({\succ})\cdot (s(u({\succ},\rhd'))-1)=-R(\blacktriangleleft')$ for all $\rhd'\in\mathcal{R}$, where $\blacktriangleleft'$ denotes the inverse ranking to $\rhd'$. Finally, since $\rhd$ maximizes $\sum_{{\succ}\in\mathcal{R}} (s(u({\succ},\rhd))-1)$, it follows that its inverse ranking $\blacktriangleleft$ minimizes $R(\blacktriangleleft)$. Since $\sum_{\succ\in\mathcal{R}} R({\succ})=1$, this implies that $R(\blacktriangleleft)\leq \frac{1}{m!}$. Hence, the only ranking with utility $0$ has a weight of at most $\frac{1}{m!}<1/{m\choose 2}$, so uJR is satisfied.
\end{proof}

\begin{remark}
    In approval-based committee voting, the Chamberlin-Courant rule is NP-hard to compute as it is closely related to the set cover problem \citep{PSZ08a,SFL16a}. By contrast, the Chamberlin-Courant SWF can be computed efficiently as its output ranking is (up to tie-breaking) the reverse of the ranking $\succ$ that minimizes $R(\succ)$. 
\end{remark}

\section{The Proportional Sequential Borda Rule and uPJR}\label{sec:PJR}

We now turn to the more demanding proportionality notion of uPJR and present an SWF that satisfies this axiom, namely the Proportional Sequential Borda rule (\psb). To this end, we first introduce a more structured proportionality axiom called rank-priceability and show that this condition implies uPJR. We then prove that \psb satisfies rank-priceability and thus also uPJR. Finally, we prove that \psb guarantees every subprofile $S$ an average utility of at least $\frac{|S|}{4} {m\choose 2}-\frac{3}{16}$.

We start by introducing rank-priceability, which is inspired by the concept of priceability in participatory budgeting \citep[e.g.,][]{PPS21a,BFL+23a}. In this setting, voters report approval ballots over costly candidates and we need to choose a representative subset of candidates subject to a budget constraint. Now, the idea of priceability is that it should be possible to explain the outcome by a payment scheme from the voters to the chosen candidates. In more detail, in participatory budgeting, a set of candidates $W$ is called priceable if there is a virtual budget $B$ that is uniformly distributed among the voters and a payment scheme $\pi$ from the voters to the candidates in $W$ that satisfies the following conditions: 
\begin{enumerate}[label=(\arabic*), topsep=2pt, itemsep=2pt]
    \item Voters only spend their share of the budget on their approved candidates in $W$.
    \item The total budget spent on each candidate in $W$ is equal to its cost.
    \item The unspent budget of any group of voters is not enough to pay for a commonly approved candidate outside of $W$.
\end{enumerate}

We next aim to transfer this axiom to rank aggregation. To this end, we assume that candidates will be bought sequentially and that the costs of the candidates depend on the considered round. In more detail, in every step, the cost of each candidate will be the maximal utility it can generate for a ranking and the payment willingness of a ranking will be the additional utility it obtains by assigning the considered candidate to the next position in the ranking. To make this more formal, let ${\rhd}=x_1,\dots,x_m$ denote an arbitrary ranking. If we place $x_i$ in the $i$-th position of the output ranking after $x_1,\dots, x_{i-1}$ have been put at positions $1,\dots, i-1$, we generate a utility of $u({\succ}, x_i, \{x_i, \dots,x_m \})$ for every input ranking $\succ$. Thus, no ranking should pay more than $u({\succ}, x_i, \{x_i, \dots,x_m \})$ for $x_i$. Further, since $u({\succ}, x_i, \{x_i, \dots,x_m \})\leq m-i$ for all ${\succ}\in\mathcal{R}$ and $u(\rhd, x_i, \{x_i, \dots,x_m \})= m-i$, we set the cost of $x_i$ to $m-i$. Consequently, the total cost of all candidates is $\sum_{i=1}^m m-i={m\choose 2}$. 
Finally, there is no counterpart to the third condition of priceability in rank aggregation as all candidates need to be ranked. Hence, we will fix the budget to $B={m\choose 2}$ and require that the total unspent budget should be less than $1$. Since this means that we may not be able to cover the costs of all candidates, we use the costs only as upper bounds on how much we can spend on each candidate. Formalizing these ideas results in the following axiom, which we call rank-priceability. 

\begin{definition}[Rank-Priceability]
    A ranking $\rhd=x_1,\dots, x_m$ is \emph{rank-priceable} for a profile $R$ if there is a payment function $\pi:\mathcal{R}\times C \rightarrow \mathbb{R}$ such that 
    \begin{enumerate}[label=(\arabic*)]
        \item $0\leq \pi({\succ}, x_i)\leq u({\succ}, x_i, \{x_{i},\dots, x_m\})$ for all ${\succ}\in\mathcal{R}$ and $x_i\in C$,
        \item $\sum_{x_i\in C} \pi({\succ}, x_i)\leq {m\choose 2}\cdot R({\succ})$ for all ${\succ}\in\mathcal{R}$,
        \item $\sum_{{\succ}\in \mathcal R} \pi({\succ}, x_i)\leq m-i$ for all $i\in \{1,\dots, m\}$, and
        \item $\sum_{{\succ}\in\mathcal{R}}\sum_{x_i\in C} \pi({\succ}, x_i)> {m\choose 2}-1$.
    \end{enumerate}
\end{definition}

As usual, an SWF $f$ is rank-priceable if $f(R)$ satisfies this condition for every profile~$R$. 

We will next show that rank-priceability implies uPJR, which mirrors an analogous insight in participatory budgeting \citep{PeSk20a,BFL+23a}.

\begin{proposition}\label{prop:RPimpliesuPJR}
    If a ranking is rank-priceable for a profile, it also satisfies uPJR.
\end{proposition}
\begin{proof}
    Assume for contradiction that there is a profile $R$ and ranking $\rhd=x_1,\dots, x_m$ such that~$\rhd$ satisfies rank-priceability for $R$ but not uPJR. Since $\rhd$ fails uPJR, there is an input ranking $\succ$ and an integer $\ell\in\mathbb{N}$ such that $R({\succ})\geq \ell/{m\choose 2}$ and $u({\succ},\rhd)<\ell$. Since both $u({\succ},\rhd)$ and $\ell$ are integers, the latter inequality means that $u({\succ},\rhd)\leq \ell-1$. Next, let $\pi$ denote a payment scheme verifying the rank-priceability of $\rhd$. By Condition (1), we have that $\pi({\succ'},x_i)\leq u({\succ'}, x_i, \{x_i,\dots, x_m\})$ for all $x_i\in C$ and ${\succ'}\in\mathcal{R}$. Because $\sum_{i=1}^{m} u({\succ}, x_i, \{x_i,\dots, x_m\})= u({\succ}, \rhd)$, we conclude that 
    \[\sum_{i=1}^{m} \pi({\succ}, x_i)\leq \sum_{i=1}^{m} u({\succ}, x_i, \{x_i,\dots, x_m\})=u({\succ},\rhd)\leq \ell-1\leq R({\succ})\cdot {m\choose 2}-1.\] 
    Further, by Condition (2) of rank-priceability, we have that 
    \[\sum_{\succ'\in\mathcal{R}\setminus \{\succ\}}\sum_{x_i\in C} \pi(\succ',x_i)\leq {m\choose 2} \cdot \sum_{\succ'\in\mathcal{R}\setminus\{\succ\}}  R(\succ')={m\choose 2} \cdot (1-R({\succ})).\] 
    By combining our previous two inequalities, we derive that 
    \begin{align*}
        \sum_{\succ'\in\mathcal{R}}\sum_{x_i\in C} \pi(\succ', x_i)
        \leq {m\choose 2}\cdot R({\succ})-1 + {m\choose 2}\cdot (1-R({\succ}))={m\choose 2}-1.
    \end{align*}
    However, this contradicts Condition (4) of rank-priceability. Hence, our initial assumption is wrong and every rank-priceable ranking also satisfies uPJR.
\end{proof}

Notably, the proof of \Cref{prop:RPimpliesuPJR} does not use the third condition of rank-priceability. Moreover, Condition (4) of this axiom can be weakened to $\sum_{x_i\in C} \pi({\succ},x_i) > {m\choose 2}\cdot R({\succ}) -1$ for all ${\succ}\in\mathcal{R}$. When weakening Condition (4) in this way and omitting Condition (3), rank-priceability is equivalent to uPJR. We nevertheless decided to define rank-priceability based on Conditions (3) and (4) because they give additional guidance for the design of SWFs. We will next clarify this point with an example demonstrating the difference between uPJR and rank-priceability. 

\begin{example}[uPJR does not imply rank-priceability.]\label{ex:RP}
        \begin{wrapstuff}[0,r,type=figure,width=4.1cm]
	\centering
	\begin{tikzpicture}
\node at (0,0) [draw=black!30] (step1) {	
			\votermultiplicity{$\frac{1}{6}$}{\weakorder{{1},{2},{3},{4},{5}}}
			\votermultiplicity{$\frac{1}{6}$}{\weakorder{{1},{3},{2},{4},{5}}}
            \votermultiplicity{$\frac{1}{6}$}{\weakorder{{2},{1},{3},{4},{5}}}
            \votermultiplicity{$\frac{1}{6}$}{\weakorder{{2},{3},{1},{4},{5}}}
            \votermultiplicity{$\frac{1}{6}$}{\weakorder{{3},{1},{2},{4},{5}}}
            \votermultiplicity{$\frac{1}{6}$}{\weakorder{{3},{2},{1},{4},{5}}}
		};

        \draw[line width = 0.2mm]
  ([yshift=-0.65cm,xshift=0.1cm]step1.north west) -- ([yshift=-0.65cm,xshift=-0.1cm]step1.north east);
    \end{tikzpicture}
    \vspace{-7pt}
	\caption{A profile where uPJR does not imply rank-priceability.}
	\label{fig:psc_exp2}
    \end{wrapstuff}
    Consider the profile $R$ shown on the right. For this profile, uPJR requires that the output ranking $\rhd$ agrees in at least $\lfloor \frac{1}{6}\cdot {5\choose 2}\rfloor=1$ pair with every input ranking $\succ_i$. While counterintuitive, this means that the ranking $\rhd=x_4,x_5,x_1,x_2,x_3$ satisfies uPJR as all input rankings put $x_4$ ahead of $x_5$. However, this ranking is not rank-priceable: no ranking is willing to pay for $x_5$, so the input rankings can pay at most $4+2+1=7$ for $x_4$, $x_1$, and $x_2$. Since the total budget of our rankings is ${5\choose 2}=10$, a budget of $3$ is remaining and $\rhd$ is not rank-priceable.  By contrast, it can be checked that every ranking $\rhd$ with $x\rhd x_4\rhd x_5$ for all $x\in \{x_1,x_2,x_3\}$ satisfies rank-priceability.
\end{example}

Finally, we will introduce the Proportional Sequential Borda rule (\psb). On a high level, this rule repeatedly chooses the candidate maximizing the Borda score, updates the weights of the input rankings, and deletes the Borda winner from the profile. To make this more formal, we assume for every step $i\in \{1,\dots, m\}$ of \psb that each ranking $\succ$ has a budget $b_i({\succ})\in\mathbb{R}_{\geq 0}$ and that the set of remaining candidates is $X_i$. In the first round, it holds that $X_1=C$ and $b_1({\succ})=R({\succ})\cdot{m\choose 2}$ for all ${\succ}\in\mathcal{R}$, where $R$ denotes the input profile. Now, in each round $i$, we choose the candidate $x^*$ that maximizes the Borda score (or utilitarian welfare) $U(b_i,x,X_i)=\sum_{{\succ}\in\mathcal{R}} b_i({\succ}) \cdot u({\succ},x,X_i)=\sum_{{\succ}\in\mathcal{R}} b_i({\succ}) \cdot |\{y\in X_i\colon x\succ y\}|$ among all candidates in $X_i$, with ties broken arbitrarily. Next, we place this candidate at the $i$-th position of the output ranking and remove it from the set of available candidates (i.e., $X_{i+1}=X_i\setminus \{x^*\}$). Furthermore, we assume that the cost of the $i$-th candidate is $m-i$ and, if possible, each ranking will pay a share of this cost that is proportional to its contribution to the Borda score. More formally, each ranking $\succ$ will pay either $\frac{(m-i)\cdot u({\succ}, x^*, X_i)\cdot b_i({\succ})}{U(b_i, x^*, X_i)}$ or its remaining budget $b_i({\succ})$ if the proportional contribution exceeds $b_i({\succ})$. Hence, we set $b_{i+1}({\succ})=b_i({\succ})-\min(\frac{(m-i)\cdot u({\succ}, x^*, X_i)\cdot b_i({\succ})}{U(b_i, x^*, X_i)}, b_i({\succ}))$ for each ranking ${\succ}\in\mathcal{R}$. After defining $X_{i+1}$ and the budgets $b_{i+1}({\succ})$, \psb continues with the next round until all candidates are placed in the output ranking. 

Subsequently, we consider an example to illustrate the Proportional Sequential Borda rule.

\begin{example}[The Proportional Sequential Borda rule]\label{ex:psb}
Fix two rankings ${\succ_1}=x_1,x_2,x_3,x_4,x_5$ and ${\succ_2}=x_4,x_5,x_1,x_3,x_2$, and let $R$ be the profile with $R({\succ_1})=0.6$ and $R({\succ_2})=0.4$. For this profile, 
\psb chooses the ranking ${\rhd}=x_1,x_4,x_2,x_5,x_3$, which is witnessed by the following sequence of profiles.
\begin{center}
    \begin{tikzpicture}
\node at (0,0) (step1) {	
			\votermultiplicity{$6$}{\weakorder{{1},{2},{3},{4},{5}}}
			\votermultiplicity{$4$}{\weakorder{{4},{5},{1},{3},{2}}}
		};

\draw[line width = 0.2mm]
  ([yshift=-0.5cm,xshift=0.1cm]step1.north west) -- ([yshift=-0.5cm,xshift=-0.1cm]step1.north east);

\node at (1,0) {$\implies$};

\node at (2,0) (step2) {
    \votermultiplicity{$3$}{\weakorder{{2},{3},{4},{5}}}
	\votermultiplicity{$3$}{\weakorder{{4},{5},{3},{2}}}
};

\draw[line width = 0.2mm]
  ([yshift=-0.5cm,xshift=0.1cm]step2.north west) -- ([yshift=-0.5cm,xshift=-0.1cm]step2.north east);

\node at (3,0) {$\implies$};

\node at (4,0) (step3) {
    \votermultiplicity{$\frac{9}{4}$}{\weakorder{{2},{3},{5}}}
	\votermultiplicity{$\frac{3}{4}$}{\weakorder{{5},{3},{2}}}
};

\draw[line width = 0.2mm]
  ([yshift=-0.65cm,xshift=0.1cm]step3.north west) -- ([yshift=-0.65cm,xshift=-0.1cm]step3.north east);

\node at (5,0) {$\implies$};

\node at (6,0) (step4) {
    \votermultiplicity{$\frac{1}{4}$}{\weakorder{{3},{5}}}
	\votermultiplicity{$\frac{3}{4}$}{\weakorder{{5},{3}}}
};

\draw[line width = 0.2mm]
  ([yshift=-0.65cm,xshift=0.1cm]step4.north west) -- ([yshift=-0.65cm,xshift=-0.1cm]step4.north east);

\node at (7,0) {$\implies$};

\node at (8,0) (step5) {
    \votermultiplicity{$\frac{1}{4}$}{\weakorder{{3}}}
	\votermultiplicity{$0$}{\weakorder{{3}}}
};

\draw[line width = 0.2mm]
  ([yshift=-0.65cm,xshift=0.1cm]step5.north west) -- ([yshift=-0.65cm,xshift=-0.1cm]step5.north east);
    \end{tikzpicture}\end{center}
On the left of this figure, we show the initial profile $R$, where the rankings are weighted by their budgets $b_1({\succ_1})=R({\succ_1})\cdot{5\choose 2}=6$ and $b_1({\succ_2})=R({\succ_2})\cdot{5\choose 2}=4$. In this profile, $x_1$ maximizes the Borda score as $U(b_1,x_1,\{x_1,\dots,x_5\})=6\cdot 4+4\cdot 2=32$. Consequently, $\succ_1$ pays $\frac{4}{32}\cdot 6\cdot 4=3$ and $\succ_2$ pays $\frac{4}{32}\cdot 4\cdot 2=1$, which means that the new budgets are $b_2({\succ_1})=b_2({\succ_2})=3$. We moreover remove $x_1$ from the profile as $X_2=C\setminus \{x_1\}$. In the second step, $x_4$ maximizes the total Borda score with $U(b_2,x_4, \{x_2,\dots, x_5\})=12$, so $\succ_1$ pays $\frac{3}{12}\cdot 3\cdot 1 =\frac{3}{4}$ and $\succ_2$ pays $\frac{3}{12}\cdot 3\cdot 3 =\frac{9}{4}$. Hence, the new budgets are given by $b_3({\succ_1})=\frac{9}{4}$ and $b_3({\succ_2})=\frac{3}{4}$ and $x_4$ is removed. In the third step, $x_2$ maximizes the total Borda score with $U(b_3,x_2, \{x_2,x_3,x_5\})=\frac{9}{2}$ and $\succ_1$ pays the total cost of $2$ since $u(\succ_2,x_2,\{x_2,x_3,x_5\})=0$. Hence, the budgets in the fourth step are $b_4({\succ_1})=\frac{1}{4}$ and $b_4({\succ_2})=\frac{3}{4}$. \psb now picks $x_5$ and $\succ_2$ will pay its remaining budget of $\frac{3}{4}$. Finally, in the last round, $x_3$ will be picked and no ranking pays as no additional utility is generated.
\end{example}

We note that the total leftover budget in Example 2 is only $\frac{1}{4}$, which implies that \psb is rank-priceable for the given profile. We will next show that this holds in general, i.e., the Proportional Sequential Borda rule satisfies rank-priceability and therefore also uPJR.

\begin{theorem}\label{thm:propBordaRP}
    The Proportional Sequential Borda rule satisfies rank-priceability. 
\end{theorem}
\begin{proof}
    Fix a profile $R$ and let $\rhd=x_1,\dots, x_m$ be the ranking chosen by \psb. For each $i\in \{1,\dots, m\}$, we denote by $X_i=\{x_i,\dots,x_m\}$ the remaining candidates in the $i$-th round of \psb and by $b_i({\succ})$ the remaining budget of the ranking $\succ$. We will show that the payment scheme $\pi$ defined by
    $\pi({\succ}, x_i)=b_i({\succ})-b_{i+1}({\succ})=\min\left(\frac{(m-i)\cdot u({\succ}, x_i, X_i)\cdot b_i({\succ})}{U(b_i, x_i, X_i)}, b_i({\succ})\right)$
    for all $i\in \{1,\dots, m-1\}$ and $\pi({\succ},x_m)=0$ satisfies the conditions of rank-priceability.\medskip
    
    \emph{Condition (1):} Fix a step $i\in \{1,\dots, m\}$ and a ranking $\succ$. We need to show that $0\leq \pi({\succ},x_i)\leq u({\succ},x_i, X_i)$. We first note that the definition of $\pi$ immediately implies the lower bound, and that both bounds hold if $i=m$ or $b_i({\succ})=0$ because $\pi({\succ},x_i)=0$ in these cases. Hence, suppose that $i<m$ and $b_i({\succ})>0$. In this situation, we have that $U(b_i, x_i, X_i)\geq (m-i)b_i({\succ})$ because $x_i$ maximizes the Borda score in the $i$-th step and $(m-i)b_i({\succ})$ is a lower bound for the Borda score of the top-ranked candidate of $\succ$ in $X_i$. This implies that 
    \[\pi({\succ}, x_i)\leq \frac{(m-i) u({\succ}, x_i, X_i) b_i({\succ})}{U(b_i, x_i, X_i)}\leq \frac{(m-i) u({\succ}, x_i, X_i) b_i({\succ})}{(m-i)b_i({\succ})}=u({\succ}, x_i, X_i),\] 
    which proves Condition (1).\medskip

    \emph{Condition (2):} We next show that $\sum_{x_i\in C} \pi({\succ},x_i)\leq R({\succ}) {m\choose 2}$ for all ${\succ}\in \mathcal{R}$. For this, we observe that $\sum_{x_i\in C} \pi({\succ},x_i)\!=\!\sum_{i=1}^{m-1} (b_i({\succ})-b_{i+1}({\succ}))\!=\!b_1({\succ})-b_m({\succ})$. Since $b_1({\succ})=R({\succ})\cdot{m\choose 2}$ by definition and $b_m({\succ})\geq 0$ (as a ranking never pays more than its remaining budget), the condition holds.\medskip

    \emph{Condition (3):} For this condition, we need to show that $\sum_{{\succ}\in\mathcal{R}} \pi({\succ},x_i)\leq m-i$ for all $x_i\in C$. This holds for $x_m$ because $\pi({\succ},x_m)=0$ for all ${\succ}\in\mathcal{R}$ and for all $i\in \{1,\dots, m-1\}$ because 
    \[\sum_{{\succ}\in\mathcal{R}} \pi({\succ},x_i)\leq \sum_{{\succ}\in\mathcal{R}} \frac{(m-i)\cdot u({\succ}, x_i, X_i)\cdot b_i({\succ})}{U(b_i, x_i, X_i)}=\frac{m-i}{U(b_i,x_i,X_i)}\cdot U(b_i,x_i,X_i)=m-i.\]

    \emph{Condition (4):} Finally, we will show that $\sum_{{\succ}\in\mathcal{R}}\sum_{x_i\in C} \pi({\succ},x_i)>{m\choose 2}-1$. To this end, we will prove that the total remaining budget after the execution of \psb is at most $\sum_{\succ\in\mathcal{R}} b_m({\succ})\leq0.75$. This implies Condition (4) because $\sum_{{\succ}\in\mathcal{R}}\sum_{x_i\in C} \pi({\succ},x_i)=\sum_{{\succ}\in\mathcal{R}}\sum_{i=1}^{m-1} b_i({\succ})-b_{i+1}({\succ})=\sum_{\succ\in\mathcal{R}} b_1({\succ})-b_m({\succ})\geq {m\choose 2}-0.75$. The 
    final inequality here uses that the total initial budget is $\sum_{\succ\in\mathcal{R}} b_1({\succ})={m\choose 2}$ and our claim that the total remaining budget is at most $\sum_{\succ\in\mathcal{R}} b_m({\succ})\leq 0.75$.
    
    To prove that $\sum_{\succ\in\mathcal{R}} b_m({\succ})\leq 0.75$, we will first show that $\sum_{{\succ}\in\mathcal{R}} b_i({\succ})=\frac{(m-i)(m-i+1)}{2}$ for all $i\leq m-2$. For $i=1$, this is true because $\sum_{{\succ}\in\mathcal{R}} b_1({\succ})={m\choose 2}\sum_{{\succ}\in\mathcal{R}} R({\succ}) = \frac{(m-1)(m)}{2}$ by definition. Next, we inductively assume that $\sum_{{\succ}\in\mathcal{R}} b_i({\succ})=\frac{(m-i)(m-i+1)}{2}$ for some $i\in \{1,\dots,m-3\}$. Since $|X_i|=m-i+1$, it follows for every ranking $\succ$ that $\sum_{x\in X_i} u({\succ}, x, X_i)=\sum_{j=0}^{m-i} j=\frac{(m-i)(m-i+1)}{2}$. Hence, the total Borda score of all candidates in this round is 
    \[\sum_{x\in X_i} \sum_{{\succ}\in\mathcal{R}} b_{i}({\succ}) u({\succ}, x, X_i)=\sum_{{\succ}\in\mathcal{R}} b_{i}({\succ}) \frac{(m-i)(m-i+1)}{2}=\left(\frac{(m-i)(m-i+1)}{2}\right)^2.\] 
    
    Since $x_{i}$ maximizes the Borda score among the candidates in $X_i$, we have that $U(b_i,x_i,X_i)\geq \frac{1}{m-i+1}\sum_{x\in X_i} U(b_i, x, X_i)=\frac{(m-i)^2(m-i+1)}{4}$. This implies that $\frac{(m-i) u({\succ}, x, X_i)b_i({\succ})}{U(b_i, x_i, X_i)}\leq \frac{4u({\succ}, x_i, X_i)b_i({\succ})}{(m-i)(m-i+1)}$ for all ${\succ}\in\mathcal{R}$. Further, it holds that $\frac{4}{m-i+1}\leq 1$ as $m-i\geq 3$ and $\frac{u({\succ}, x_i, X_i)}{m-i}\leq 1$ by definition of $u$. We thus infer that $\pi({\succ}, x_i)=\min\left(\frac{(m-i) u({\succ}, x, X_i) b_i({\succ})}{U(b_i, x_i, X_i)}, b_i({\succ})\right)=\frac{(m-i) u({\succ}, x, X_i) b_i({\succ})}{U(b_i, x_i, X_i)}$ for all ${\succ}\in\mathcal{R}$. Finally, because $\sum_{{\succ}\in\mathcal{R}}\frac{(m-i) u({\succ}, x, X_i) b_i({\succ})}{U(b_i, x_i, X_i)}=\frac{m-i}{U(b_i, x_i, X_i)}{U(b_i, x_i,X_i)}=m-i$ and $\sum_{{\succ}\in\mathcal{R}} b_i({\succ})=\frac{(m-i)(m-i+1)}{2}$, we conclude that 
    {
    \begin{align*}
        \sum_{{\succ}\in\mathcal{R}} b_{i+1}({\succ})&=\sum_{{\succ}\in\mathcal{R}} \left(b_{i}({\succ}) - \pi({\succ}, x_i)\right) = \frac{(m-i)(m-i+1)}{2} - (m-i)
        =\frac{(m-i)(m-i-1)}{2}.
    \end{align*}
    }

     This proves the induction step, so it follows for $i=m-2$ that $\sum_{{\succ}\in\mathcal{R}} b_{i}({\succ})=\frac{(m-i)(m-i+1)}{2}=3$. Furthermore, when $i=m-2$, we are left with the candidates $X_{m-2}=\{x_{m-2}, x_{m-1}, x_m\}$ and we know that $x_{m-2}$ maximizes the Borda score among these candidates with respect to $b_{m-2}$. Analogous to before, we hence infer that $U(b_{m-2}, x_{m-2}, X_{m-2})\geq \frac{1}{3}\sum_{x\in X_{m-2}}\sum_{{\succ}\in\mathcal{R}} b_{m-2}({\succ}) u({\succ}, x_{m-2}, X_{m-2}) = 3$. Since $u({\succ},x_{m-2}, X_{m-2})\leq 2$ for all ${\succ}\in\mathcal{R}$, this implies that 
     $\frac{3}{4}\cdot \frac{2 u({\succ}, x_{m-2}, X_{m-2}) b_{m-2}({\succ})}{U(b_{m-2}, x_{m-2}, X_{m-2})}\leq \frac{3}{4}\cdot \frac{4 b_{m-2}({\succ})}{3}=b_{m-2}({\succ})$ and thus
     $\frac{3}{4}\cdot \frac{2 u({\succ}, x_{m-2}, X_{m-2}) b_{m-2}({\succ})}{U(b_{m-2}, x_{m-2}, X_{m-2})}\leq \min(\frac{2 u({\succ}, x_{m-2}, X_{m-2}) b_{m-2}({\succ})}{U(b_{m-2}, x_{m-2}, X_{m-2})}, b_{m-2}({\succ}))=\pi({\succ},x_{m-2})$ for all rankings $\succ$. In turn, it follows that $\sum_{{\succ}\in\mathcal{R}} \pi({\succ}, x_{m-2})\geq \frac{3}{4}\sum_{{\succ}\in\mathcal{R}} \frac{2 u({\succ}, x_{m-2}, X_{m-2}) b_{m-2}({\succ})}{U(b_{m-2}, x_{m-2}, X_{m-2})} = \frac{3}{2}$ since $\sum_{{\succ}\in\mathcal{R}} \frac{2 u({\succ}, x_{m-2}, X_{m-2}) b_{m-2}({\succ})}{U(b_{m-2}, x_{m-2}, X_{m-2})} =2$. This means that the total budget in the $(m-1)$-th round is at most $\sum_{{\succ}\in\mathcal{R}} b_{m-1}({\succ})=\sum_{{\succ}\in\mathcal{R}} (b_{m-2} - \pi({\succ},x_{m-2}))\leq 3-\frac{3}{2}=\frac{3}{2}$.

    Finally, in round $i=m-1$, we are left with the candidates $X_{m-1}=\{x_{m-1}, x_m\}$, which means that $u({\succ},x_{m-1}, X_{m-1})\in \{0,1\}$ for all rankings $\succ$. If $U(b_{m-1},x_{m-1}, X_{m-1})\geq1$, it holds for all input rankings $\succ$ that $\pi({\succ},x_{m-1})=\min(\frac{u({\succ}, x_{m-1}, X_{m-1}) b_{m-1}({\succ})}{U(b_{m-1}, x_{m-1}, X_{m-1})}, b_{m-1}({\succ}))=\frac{u({\succ}, x_{m-1}, X_{m-1}) b_{m-1}({\succ})}{U(b_{m-1}, x_{m-1}, X_{m-1})}$. Consequently, we decrease the total budget by $1$ and the total remaining budget is at most $0.5$ in this case.
    On the other hand, if $U(b_{m-1},x_{m-1}, X_{m-1})<1$, each ranking with $u({\succ}, x_{m-1}, X_{m-1})=1$ will contribute its remaining budget. Since $x_{m-1}$ maximizes the Borda score when only $x_{m-1}$ and $x_m$ are present, we derive that $\sum_{{\succ}\in\mathcal{R}\colon x_{m-1}\succ x_m} b_{m-1}({\succ})\geq \sum_{{\succ}\in\mathcal{R}\colon x_{m}\succ x_{m-1}} b_{m-1}({\succ})$. Thus, we reduce the total remaining budget by at least half, so $
    \sum_{\succ\in\mathcal{R}} b_m({\succ})\leq \frac{1}{2}\sum_{\succ\in\mathcal{R}} b_{m-1}({\succ})\leq 0.75$. This shows for both cases that the total remaining budget is at most $\sum_{\succ\in\mathcal{R}} b_m({\succ})\leq 0.75$. 
    \end{proof}

While \Cref{thm:propBordaRP} shows that \psb satisfies a demanding proportionality axiom for individual rankings, neither uPJR nor rank-priceability give any guarantees for groups of rankings. For example, this means that \psb may fail to properly represent a subprofile that consists of many similar rankings with small weights, even if their total weight is substantial. To show that this is not the case, we will follow the works of \citet{SkGo22a} and \citet{LPW24a} and analyze the average utility that \psb guarantees to subprofiles, as a function of the size of the subprofile. Intuitively, a proportional SWF should guarantee to each subprofile $S$ a fraction of the total utility ${m\choose 2}$ that is at least linear in $|S|$. We show that \psb meets this condition as it guarantees to every subprofile $S$ an average utility of at least $ \frac{|S|}{4}-\frac{3}{16}$.
Due to space constraints, we defer the proof of this theorem to \Cref{app:propBordaAR} and only present a (detailed) proof sketch here. 

\begin{restatable}{theorem}{propBordaAR}\label{thm:propBordaAR}
    Let $R$ be a profile on $m$ candidates and $\rhd=\text{\psb}(R)$ be the ranking chosen by the Proportional Sequential Borda rule. It holds for every subprofile $S$ of $R$ that 
    \[\frac{1}{|S|}\sum_{{\succ}\in \mathcal{R}} S({\succ}) u({\succ},\rhd)\geq {m
    \choose 2}\cdot \frac{|S|}{4}-\frac{3}{16}.\]
\end{restatable}
\begin{proof}[Proof Sketch] 
Fix a profile $R$, let $\rhd=x_1,\dots,x_m$ denote the ranking chosen by \psb, and let $S$ denote an arbitrary subprofile of $R$. As before, we denote by $X_i=\{x_i,\dots,x_m\}$ the set of available candidates in the $i$-th round of $\psb$ and by $b_i({\succ})$ the remaining budget of the ranking $\succ$. Moreover, we define by $b_i^S({\succ})=\frac{S({\succ})}{R({\succ})} b_i({\succ})$ the budget of $\succ$ in the $i$-th round that is due to $S$ (where we assume for simplicity that $\frac{0}{0}=0$). Lastly, we let $c_i^S({\succ})=b_i^S({\succ})-b_{i+1}^S({\succ})=\frac{S({\succ})}{R({\succ})}(b_i({\succ})-b_{i+1}({\succ}))$ be the payment of $\succ$ in the $i$-th round that is due to $S$. 

Now, fix a round $i\in \{1,\dots, m-1\}$ and a ranking $\succ$ with $R({\succ})>0$. By the definition of \psb, it holds that $b_i({\succ})-b_{i+1}({\succ})=\min\left(\frac{(m-i)b_i({\succ})u({\succ},x_i,X_i)}{U(b_i,x_i,X_i)}, b_i({\succ})\right)\leq \frac{(m-i)b_i({\succ})u({\succ},x_i,X_i)}{U(b_i,x_i,X_i)}$. By rearranging this inequality and multiplying both sides with $\frac{S({\succ})}{R({\succ})}$, we derive that 
\[c_i^S({\succ})\cdot \frac{U(b_i,x_i,X_i)}{m-i}=\frac{S({\succ})}{R({\succ})}\cdot (b_i({\succ})-b_{i+1}({\succ}))\cdot \frac{U(b_i,x_i,X_i)}{m-i}\leq \frac{S({\succ})}{R({\succ})}\cdot b_i({\succ})
\cdot u({\succ},x_i,X_i).\]

Further, it holds that $b_1({\succ})\geq b_i({\succ})$ for all $i\in \{1,\dots, m-1\}$ because our budgets are non-increasing. It thus follows that $c_i^S({\succ})\cdot \frac{U(b_i,x_i,X_i)}{m-i}\leq {m\choose 2} \cdot S({\succ})\cdot u({\succ},x_i,X_i)$ since $b_1({\succ})=R({\succ})\cdot {m\choose 2}$. Next, analogous to the proof of \Cref{thm:propBordaRP}, it can be shown that $U(b_i, x_i,X_i)\geq \frac{(m-i)^2(m-i+1)}{4}$, which implies that $c_i^S({\succ})\cdot \frac{(m-i)(m-i+1)}{4}\leq c_i^S({\succ})\cdot\frac{U(b_i,x_i,X_i)}{m-i}\leq {m\choose 2} \cdot S({\succ})\cdot u({\succ},x_i,X_i)$. Lastly, let $C_i^S=\sum_{{\succ}\in\mathcal{R}\colon S({\succ})>0} c_i^S({\succ})$ be the total payment of $S$ in the $i$-th step and let $C^S=\sum_{i=1}^{m-1} C^S_i$. By summing over all rankings ${\succ}\in\mathcal{R}$ and all rounds $i\in \{1,\dots, m-1\}$, we get that 
\[\sum_{i=1}^{m-1} C_i^S \frac{(m-i)(m-i+1)}{4}\leq \sum_{{\succ}\in\mathcal{R}}\sum_{i=1}^{m-1} {m\choose 2} S({\succ}) u({\succ},x_i,X_i)={m\choose 2}\sum_{{\succ}\in\mathcal{R}} S({\succ}) u({\succ},\rhd).\]

In the next step, we aim to find a lower bound for the left-hand side that only depends on $C^S$. To this end, we observe that the function $\frac{(m-i)(m-i+1)}{4}$ is decreasing as $i$ increases, so we minimize our left-hand side if we assume that $S$ only pays in late rounds. On the other hand, we know that $C_i^S\leq (m-i)$ because the total budget decrease in the $i$-th round is at most $(m-i)$. Hence, let $k$ denote the largest integer such that $\frac{k(k+1)}{2}=\sum_{i=1}^{k} i\leq C^S$ and let $\ell=C^S-\frac{k(k+1)}{2}$. This means we minimize $\sum_{i=1}^{m-1} C_i^S \frac{(m-i)(m-i+1)}{4}$ when setting $C_{i}^S=m-i$ for all $i\in \{m-k,\dots,m-1\}$, $C_{m-k-1}^S=\ell$, and $C_i^S=0$ for all $i\in \{1,\dots,m-k-2\}$. Consequently, it holds that
\[\sum_{i=1}^{m-1} C_i^S\cdot \frac{(m-i)(m-i+1)}{4}\geq \frac{\ell(m\mspace{-1.5mu}-\mspace{-1.5mu}(m\mspace{-1.5mu}-\mspace{-1.5mu}k\mspace{-1.5mu}-\mspace{-1.5mu}1))(m\mspace{-1.5mu}-\mspace{-1.5mu}(m\mspace{-1.5mu}-\mspace{-1.5mu}k\mspace{-1.5mu}-\mspace{-1.5mu}1)\mspace{-1.5mu}+\mspace{-1.5mu}1)}{4}\,+\!\sum_{i={m-k}}^{m-1} \frac{(m\mspace{-1.5mu}-\mspace{-1.5mu}i)^2(m\mspace{-1.5mu}-\mspace{-1.5mu}i\mspace{-1.5mu}+\mspace{-1.5mu}1)}{4}.\]

By changing the order of summation, this equivalently means that $\sum_{i=1}^{m-1} C_i^S\cdot \frac{(m-i)(m-i+1)}{4}\geq \frac{\ell(k+1)(k+2)}{4}+\sum_{i=1}^k\frac{i^2(i+1)}{4}$. Using the fact that $\sum_{i=1}^k {i^2(i+1)}= \frac{k^4}{4}+\frac{5k^3}{6}+\frac{3k^2}{4}+\frac{k}{6}$, we then show in a series of mathematical transformations that $\frac{\ell(k+1)(k+2)}{4}+\sum_{i=1}^k\frac{i^2(i+1)}{4}\geq\frac{C^S(C^S+1)}{4}$. We hence conclude that $\frac{C^S(C^S+1)}{4}\leq {m\choose 2}\sum_{{\succ}\in\mathcal{R}} S({\succ})\cdot u({\succ},\rhd)$.
Next, we have shown in the proof of \Cref{thm:propBordaRP} that the total left over budget of \psb is at most $\frac{3}{4}$, which means that $C^S\geq |S|\cdot{m\choose 2}-\frac{3}{4}$. By substituting this into our previous inequality, we derive that 
\[\frac{1}{4}\cdot(|S|{m\choose 2}-\frac{3}{4})\cdot|S|{m\choose 2}\leq \frac{1}{4}\cdot(|S|{m\choose 2}-\frac{3}{4})\cdot(|S|{m\choose 2}+\frac{1}{4})\leq {m\choose 2}\sum_{{\succ}\in\mathcal{R}} S({\succ})u({\succ},\rhd).\] 

Finally, our theorem follows by dividing both sides by $|S|{m\choose 2}$.
\end{proof}

\begin{remark}
    \citet{LPW24a} have computed a bound analogous to \Cref{thm:propBordaAR} for the Squared Kemeny rule, which is, however, sub-linear and trivial if $|S|\leq \frac{1}{4}$. Hence, the Proportional Sequential Borda rule gives stronger proportionality guarantees both in an axiomatic and a quantitative sense. Further, at least when $|S|=1$, our bound for \psb is roughly within a factor $2$ of the optimum. To see this, consider the profile where the ranking $x_1,\dots, x_m$ and its inverse $x_m\dots x_1$ both have weight $50\%$. In this case, every output ranking achieves an average utility of $\frac{1}{2}{m\choose 2}$ for the (sub)profile $S=R$, so no SWF can guarantee an average utility of more than $\frac{1}{2}{m\choose 2}$ to a (sub)profile $S$ with $|S|=1$. In comparison, \Cref{thm:propBordaAR} shows that \psb guarantees an average utility of at least $\frac{1}{4}{m\choose 2}-\frac{3}{16}$ for every profile~$R$ and the subprofile $S$ with $|S|=|R|=1$.
\end{remark}

\begin{remark}
    While \psb computes the winning ranking top-down, one can design an analogous SWF that computes the winning ranking bottom-up when redefining the utility $u({\succ}, x, X)$ by $|\{y\in X\setminus \{x\}\colon y\succ x\}|$. That is, we now give a score of $|X|-1$ to the least-preferred candidate in $X$ and of $0$ to the most-preferred candidate in this set. In each step, we then determine the candidate with maximal score, put it on the lowest available position of the output ranking, and update the budgets as for \psb. While this rule seems less intuitive, it satisfies the same proportionality guarantees than \psb. An analogous observation will apply to all SWFs in this paper because, intuitively, we generate the same utility for a ranking by placing its favorite candidate at the $k$-th position or its least favorite candidate at the $(m+1-k)$-th position. 
\end{remark}

\begin{remark}
    The Proportional Sequential Borda rule resembles the Method of Equal Shares (\mes) \citep{PeSk20a,PPS21a}, a prominent rule in participatory budgeting. In this setting, a finite set of voters $N$ reports approval ballots $A_i$ over a set of candidates $C$. Moreover, every candidate $x$ comes with a cost $c(x)$ and we need to choose a representative subset of the candidates whose total cost does not exceed a predefined budget $B$. When formalizing this setting via vote shares, a profile $R$ maps every approval ballot $A$ to the fraction of voters $R(A)$ that report $A$. Now, \mes works as follows: first, each approval ballot $A$ (or group of voters that report $A$) is assigned a budget of $b_1(A)=R(A)\cdot B$. Then, \mes iteratively picks the candidate $x^*$ that minimizes the value $\rho$ such that $\sum_{A} \min(\rho\cdot b_1(A)\cdot u(x,A), b_i(A))=c(x)$, where $u(x,A)=1$ if $x\in A$ and $0$ otherwise, removes $x^*$ from the profile, and updates the budgets to $b_{i+1}(A)=b_i(A)-\min(\rho\cdot b_1(A)\cdot u(x^*,A), b_i(A))$. This process is repeated until no candidate can be afforded anymore. 
    
    Interestingly, it can be shown that \psb picks in all rounds $i\in \{1,\dots, m-3\}$ the candidate $x_i$ that minimizes the value $\rho$ such that $\sum_{{\succ}\in\mathcal{R}} \min(\rho\cdot b_i({\succ})\cdot u({\succ},x_i,X_i), b_i({\succ}))=m-i$. Put differently, \psb optimizes the same objective as \mes, except that it uses the current budget $b_i({\succ})$ in the first argument instead of the initial budget $b_1({\succ})$. Intuitively, this alternative definition of \psb works because the Borda winner optimizes the cost per utility ratio for the current budgets. We refer to \Cref{app:RMES} for more details on the relation between \mes and \psb. Moreover, in this appendix, we also present another SWF called the Ranked Method of Equal Shares (\rmes), which is directly inspired by the \mes and satisfies rank-priceability while being more utilitarian than~\psb.
    \end{remark}

\section{The Flow-adjusting Borda Rule and sPJR}\label{subsec:FB}

As our last contribution, we turn to sPJR, our strongest proportionality axiom. In particular, we first show that \psb fails this condition, which necessitates the design of another SWF. To this end, we introduce a strengthening of rank-priceability called pair-priceability, and show that this condition implies sPJR. Based on this insight, we design the Flow-adjusting Borda rule (\fb) and prove that it is pair-priceable. Moreover, we also show that \fb achieves the same guarantee as \psb for the average utility of an arbitrary subprofile. All proofs in this section, except that of \Cref{ex:sPJRcounter}, are deferred to \Cref{app:FB}, and we provide proof sketches instead.

As the first contribution of this section, we show that \psb fails sPJR. 

\begin{proposition}\label{ex:sPJRcounter}
    \psb fails sPJR.
\end{proposition}
\begin{proof}
    Let $C=\{y,x_1,\dots,x_4,z_1,\dots, z_{21}\}$ and consider the following $8$ rankings: 
    \begin{align*}
        &{\succ_1}=x_1,x_2,x_3,x_4,y,z_1,\dots,z_{21}\qquad {\succ_2}=x_2,x_3,x_4,x_1,y,z_1,\dots,z_{21}\qquad \\
        &{\succ_3}=x_3,x_4,x_1,x_2,y,z_1,\dots,z_{21}\qquad 
        {\succ_4}=x_4,x_1,x_2,x_3,y,z_1,\dots,z_{21}\qquad \\
        &{\succ_5}=y,z_{21},\dots,z_{1},x_4,x_3,x_2,x_1\qquad 
        {\succ_6}=y,z_{21},\dots,z_{1},x_3,x_2,x_1,x_4\qquad \\
        &{\succ_7}=y,z_{21},\dots,z_{1},x_2,x_1,x_4,x_3\qquad 
        {\succ_8}=y,z_{21},\dots,z_{1},x_1,x_4,x_3,x_2\qquad         
    \end{align*}
    
    Less formally, our rankings can be partitioned into $2$ groups: the rankings $\succ_1,\dots, \succ_4$ rank all $x_i$ ahead of $y$ ahead of all $z_j$, order the candidates $z_j$ in increasing order of their indices, and the candidates $x_i$ are arranged cyclic within these rankings. Conversely, the rankings $\succ_5,\dots,\succ_8$ rank $y$ ahead of all $z_j$ ahead of all $x_i$, sort the candidates $z_j$ in decreasing order of their indices, and the candidates $x_i$ are also arranged in a cycle within these rankings.
    Further, since there are $m=26$ candidates, there are ${26\choose 2}=325$ pairwise comparisons in each ranking. Next, we define the profile~$R$ by $R(\succ_i)=\frac{293}{4}\cdot \frac{1}{325}$ for all $i\in \{1,\dots, 4\}$ and $R(\succ_i)=\frac{32}{4}\cdot \frac{1}{325}$ for all $i\in \{5,\dots, 8\}$. Now, sPJR requires that the output ranking $\rhd$ chooses $\frac{293}{4}\cdot 4=293$ pairwise comparisons from the union of $\succ_1,\dots, \succ_4$. However, we have shown with the help of a computer that, up to reorderings of $x_1,\dots,x_4$, \psb uniquely chooses the following ranking: 
    \begin{align*}
        &\text{\psb}(R)=\rhd_{\psb}=y,x_1,x_2,x_3,x_4,z_1,\dots,z_{10},z_{21},z_{11},z_{20},z_{12},z_{19},z_{13},z_{14},z_{18},z_{17},z_{15},z_{16}
    \end{align*}

    It can be verified that $\rhd_{\psb}$ only agrees with $292$ pairwise comparisons in the union of $\succ_1,\dots,\succ_4$, regardless of the order of $x_1,\dots,x_4$. Hence, \psb fails sPJR for the given profile.
    While the full computation of \psb is tedious, we note that the central "mistake" already happens in the first round. In this round, $y$ is chosen as it maximizes the Borda score with $U(b_1,y,C)=293\cdot21+32\cdot25=6953$ (where $b_1({\succ})=R({\succ})\cdot 325$). Therefore, each ranking $\succ_i$ with $i\in \{1,\dots, 4\}$ pays $\frac{25}{6953}\cdot 21\cdot \frac{293}{4}\approx 5.531$ and each ranking $\succ_i$ with $i\in \{5,\dots, 8\}$ pays $\frac{25}{6953}\cdot 25\cdot \frac{32}{4}\approx 0.719$. However, this means that the rankings $\succ_1,\dots,\succ_4$ pay in total $\approx 22.124$, even though they only agree with $21$ pairwise comparisons when placing $y$ first. Put differently, as a group, these rankings pay more than their obtained utility, so they cannot afford enough additional utility during the further execution of \psb.
\end{proof}

We note that the proof of \Cref{ex:sPJRcounter} also showcases a potential flaw in the definition of rank-priceability: this axiom only precludes that individual rankings spend more on a candidate than the utility they obtain, but this guarantee does not extend to groups. To design SWFs that satisfy sPJR, we will therefore present a refined version of rank-priceability called pair-priceability. The idea of this axiom is to view the output ranking $\rhd$ as the set of pairs $A(\rhd)=\{(x,y)\in C^2\colon x\rhd y\}$ and that every pair of this set needs to be bought for a price of $1$ by the input rankings. 

\begin{definition}[Pair-Priceability]
    A ranking $\rhd=x_1,\dots, x_m$ is \emph{pair-priceable} for a profile $R$ if there is a payment function $\pi:\mathcal{R}\times A(\rhd)\rightarrow [0,1]$ such that 
    \begin{enumerate}[label=(\arabic*)]
        \item $\pi({\succ}, (x_i, x_j))\leq u({\succ}, x_i, \{x_i,x_j\})$ for all ${\succ}\in\mathcal{R}$ and $(x_i, x_j)\in A(\rhd)$,
        \item $\sum_{(x_i, x_j)\in A(\rhd)} \pi({\succ}, (x_i, x_j))\leq {m\choose 2}\cdot R({\succ})$ for all ${\succ}\in\mathcal{R}$,
        \item $\sum_{{\succ}\in \mathcal R} \pi({\succ}, (x_i, x_j))\leq 1$ for all $(x_i, x_j)\in A(\rhd)$, and
        \item $\sum_{{\succ}\in\mathcal{R}}\sum_{(x_i,x_j)\in A(\rhd)} \pi({\succ}, (x_i,x_j))> {m\choose 2}-1$. 
    \end{enumerate}
\end{definition}

Pair-priceability differs from rank-priceability only in that the conditions are formulated for pairs of candidates instead of the ranks of candidates. For instance, Condition (1) now states that a ranking $\succ$ is only allowed to pay for a pair of candidates $(x_i,x_j)$ if $x_i\succ x_j$. Hence, pair-priceability requires a more fine-grained payment scheme than rank-priceability. Further, we note that rank-priceability rules out the problem observed in the proof of \Cref{ex:sPJRcounter}: the rankings $\succ_1,\dots,\succ_4$ in this example can pay at most $21$ for the pairwise comparisons including $y$ because they all rank $y$ only ahead of $z_1,\dots, z_{21}$. More generally, we will next show that pair-priceability implies sPJR. 

\begin{restatable}{proposition}{pairpriceable}\label{prop:sPJR}
If a ranking is pair-priceable for a profile, it also satisfies sPJR.  
\end{restatable}
\begin{proof}[Proof Sketch]
    The proof of this proposition works similar to that of \Cref{prop:RPimpliesuPJR}. Assume for contradiction that there is a profile $R$ and a ranking $\rhd$ such that $\rhd$ satisfies pair-priceability for $R$ but fails sPJR. By the latter condition, there is a subprofile $S$ of $R$ and an integer $\ell$ such that $|S|\geq \ell/{m\choose 2}$ but $|A(\rhd)\cap\bigcup_{\succ\in\mathcal{R}\colon S({\succ})>0} A({\succ})|<\ell$. Since $\ell$ and $|A(\rhd)\cap\bigcup_{\succ\in\mathcal{R}\colon S({\succ})>0} A({\succ})|$ are integers, we infer that $|A(\rhd)\cap\bigcup_{\succ\in\mathcal{R}\colon S({\succ})>0} A({\succ})|\leq\ell-1$. Intuitively, this means that, for any payment scheme $\pi$ that satisfies the conditions of pair-priceability, the rankings with positive weights in $S$ can pay at most $\ell-1$ for the pairs in $|A(\rhd)\cap\bigcup_{\succ\in\mathcal{R}\colon S({\succ})>0} A({\succ})|$. However, these rankings control a budget of at least $|S|\cdot {m\choose 2}\geq \ell$, so a budget of at least $1$ remains unspent by these rankings. This contradicts Condition (4) of pair-priceability. 
\end{proof}
    
In light of this proposition, we next aim to design an SWF that satisfies pair-priceability. To this end, we will now discuss the \emph{Flow-adjusting Borda rule}  (\fb), which relies on ideas similar to \psb: in each round, we will add the Borda winner with respect to the current budgets to the output ranking, decrease the budgets of the rankings, and remove the Borda winner from consideration. To make this more formal, we denote again by $b_i({\succ})$ the budget of ranking $\succ$ in the $i$-th round and by $X_i$ the remaining candidates. Just as for \psb, we have that $b_1({\succ})=R({\succ})\cdot {m\choose 2}$ and $X_1=C$, where $R$ is the input profile. For each round $i$, we will then choose a candidate $x^*\in \arg\max_{x\in X_i} U(b_i,x,X_i)$ maximizing the Borda score (with ties broken arbitrarily), place it at the $i$-th position of the output ranking, and set $X_{i+1}=X_i\setminus \{x^*\}$. However, in contrast to \psb, \fb determines the payments of the rankings based on a maximum flow in the following flow network $G_{x^*}=(V,E,c)$.\footnote{
We recall here some basics of the maximum flow problem. A flow network $G=(V,E,c)$ is a capacitated directed graph where $c:E\rightarrow\mathbb{R}_{\geq 0}$ specifies the capacity of every edge and $V$ contains two designated vertices $s$ and $t$ called source and sink. A flow in such a network is a function $f:E\rightarrow \mathbb{R}_{\geq 0}$ such that \emph{(i)} $f(e)\leq c(e)$ for all $e\in E$ (capacity constraint) and \emph{(ii)} $\sum_{(u,v)\in E} f(u,v)=\sum_{(v,w)\in E} f(v,w)$ for all $v\in V\setminus \{s,t\}$ (flow conservation). The value of a flow $f$ is the net outflow of the source $s$, i.e., $\sum_{(s,v)\in E} f(s,v)-\sum_{(v,s)\in E} f(v,s)$. Finally, a maximum flow is a flow with maximum value.
}
\begin{itemize}
\item The set of vertices $V$ consists of the source $s$, a ranking vertex $v_\succ$ for every ranking ${\succ}\in\mathcal{R}$, a candidate vertex $v_y$ for every candidate $y\in X_i\setminus\{x^*\}$, and the sink $t$.
    \item For every ranking $\succ$, there is an edge from the source $s$ to the ranking vertex $v_\succ$ with a capacity equal to the remaining budget of $\succ$, i.e., $c(s,v_{\succ})=b_i({\succ})$.
    \item For every ranking $\succ$ and every candidate $y\in X_i\setminus \{x^*\}$ with $x^*\succ y$, there is an edge from $v_\succ$ to $v_y$ with unbounded capacity.
    \item For every candidate $y$, there is an edge from the vertex $v_y$ to the sink $t$ with capacity $c(v_y, t)\!=\!1$.
\end{itemize}

Now, among all maximum flows in $G_{x^*}$, let $f$ denote one that minimizes the maximum cost per utility ratio of an input ranking, i.e., $\max_{{\succ}\in\mathcal{R}}\frac{f(s,v_{\succ})}{b_i({\succ}) u({\succ},x^*, X_i)}$ (where we assume for simplicity that $\frac{0}{0}=0$). After determining this flow, we set $b_{i+1}({\succ})=b_i({\succ})-f(s,v_{\succ})$ for every ranking $\succ$ and proceed with the next round. Thus, \fb only augments \psb by using a more sophisticated payment scheme. Even more, if possible, the rankings pay in \fb the same as in \psb because $\max_{{\succ}\in\mathcal{R}}\frac{f(s,v_{\succ})}{b_i({\succ}) u({\succ},x^*, X_i)}$ is minimized by if the maximum flow $f$ is proportional to $\frac{b_i(\succ)u(\succ,x^*,X_i)}{U(b_i,x^*,X_i)}$. 

We next illustrate the Flow-adjusting Borda rule with an example.

\begin{example}[The Flow-adjusting Borda rule]\label{ex:flowBorda}

We consider the four rankings ${\succ_1}=x_2,x_3,x_1,x_4,x_5$, ${\succ_2}=x_3,x_2,x_1,x_4,x_5$, ${\succ_3}=x_1,x_4,x_5,x_2,x_3$, and ${\succ_4}=x_1,x_4,x_5,x_3,x_2$, and let $R$ be the profile given by $R({\succ_1})=R({\succ_2})=\frac{7}{20}$ and $R(\succ_3)=R(\succ_4)=\frac{3}{20}$. Assuming that ties are broken in favor of candidates with smaller indices, \fb chooses the ranking $\rhd=x_1,x_2,x_3,x_4,x_5$ for $R$, whereas \psb chooses $\rhd'=x_1,x_2,x_4,x_3,x_5$. The computation for \fb can be verified based on the following profiles.

\begin{center}
    \begin{tikzpicture}
\node at (0,0) (step1) {	
			\votermultiplicity{$\frac{7}{2}$}{\weakorder{{2},{3},{1},{4},{5}}}
			\votermultiplicity{$\frac{7}{2}$}{\weakorder{{3},{2},{1},{4},{5}}}
            \votermultiplicity{$\frac{3}{2}$}{\weakorder{{1},{4},{5},{2},{3}}}
            \votermultiplicity{$\frac{3}{2}$}{\weakorder{{1},{4},{5},{3},{2}}}
		};

\draw[line width = 0.2mm]
  ([yshift=-0.65cm,xshift=0.1cm]step1.north west) -- ([yshift=-0.65cm,xshift=-0.1cm]step1.north east);

\node at (1.5,0) {$\implies$};

\node at (3,0) (step2) {	
			\votermultiplicity{$\frac{5}{2}$}{\weakorder{{2},{3},{4},{5}}}
			\votermultiplicity{$\frac{5}{2}$}{\weakorder{{3},{2},{4},{5}}}
            \votermultiplicity{$\frac{1}{2}$}{\weakorder{{4},{5},{2},{3}}}
            \votermultiplicity{$\frac{1}{2}$}{\weakorder{{4},{5},{3},{2}}}
		};

\draw[line width = 0.2mm]
  ([yshift=-0.65cm,xshift=0.1cm]step2.north west) -- ([yshift=-0.65cm,xshift=-0.1cm]step2.north east);

\node at (4.5,0) {$\implies$};

\node at (5.5,0) (step3) {
    \votermultiplicity{$\frac{55}{26}$}{\weakorder{{3},{4},{5}}}
	\votermultiplicity{$\frac{23}{26}$}{\weakorder{{4},{5},{3}}}
};

\draw[line width = 0.2mm]
  ([yshift=-0.65cm,xshift=0.1cm]step3.north west) -- ([yshift=-0.65cm,xshift=-0.1cm]step3.north east);

\node at (6.5,0) {$\implies$};

\node at (7.5,0) (step4) {
    \votermultiplicity{$\frac{3}{26}$}{\weakorder{{4},{5}}}
	\votermultiplicity{$\frac{13}{26}$}{\weakorder{{4},{5}}}
};

\draw[line width = 0.2mm]
  ([yshift=-0.65cm,xshift=0.1cm]step4.north west) -- ([yshift=-0.65cm,xshift=-0.1cm]step4.north east);

\node at (8.5,0) {$\implies$};

\node at (9.5,0) (step5) {
    \votermultiplicity{$0$}{\weakorder{{5}}}
	\votermultiplicity{$0$}{\weakorder{{5}}}
};

\draw[line width = 0.2mm]
  ([yshift=-0.5cm,xshift=0.1cm]step5.north west) -- ([yshift=-0.5cm,xshift=-0.1cm]step5.north east);
    \end{tikzpicture}\end{center}

    \begin{wrapstuff}[5,r,type=figure,width=6.5cm]
    \centering
    \scalebox{0.8}{
    \begin{tikzpicture}
        \node[draw, shape=circle, minimum width=0.8cm] (s) at (0,0) {$s$};
        
        \node[draw, shape=circle, minimum width=0.8cm] (vp1) at (2.2,2) {\small${\succ_1}$};
        \node[draw, shape=circle, minimum width=0.8cm] (vp2) at (2.2,0.66) {\small${\succ_2}$};
        \node[draw, shape=circle, minimum width=0.8cm] (vp3) at (2.2,-0.66) {\small${\succ_3}$};
        \node[draw, shape=circle, minimum width=0.8cm] (vp4) at (2.2,-2) {\small${\succ_4}$};

        \node[draw, shape=circle, minimum width=0.8cm] (vx1) at (4.4,2) {\small${x_2}$};
        \node[draw, shape=circle, minimum width=0.8cm] (vx2) at (4.4,0.66) {\small${x_3}$};
        \node[draw, shape=circle, minimum width=0.8cm] (vx3) at (4.4,-0.66) {\small${x_4}$};
        \node[draw, shape=circle, minimum width=0.8cm] (vx4) at (4.4,-2) {\small${x_5}$};

        \node[draw, shape=circle, minimum width=0.8cm] (t) at (6.6,0) {$t$};

        \draw[->] (s) --node[above] {3.5} (vp1);
        \draw[->] (s) --node[above] {3.5} (vp2);
        \draw[->] (s) --node[above] {1.5} (vp3);
        \draw[->] (s) --node[above] {1.5} (vp4);

        \draw[->] (vx1) --node[above] {1} (t);
        \draw[->] (vx2) --node[above] {1} (t);
        \draw[->] (vx3) --node[above] {1} (t);
        \draw[->] (vx4) --node[above] {1} (t);

        \draw[->] (vp1) -- (vx3);
        \draw[->] (vp1) -- (vx4);
        \draw[->] (vp2) -- (vx3);
        \draw[->] (vp2) -- (vx4);
        \draw[->] (vp3) -- (vx1);
        \draw[->] (vp3) -- (vx2);
        \draw[->] (vp3) -- (vx3);
        \draw[->] (vp3) -- (vx4);
        \draw[->] (vp4) -- (vx1);
        \draw[->] (vp4) -- (vx2);
        \draw[->] (vp4) -- (vx3);
        \draw[->] (vp4) -- (vx4);
    \end{tikzpicture}
    }
    \vspace{-7pt}
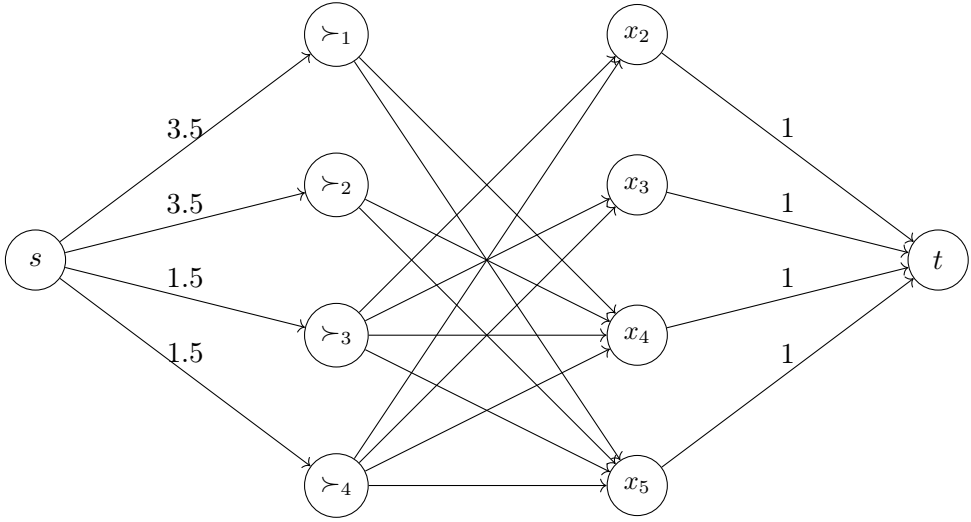
    \caption{The flow network $G_{x_1}$ (restricted to rankings with non-zero budget) used for the first step of \fb for the profile in \Cref{ex:flowBorda}. Edges are labeled by their capacity and edges without label have an unbounded capacity.}
    \label{fig:flowBorda}
\end{wrapstuff}
We show in this figure again the rankings restricted to the available candidates and weighted by their budget in each round. Moreover, to save space, we collapsed in the third step the rankings $\succ_1$ and $\succ_2$ as well as $\succ_3$ and $\succ_4$ into single rankings, which does not affect the computation of \fb. In the first round, it holds for all $x\in \{x_1,x_2,x_3\}$ that $U(b_1,x, X_1)=26$. Assuming that the tie-breaking favors $x_1$ among these three candidates, $x_1$ is chosen first. We thus need to identify a maximum flow in the network $G_{x_1}$, which is shown on the right. In this network, the rankings $\succ_1$ and $\succ_2$ together can pay at most $2$ for $x_4$ and $x_5$ and the rankings $\succ_3$ and $\succ_4$ can pay $1$ each for $x_2$ and $x_3$. Hence, the maximum flow has value $4$ and it can be shown that the cost per utility ratio is minimized if each ranking pays $1$. Consequently, the budgets in the second step are $b_2({\succ_1})=b_2({\succ_2})=\frac{5}{2}$ and $b_2(\succ_3)=b_2(\succ_4)=\frac{1}{2}$. By contrast, in \psb, $\succ_1$ and $\succ_2$ each pay $\frac{4}{26}\cdot 2\cdot\frac{7}{2}=\frac{28}{26}$, which is the main reason for the different outcome. Starting from the second round on, \fb behaves exactly like \psb, because the payments made by \psb can be transformed into a maximum flow of the corresponding network. We hence leave the verification of the remaining steps to the reader. 
\end{example}

We will next show that \fb is pair-priceable and thus satisfies sPJR.
Moreover, we note that the following statement holds regardless of the exact maximum flow chosen in the flow network $G_{x_i}$, i.e., it is not necessary to minimize the cost per utility ratio.  

\begin{restatable}{theorem}{flowBordaPR}\label{thm:flowBordaPR}
The Flow-adjusting Borda rule is pair-priceable. 
\end{restatable}
\begin{proof}[Proof Sketch]
    Fix a profile $R$ and let $\rhd=x_1,\dots,x_m$ denote the ranking chosen by \fb. As usual, we denote by $b_i({\succ})$ the remaining budget of the ranking $\succ$ in the $i$-th round and by $X_i$ the remaining candidates. Further, we let $G_{x_i}$ denote the flow network that \fb uses to determine the budget updates in the $i$-th round and $f_i$ is the chosen maximum flow in this network. This means for all $i\in \{1,\dots, m-1\}$ and rankings ${\succ}\in\mathcal{R}$ that $b_i({\succ})-b_{i+1}({\succ})=f_i(s,v_{\succ})$. Now, we define the payment scheme $\pi$ by $\pi({\succ}, (x_i,x_j))=f_i(v_{\succ}, v_{x_j})$ if $(v_{\succ}, v_{x_j})$ is an edge in $G_{x_i}$ and $\pi({\succ},(x_i,x_j))=0$ otherwise. We claim that $\pi$ satisfies all conditions of pair-priceability. 

    First, our payment scheme satisfies Condition (1) of pair-priceability (i.e., $\pi({\succ},(x_i,x_j))\leq u({\succ},x_i, \{x_i,x_j\})$ for all ${\succ}\in\mathcal{R}$ and $(x_i,x_j)\in A(\rhd)$) by definition of the flow network $G_{x_i}$. Specifically, if $x_j\succ x_i$, then the edge $(v_{\succ},x_{v_j})$ is not in $G_{x_i}$ and $\pi({\succ},(x_i,x_j))=0$. On the other hand, if $x_i\succ x_j$, then $\pi({\succ},(x_i,x_j))=f_i(v_{\succ},v_{x_j})\leq f_i(v_{x_j},t)\leq 1=u({\succ},x_i,\{x_i,x_j\})$ since $f$ is a flow. 

    For Condition (2) (i.e., $\sum_{(x_i, x_j)\in A(\rhd)} \pi({\succ}, (x_i, x_j))\leq {m\choose 2}\cdot R({\succ})$ for all ${\succ}\in\mathcal{R}$), we observe that $f_i(s,v_{\succ})\leq b_i({\succ})$ for all $i\in \{1,\dots,m-1\}$ and ${\succ}\in\mathcal{R}$ because the capacity of the edge $(s,v_{\succ})$ in $G_{x_i}$ is $b_i({\succ})$. This means that the final budget $b_m({\succ})$ of every ranking is non-negative.
    Further, it holds that $b_i({\succ})-b_{i+1}({\succ})=f_i(s,v_{\succ})=\sum_{x_j\in X_i\colon x_i\succ x_j} f_i(v_{\succ},v_{x_j})=\sum_{x_j\in X_i\setminus \{x_i\}} \pi({\succ}, (x_i,x_j))$ by our definitions and flow conservation. Since $b_1({\succ})=R({\succ})\cdot {m\choose 2}$, it hence follows that $\sum_{(x_i, x_j)\in A(\rhd)} \pi({\succ}, (x_i, x_j))=\sum_{i=1}^{m-1} \sum_{x_j\in X_i\setminus \{x_i\}} \pi({\succ}, (x_i,x_j))=\sum_{i=1}^{m-1} b_i({\succ})-b_{i+1}({\succ})=b_1({\succ})-b_m({\succ})\leq {m\choose 2} \cdot R({\succ})$. 

    For Condition (3) (i.e., $\sum_{{\succ}\in \mathcal R} \pi({\succ}, (x_i, x_j))\leq 1$ for all $(x_i, x_j)\in A(\rhd)$), it suffices to note that $\sum_{{\succ}\in \mathcal R} \pi({\succ}, (x_i, x_j))\mspace{-1.5mu}=\mspace{-1.5mu}\sum_{{\succ}\in \mathcal R\colon x_i\succ x_j} f_i(v_{\succ},v_{x_j})\mspace{-1.5mu}=\mspace{-1.5mu}f_i(v_{x_j},t)\leq 1$, where the first equality follows by definition of $\pi$, the second one by flow conservation, and the last one because $c(v_{x_j},t)=1$. 

    Finally, for Condition (4) (i.e., $\sum_{{\succ}\in\mathcal{R}}\sum_{(x_i,x_j)\in A(\rhd)} \pi({\succ}, (x_i,x_j))> {m\choose 2}-1$), we show for all $i\in \{1,\dots, m-3\}$ that the maximum flow in $G_{x_i}$ has value $m-i$. To prove this claim, we inductively assume that it holds for the first $i-1$ rounds, which implies that $\sum_{{\succ}\in\mathcal{R}} b_i({\succ})=\frac{(m-i)(m-i+1)}{2}$. Now, if the flow network $G_{x_i}$ does not admit a maximum flow of value $m-i$, we derive from the MaxFlow-MinCut equivalence that there is a non-empty set of candidates $Z\subseteq X_i$ and a corresponding set of rankings $\bar{\mathcal{R}}=\{{\succ}\in\mathcal{R}\colon\exists y\in Z\colon x_i\succ y\}$ such that $\sum_{{\succ}\in\bar{\mathcal{R}}} b_i({\succ})<|Z|$. Conversely, this means that $\sum_{{\succ}\in\mathcal{R}\setminus\bar{\mathcal{R}}} b_i({\succ})>\frac{(m-i)(m-i+1)}{2}-|Z|$, i.e., the rankings that put all candidates in $Z$ ahead of $x_i$ control most of the remaining budget. Based on this insight, we show that a candidate in $Z$ must have a higher Borda score than $x_i$, which contradicts the definition of \fb. Hence, we infer inductively that the total remaining budget in round $i=m-2$ is $\frac{(m-i)(m-i+1)}{2}=3$. Finally, we show that at least $\frac{9}{4}$ of this budget will be spent in the last two rounds, i.e., the total left-over budget is at most $\frac{3}{4}$. This proves this condition because the sum of all payments is equal to the the difference between the total initial budget and the total remaining budget after the execution of \fb. In particular, this means that $\sum_{{\succ}\in\mathcal{R}}\sum_{(x_i,x_j)\in A(\rhd)} \pi({\succ}, (x_i,x_j))=\sum_{{\succ}\in\mathcal{R}} b_1({\succ})-b_m({\succ})\geq {m\choose 2}-\frac{3}{4}$.
\end{proof}

Lastly, analogously to \Cref{thm:propBordaAR}, we examine the average utility that \fb guarantees to an arbitrary subprofile $S$ as a function of $|S|$. Specifically, we will prove that \fb gives the same lower bound on the average utility of every subprofile as \psb. For this result, it is crucial that \fb chooses the maximum flow that minimizes the cost per utility ratio in every step. Moreover, we note that the following statement is largely unrelated to sPJR: while this axiom guarantees that the output ranking agrees with any subprofile $S$ in a number of pairs of candidates that is at least proportional to $|S|$, the utility of each ranking in $S$---and therefore also the average utility---may still be low. 

\begin{restatable}{theorem}{flowBordaAR}\label{thm:flowBordaAR}
    Let $R$ be a profile on $m$ candidates and $\rhd=\text{\fb}(R)$ the ranking chosen by the Flow-adjusting Borda rule. It holds for every subprofile $S$ of $R$ that 
    \[\frac{1}{|S|}\sum_{{\succ}\in\mathcal{R}} S({\succ})u({\succ},\rhd)\geq {m\choose2}\cdot\frac{|S|}{4}-\frac{3}{16}.\]
\end{restatable}
\begin{proof}[Proof Sketch]
Fix a profile $R$, let $\rhd=x_1,\dots,x_m$ be the ranking returned by \fb, and, for all $i\in\{1,\dots,m-1\}$, let $b_i({\succ})$, $X_i$, $G_{x_i}$, and $f_i$ be defined as in \Cref{thm:flowBordaPR}. Our goal is to show that $\frac{f_i(s,v_{\succ})}{b_i({\succ})\cdot u({\succ},x_i,X_i)}\leq \frac{4}{(m-i)(m-i+1)}$ for all ${\succ}\in\mathcal{R}$ and $i\in \{1,\dots, m-3\}$ (where we assume that $\frac{0}{0}=0$). Since $f_i(s,v_{\succ})=b_i({\succ})-b_{i+1}({\succ})$, this implies that $(b_i({\succ})-b_{i+1}({\succ}))\cdot \frac{(m-i)(m-i+1)}{4}\leq b_i({\succ}) \cdot u({\succ},x_i,X_i)$. From this point on, we can the use a similar argument as for \Cref{thm:propBordaAR}. 
    Now, to prove our auxiliary inequality, we fix a round $i\in \{1,\dots, m-3\}$ and consider the modified flow network $G_{x_i}'=(V,E,c')$, which is obtained from the network $G_{x_i}$ used by \fb by setting $c'(s,v_{\succ})=b_i({\succ})\cdot u({\succ},x_i,X_i)$ for all ranking vertices $v_\succ$ and $c'(v_x,t)=\frac{(m-i)(m-i+1)}{4}$ for all candidate vertices $v_x$ (the capacities between ranking vertices and candidate vertices are still unbounded). We prove that this network admits a maximum flow $f^*$ of value $\frac{(m-i)^2(m-i+1)}{4}$ because, otherwise, $x_i$ would not maximize the Borda score $U(b_i,x_i,X_i)$. Based on $f^*$, we next show that $f'(e)=\frac{4f^*(e)}{(m-i)(m-i+1)}$ is a valid maximum flow in the original network $G_{x_i}$. Moreover, this flow satisfies that $\frac{f'(s,v_{\succ})}{b_i({\succ})u({\succ},x_i,X_i)}\leq \frac{4}{(m-i)(m-i+1)}$ for all ${\succ}\in\mathcal{R}$. Since the maximum flow $f_i$ chosen by \fb minimizes $\max_{\succ\in\mathcal{R}} \frac{f_i(s,v_{\succ})}{b_i({\succ})u({\succ},x_i,X_i)}$, this implies that $\frac{f_i(s,v_{\succ})}{b_i({\succ})\cdot u({\succ},x_i,X_i)}\leq \frac{4}{(m-i)(m-i+1)}$ for all ${\succ}\in\mathcal{R}$.
\end{proof}

\section{Conclusion}

In this paper, we design proportional social welfare functions (SWFs) by transferring insights from approval-based committee voting and participatory budgeting to rank aggregation. In more detail, our central proportionality axiom is called uPJR and requires that every input ranking with weight $\alpha$ agrees with at least $\lfloor\alpha{m\choose 2}\rfloor$ pairwise comparisons of the output ranking. We first show that the Squared Kemeny rule, which was suggested by \citet{LPW24a} to compute proportional rankings, fails even a weakening of this axiom. We hence design new SWFs and, to this end, prove that uPJR is implied by a more structured fairness notion called rank-priceability. Based on this insight, we introduce the Proportional Sequential Borda rule (\psb), a simple rule that satisfies rank-priceability and thus also uPJR. We also prove that \psb guarantees to every subprofile $S$ an average utility that is linear in the size of $S$, thereby giving another strong fairness property that extends beyond individual rankings. Moreover, we introduce a variant of \psb called the Flow-adjusting Borda rule (\fb), which satisfies an even stronger proportionality axiom called sPJR. This axiom extends uPJR to arbitrary groups of rankings: any group of rankings needs to be represented by a number of pairwise comparisons that is at least proportional to it total weight. 

Our work offers numerous possibilities for future work and we next discuss three particularly interesting directions. Firstly, given our success in designing SWFs that satisfy variants of PJR in the context of rank aggregation, it seems interesting to analyze stronger fairness notions. One could, for instance, adopt notions such as EJR or core-stability from approval-based committee voting to rank aggregation and aim to find mechanisms satisfying these properties. Secondly, we note that, while most fairness notions in participatory budgeting and committee voting focus on groups of voters with similar preferences, none of our results relies on this idea. Partly, this is because there are multiple ways to define similar input rankings (e.g., we may consider two rankings similar if they have a small swap distance or if they agree on a large prefix) and because it is not clear how to exploit this precondition. However, we would find it interesting to strengthen both our axiomatic and quantitative results by focusing on cohesive groups of rankings. Lastly, maybe the biggest restriction of this paper is to define the utility in terms of the pairwise agreement of rankings. While this approach is frequently encountered in the literature, it, e.g., neglects that the first position of the output ranking has often a higher value than other positions. Thus, it seems appealing to extend our results to more general utility functions.

\section*{Acknowledgments}

The author thanks Haris Aziz and Dominik Peters for helpful discussions and the anonymous EC reviewers for their valuable feedback. The research reported in this paper was mostly conducted while Patrick Lederer was a postdoctoral fellow at UNSW Sydney, where he was supported by the NSF-CSIRO grant on ``Fair Sequential Collective Decision-Making'' (RG230833). At the University of Amsterdam, Patrick Lederer is supported by the European Research Council (ERC) via the ERC Synergy Grant on "Advancing Digital Democracy Innovation" (\url{doi.org/10.3030/101166894}). 

Funded by the European Union. Views and opinions expressed are however those of the authors only and do not necessarily reflect those of the European Union or the European Research Council Executive Agency. Neither the European Union nor the granting authority can be held responsible for them.

\clearpage
\appendix

\section{The Ranked Method of Equal Shares}\label{app:RMES}

In this appendix, we discuss another variant of the Proportional Sequential Borda rule, with the aim to clarify the relation between our rules and the Method of Equal Shares \citep{PeSk20a,PPS21a}. To this end, we will introduce the Ranked Method of Equal Shares (\rmes). Moreover, we will show that this rule satisfies uPJR and additionally guarantees to pick the first $\lfloor\frac{m}{4}\rfloor$ candidates in a highly utilitarian way. 

To introduce \rmes, we will first outline the Method of Equal Shares as used in participatory budgeting. In this setting there is a finite set of voters $N$, a set of candidates $C$ with costs $c$, and each voter $j\in N$ reports an approval ballot $A_j$ over the candidates. Moreover, there is a total budget $B$ and the goal is to choose a representative subset of the candidates whose total cost does not exceed $B$. Now, just as \psb, the Method of Equal Shares uniformly distributes the budget $B$ to the voters who use it to buy candidates. Thus, the initial budget of every voter is $b_1(j)=B/|N|$ and we denote the remaining budget of voter $j$ at round $i$ with $b_i(j)$. Then, the Method of Equal Shares chooses in the $i$-th round the candidate $x^*$ that minimizes the cost per utility ratio, i.e., the candidate that minimizes the value $\rho$ with $\sum_{j\in N} \min(\rho\cdot u(x^*,A_j), b_i(j))=c(x^*)$, where $u(x^*,A_j)=1$ if $x\in A_j$ and $u(x^*,A_j)=0$ otherwise. After buying candidate $x^*$ and adding it to the outcome, the budget of each voter is decreased by his contribution to the cost of $x^*$, i.e., $b_{i+1}(j)=b_i(j)-\min(\rho\cdot u(x^*,A_j), b_i(j))$. This process is repeated until no candidate can be bought anymore. 

We first note that the Method of Equal Shares can be easily extended to account for vote shares to make it more similar to our rank aggregation setting. In this case, a profile $R$ specifies for every approval ballot $A$ the share of voters that report $A$, and we denote this share by $R(A)$. For example, when there are $100$ voters and $30$ of them report the same approval ballot $A$, then $R(A)=30\%$. This means that the total share of the voters reporting a given approval ballot $A$ is initially $b_1(A)=R(A)\cdot B$. We note that, when $B={m\choose 2}$, this is precisely the initial budget handed out by \psb. Now, when using this voter share representation, the Method of Equal Shares picks in each round the candidate $x^*$ that minizes the value $\rho$ for which $\sum_{A\colon A\subseteq C} \min(\rho\cdot b_1(u)\cdot u(x^*, A), b_i(u))=c(x^*)$.\footnote{It may seem counterintuitive to have the factor $b_1(u)$ in the first argument of the minimum. To see why this is correct, we start again at the discrete setting and suppose that there is a set $S$ of $k$ voters that report the same ballot $A$. For any candidate $x$, these voters pay in total $\sum_{j\in S} \min(\rho_x u(x, A), b_i(j))$ (where $\rho_x$ is the value for which $x$ can be afforded). Since all voters with the approval ballot have the same budget during the execution of the Method of Equal Shares, this term is equivalent to $\min(k\cdot\rho_x  u(x,A), k \cdot b_i(j))$. Now, $k\cdot b_i(j)$ is precisely $b_i(A)$ in our notation. Further, for the first parameter, we observe that we can arbitrarily rescale the vote shares because the values $\rho_x$ absorb such changes: if we multiple all vote shares by a constant $c$, the value $\rho_x$ of every candidate will simply be divided by $c$. Hence, we can replace the value $k$ with $b_1(u)=\frac{k}{|N|}\cdot B$.} 
Further, each budget $b_i(A)$ is reduced by $\min(\rho\cdot b_1(A)\cdot u(x^*), b_i(A))$ before proceeding with the next round.

Now, it follows from the proof of \Cref{thm:propBordaRP} that \psb admits an almost equivalent formulation. Specifically, as long as there are at least $4$ candidates left, the \psb rule always picks the candidate $x$ that minimizes the value $\rho$ such that $\sum_{{\succ}\in\mathcal{R}} \min(\rho\cdot b_i({\succ})\cdot u({\succ},x_i,X_i), b_i({\succ}))=m-i$. We emphasize here that the first argument of the minimum is defined based on $b_i({\succ})$, the budget in the $i$-th round, whereas the Method of Equal Shares is defined based on $b_1({\succ})$. To see why this formulation indeed captures \psb, we note that the value $\rho_x$ for which a candidate $x$ can be bought is lower bounded by $\frac{m-i}{U(b_i,x,X_i)}$ because $\sum_{{\succ}\in\mathcal{R}} \min(\frac{m-i}{U(b_i,x,X_i)}\cdot b_i({\succ})\cdot u({\succ},x_i,X_i), b_i({\succ}))\leq \sum_{{\succ}\in\mathcal{R}} \frac{m-i}{U(b_i,x,X_i)}\cdot b_i({\succ})\cdot u({\succ},x_i,X_i)=m-i$. Clearly, the fraction $\frac{m-i}{U(b_i,x,X_i)}$ is minimized by the Borda winner. Moreover, as we have shown in the proof of \Cref{thm:propBordaRP}, it holds that for each ranking $\succ$ that $\frac{m-i}{U(b_i,x^*,X_i)}\cdot b_i({\succ})\cdot u({\succ},x^*,X_i)\leq  b_i({\succ})$ for the Borda winner $x^*$ when $|X_i|\geq 4$, so the Borda winner indeed minimizes the cost per utility ratio. 

Motivated by this analogy between the Method of Equal Shares and the Proportional Sequential Borda rule, it seems tempting to consider the SWF that picks in every round $i$ the candidate $x_i$ that minimizes the price $\rho$ such that $\sum_{{\succ}\in\mathcal{R}} \min(\rho \cdot b_1({\succ}) \cdot u({\succ},x_i, X_i), b_i({\succ}))=m-i$ and then update the budgets accordingly. However, it turns out that this method fails rank-priceability as a ranking may pay more for a candidate than the utility it obtains.\footnote{
For an example, consider the rankings $\succ_1=x_1,\dots, x_6$ and $\succ_2=x_6,\dots,x_1$ and let $R$ be the profile defined by $R({\succ_1})=\frac{47}{60}=\frac{11.75}{15}$ and $R(\succ_2)=\frac{13}{60}=\frac{3.25}{15}$. Since ${6\choose 2}=15$, the initial budgets of $\succ_1$ and $\succ_2$ are $b_1({\succ_1})=R({\succ_1})\cdot 15=11.75$ and $b_1(\succ_2)=R(\succ_2)\cdot 15=3.25$. In the first two rounds, it is easy to verify that \rmes{} chooses $x_1$ and $x_2$ and that $\succ_1$ will pay the full cost of these candidates. Hence, in the third step, we have that $b_3({\succ_1})=11.75-5-4=2.75$ and $b_3(\succ_2)=3.25$. This means that $x_3$ is no longer feasible because $\succ_1$ has not enough budget left to pay for this candidate and $\succ_2$ gains no utility from $x_3$. If we do not include the utility $u({\succ}, x, X_i)$ in the minimum for computing $\rho$, we would thus buy $x_4$ for a cost per utility ratio of $\rho=\frac{3}{2\cdot11.75+1\cdot 3.25}=\frac{12}{107}$. In turn, we infer that $\succ_1$ pays $\frac{12}{107}\cdot \frac{47}{4}\cdot 2=\frac{282}{107}\approx2.64$ and $\succ_2$ pays $\frac{12}{107}\cdot \frac{13}{4}\cdot {1}=\frac{39}{107}\approx 0.36$. However, this means that $\succ_2$ pays more than its obtained utility in this step.} 
We will thus include the term $u({\succ}, x_i, X_i)$ as a third argument of the minimum. Thus, the \emph{Ranked Method of Equal Shares (\rmes)} iteratively selects candidates based on the budgets $b_i({\succ})$ of the input rankings and the set of remaining candidates $X_i$. As for \psb, it holds in the first round that $b_1({\succ})=R({\succ})\cdot{m\choose2 }$ and $X_1=C$. Now, in each round $i\in \{1,\dots, m-2\}$, \rmes identifies the candidate $x_i\in X_i$ that minimizes the value $\rho_i$ such that 
\[\sum_{{\succ}\in\mathcal{R}} \min(\rho_i\cdot b_1({\succ})\cdot u({\succ}, x_i, X_i), b_i({\succ}), u({\succ}, x_i, X_i))=m-i.\]
Then, we place this candidate at the $i$-th position of the output ranking, remove $x_i$ from the active candidates, and reduce the budget of every ranking according to its contribution to the cost of $x_i$. More formally, we set $X_{i+1}=X_i\setminus \{x_i\}$ and $b_{i+1}({\succ})=b_i({\succ})-\min(\rho_{i}\cdot b_1({\succ})\cdot u({\succ},x_{i}, X_i), b_i({\succ}), u({\succ}, x_i, X_i))$ for all $\succ$. After this, we proceed with the next round. Finally, since this approach is only guaranteed to work when $|X_i|\geq 3$, we decide the order over the last two candidates by majority voting with respect to the remaining budgets: assuming that $x$ and $y$ are the last two candidates, we place $x$ ahead of $y$ in the output ranking if $\sum_{{\succ}\in\mathcal{R}\colon x\succ y} b_{m-1}({\succ})\geq \sum_{{\succ}\in\mathcal{R}\colon y\succ x} b_{m-1}({\succ})$. Otherwise, we put $y$ at the $m-1$-th position. As usual, ties can be broken arbitrarily.

\begin{example}[The Ranked Method of Equal Shares] Consider the two ranking ${\succ_1}=x_1,x_2,x_3,x_4, x_5$ and ${\succ_2}=x_4,x_5,x_1,x_3,x_2$ and let $R$ the be same profile as in \Cref{ex:psb}, i.e., $R({\succ_1})=0.6$ and $R(\succ_2)=0.4$. If the tie-breaking favors candidates with smaller indices, \rmes chooses the ranking $\rhd=x_1,x_2,x_4,x_5,x_3$ for this profile, as verified by the following computations.

\begin{center}
    \begin{tikzpicture}
\node at (0,0) (step1) {	
			\votermultiplicity{$6$ $(6)$}{\weakorder{{1},{2},{3},{4},{5}}}
			\votermultiplicity{$4$ $(4)$}{\weakorder{{4},{5},{1},{3},{2}}}
		};

\draw[line width = 0.2mm]
  ([yshift=-0.55cm,xshift=0.1cm]step1.north west) -- ([yshift=-0.55cm,xshift=-0.1cm]step1.north east);

\node at (1.3,0) {$\implies$};

\node at (2.6,0) (step2) {
    \votermultiplicity{$6$ $(3)$}{\weakorder{{2},{3},{4},{5}}}
	\votermultiplicity{$4$ $(3)$}{\weakorder{{4},{5},{3},{2}}}
};

\draw[line width = 0.2mm]
  ([yshift=-0.55cm,xshift=0.1cm]step2.north west) -- ([yshift=-0.55cm,xshift=-0.1cm]step2.north east);

\node at (3.9,0) {$\implies$};

\node at (5.2,0) (step3) {
    \votermultiplicity{$6$ $(0)$}{\weakorder{{3},{4},{5}}}
	\votermultiplicity{$4$ $(4)$}{\weakorder{{4},{5},{3}}}
};

\draw[line width = 0.2mm]
  ([yshift=-0.55cm,xshift=0.1cm]step3.north west) -- ([yshift=-0.55cm,xshift=-0.1cm]step3.north east);

\node at (6.5,0) {$\implies$};

\node at (7.8,0) (step4) {
    \votermultiplicity{$6$ $(0)$}{\weakorder{{3},{5}}}
	\votermultiplicity{$1$ $(0)$}{\weakorder{{5},{3}}}
};

\draw[line width = 0.2mm]
  ([yshift=-0.55cm,xshift=0.1cm]step4.north west) -- ([yshift=-0.55cm,xshift=-0.1cm]step4.north east);

\node at (9.1,0) {$\implies$};

\node at (10.4,0) (step5) {
    \votermultiplicity{$6$ $(0)$}{\weakorder{{3}}}
	\votermultiplicity{$4$ $(0)$}{\weakorder{{3}}}
};

\draw[line width = 0.2mm]
  ([yshift=-0.55cm,xshift=0.1cm]step5.north west) -- ([yshift=-0.55cm,xshift=-0.1cm]step5.north east);
    \end{tikzpicture}\end{center}

In this figure, we show the input rankings restricted to the available candidates and weighted by their initial budget $b_1({\succ})$. In brackets, we also show the remaining budget in each round. Analogous to \psb, \rmes picks in the first step $x_1$ for a price $\rho_1=\frac{1}{8}$, so $\succ_1$ pays $\frac{1}{8}\cdot 4\cdot 6=3$ and $\succ_2$ pays $\frac{1}{8}\cdot 2\cdot 4=1$. Hence, the new budgets are $b_2({\succ_1})=3$ and $b_2(\succ_2)=3$. In the second step, both $x_2$ and $x_4$ can be bought for a price of $\rho=\frac{1}{6}=\frac{3}{18}$. Because we assume that the tie-breaking favors $x_2$ to $x_4$, we pick $x_2$ next. Consequently, $\succ_1$ pays $\frac{1}{6}\cdot 3\cdot 6=3$ and $\succ_2$ pays $\frac{1}{6}\cdot 0\cdot 4=0$, which means that $b_3({\succ_1})=0$ and $b_3(\succ_2)=3$. From here on, \rmes picks the candidates according to $\succ_2$ as $\succ_1$ has no budget left. 
\end{example}

We note that in this example, there is always a candidate $x_i$ that can be bought for a finite price. We next show that this observation holds in general as \rmes{} is well-defined. Moreover, we will also prove that \rmes satisfies rank-priceability and thus uPJR.

\begin{restatable}{theorem}{rmesRP}\label{thm:rmesRP}
    \rmes{} is well-defined and satisfies rank-priceability. 
\end{restatable}
\begin{proof}
    We fix a profile $R$ and show both claims of the theorem independently.\medskip

    \textbf{Claim 1: \rmes{} is well-defined.}\\
    To show that \rmes{} is well-defined, we need to prove that for each round $i\in \{1,\dots, m-2\}$, there is a candidate $x$ with $\rho_x<\infty$. In particular, note that \mes is obviously well-defined in the last round as we simply apply the majority rule for the remaining two candidates. Now, fix some round $i\in \{1,\dots,m-2\}$ and assume that a candidate was bought in all previous rounds. Clearly, if $i=1$, this assumption is true and it will hold inductively for $i>1$. Moreover, let $X_i$ denote the candidates that have not been placed in the output ranking yet, $b_i({\succ})$ the remaining budgets, and $b_1({\succ})$ the initial budgets. We first note that for every candidate $x\in X_i$ and ranking ${\succ}\in\mathcal{R}$, it holds that 
    \[\min(b_1({\succ})\cdot u({\succ}, x, X_i), b_i({\succ}), u({\succ},x,X_i))=\min(b_i({\succ}), u({\succ},x,X_i)). \]
    In more detail, if $u({\succ},x,X_i)=0$, then clearly $\min(b_1({\succ})\cdot u({\succ}, x, X_i), b_i({\succ}), u({\succ},x,X_i))=u({\succ},x,X_i)$ as all other values are non-negative. On the other hand, if $u({\succ}, x, X_i)\geq 1$, then $b_1({\succ}) u({\succ}, x, X_i)\geq b_i({\succ})$ since $b_1({\succ})\geq b_i({\succ})$. Hence, to show that there is always an affordable candidate, it suffices to prove that there is always a candidate $x$ such that $\sum_{{\succ}\in\mathcal{R}} \min(b_i({\succ}), u({\succ}, x, X_i))\geq m-i$. This means that $x$ is affordable for a price $\rho\leq 1$.

    Now, to prove this claim, we first recall that the total initial weight is $\sum_{{\succ}\in\mathcal{R}} b_1({\succ})={m\choose 2}$ and that we decrease the budget by $m-j$ in all rounds $j<i\leq m-2$. Hence, the total remaining budget in the $i$-th round is $\sum_{{\succ}\in\mathcal{R}} b_i({\succ})=\frac{(m-i)(m-i+1)}{2}$. Next, we proceed with a case distinction and first suppose that there is a ranking $\succ^*$ such that $b_i(\succ^*)\geq m-i-1$. Furthermore, let $x$ denote the top-ranked candidate of $\succ^*$. If it even holds that $b_i(\succ^*)\geq m-i$, then this ranking alone can afford $x$ by itself because $b_i(\succ^*)\geq u(\succ^*, x, X_i)=m-i$. On the other hand, if $b_i({\succ^*})< m-i$, then $\min(u({\succ^*}, x, X_i), b_i({\succ^*}))=b_i({\succ^*})\geq m-i-1$. Next, let $Z$ denote the set of rankings $\succ$ such that ${\succ}\neq{\succ^*}$ and $u({\succ}, x, X_i)>0$, and let $B_i^Z=\sum_{{\succ}\in Z} b_i({\succ})$ denote the total remaining budget of these rankings. If $B_i^Z\geq m-i-b_i(\succ^*)$, $x$ can again be bought. In more detail, if there is a ranking ${\succ}\in Z$ with $b_i({\succ})\geq 1$, then $\min(b_i(\succ^*), u(\succ^*, x, X_i))+\min(b_i({\succ}), u({\succ}, x, X_i))\geq m-i$. On the other hand, if $b_i({\succ})<1$ for all ${\succ}\in Z$, it holds for each of these rankings that $\min(b_i({\succ}), u({\succ}, x, X_i))=b_i({\succ})$. Consequently, $\min(b_i(\succ^*), u(\succ^*, x, X_i))+\sum_{{\succ}\in Z} \min(b_i({\succ}), u(\succ^*, x, X_i))\geq m-i$. 
    
    As the last subcase, suppose that $B_i^Z<m-i-b_i(\succ^*)$ and let $y$ denote the second-best candidate in $X_i$ with respect to $\succ^*$. Since $b_i(\succ^*)\geq m-i-1$ by assumption, it holds that $\min(b_i({\succ^*}), u({\succ^*}, y, X_i))=u({\succ^*}, y,X_i)=m-i-1$. Furthermore, we observe that all rankings in  $\mathcal{R}\setminus(Z\cup \{\succ^*\})$ bottom-rank $x$. Moreover, we have that \begin{align*}
        \sum_{{\succ}\in\mathcal{R}\setminus (Z\cup \{\succ^*\})} b_i({\succ})&=\frac{(m-i)(m-i+1)}{2}-\sum_{{\succ}\in Z\cup \{\succ^*\}} b_i({\succ})\\
        &> \frac{(m-i)(m-i+1)}{2}-(m-i)\\
        &\geq 1.
    \end{align*}
    Here, the first inequality follows because $B_i^Z<m-i-b_i(\succ^*)$ and the second one because $m-i\geq 2$. Further, we note that $u({\succ}, y, X_i)\geq 1$ for all rankings ${\succ}\in \mathcal{R}\setminus (Z\cup \{\succ^*\})$ because $y\succ x$. Based on an analogous case distinction as for $x$ when $B\geq m-i-b_i(\succ^*)$, one can now show that $y$ can be afforded.
    
    As the last case, suppose that $b_i({\succ})<m-i-1$ for all ${\succ}\in \mathcal R$. In this case, we note that, for every ranking $\succ$, there are at least two candidates $x$ and $y$ such that $u(\succ,x,X_i)>b_i(\succ)$ and $u(\succ,y,X_i)>b_i(\succ)$. In more detail, $x$ and $y$ are the best and second-best candidates among $X_i$ in $\succ$. By this insight, it holds that $\sum_{x\in X_i} \min(u(\succ,x,X_i),b_i(\succ))\geq 2 b_i(\succ)$ for all rankings $\succ$. Now, by summing over all rankings, we get that 
    \[\sum_{x\in X_i}\sum_{\succ\in\mathcal{R}} \min(u(\succ,x,X_i),b_i(\succ))\geq 2 \sum_{\succ\in\mathcal{R}}b_i(\succ)=(m-i)(m-i+1).\]

    Here, the last inequality uses that $\sum_{\succ\in\mathcal{R}}b_i(\succ)=\frac{(m-i)(m-i+1)}{2}$. Finally, since there are $m-i+1$ candidates in $X_i$, there must be a candidate such that $\sum_{{\succ}\in \mathcal{R}} \min(b_i({\succ}), u({\succ}, x, X_i))\geq m-i$. Hence, we conclude that \rmes is well-defined during the first $m-2$ steps.\medskip

\textbf{Claim 2: \rmes{} satisfies rank-priceability.}\\
    Consider the ranking $\rhd=x_1,\dots, x_m$ chosen by \rmes{} for our input profile $R$. Moreover, let $\rho_i$ denote the price for which candidate $x_i$ is bought for all $i\in \{1,\dots, m-2\}$, let $b_i({\succ})$ denote the budget of ranking $\succ$ in the $i$-th step, and let $X_i=\{x_i,\dots,x_m\}$ the remaining candidates in the $i$-th round. We will analyze the payment scheme $\pi$ defined as follows: for $i\in \{1,\dots, m-2\}$, we set $\pi({\succ}, x_i)=\min(\rho_i b_1({\succ}) u({\succ},x_i,X_i), b_i({\succ}), u({\succ},x_i,X_i))$ for all $\succ$. Further, for $i=m-1$, we set $\pi({\succ},x_i)=b_{m-1}({\succ})$ if $x_{m-1}\succ x_m$ and $\pi({\succ},x_i)=0$ otherwise. Finally, the definition of rank-priceability requires that $\pi({\succ},x_m)=0$ for all ${\succ}\in\mathcal{R}$.

    We first note that Condition (1) of rank-priceability is satisfies for all $i\in \{1,\dots, m-2\}$ because $\pi({\succ}, x_i)\leq u({\succ}, x_i, \{x_i,\dots, x_m\})$ by definition of our scheme. Further, when $i=m-1$, then it holds that $\sum_{{\succ}\in\mathcal{R}} b_i({\succ})=1$. Since $\pi({\succ}, x_{m-1})=0$ if $u({\succ}, x_{m-1}, \{x_{m-1}, x_m\})$ and $\pi({\succ}, x_{m-1})=b_{m-1}({\succ})\leq 1$ if $u({\succ}, x_{m-1}, \{x_{m-1}, x_m\})$, Condition (1) also holds in this round. Secondly, since $\pi({\succ}, x_i)\leq b_i({\succ})$ in every step and $b_1({\succ})=R({\succ})\cdot{m\choose 2}$, Condition (2) of rank-priceability follows. Condition (3) of rank-priceability follows immediately from the definition of \rmes because we pay exactly $m-i$ during all rounds $i\in \{1,\dots, m-2\}$ and at most $1$ in the $m-1$-th round because the total remaining budget is $1$. This also implies Step (4). In more detail, before the $m-1$-th step, the total remaining budget is $1$ and we spent at least half on $x_i$. Hence, the total remaining budget in the end is at most $0.5$. Because the total initial budget is ${m\choose 2}$, this means that the total payments sum up to at least ${m\choose 2}-1$. Hence, rank-priceability is indeed satisfied.
\end{proof}

While \Cref{thm:rmesRP} establishes that \rmes is a proportional SWF, we will next show that this rule is still rather utilitarian. Specifically, we will prove that the first $\lfloor\frac{m}{4}\rfloor$ candidates of this rule are chosen only based on the Borda scores with respect to the initial weights. Put differently, \rmes agrees for roughly the first quarter of the candidates with the highly utilitarian ranking obtained by repeatedly placing the Borda winner in the next available position of the output ranking and removing it from the input profile. These candidates determine roughly $\frac{7}{16}$ of all pairwise comparisons, thus showing that a significant portion of the total utility is assigned in a utilitarian way. 

\begin{restatable}{proposition}{rmesprices}\label{prop:rmesprices}
    Fix a profile $R$ on $m$ candidates and let $\rhd=x_1,\dots,x_m$ denote the ranking chosen by \rmes. It holds for all $i\in \{1,\dots, \lfloor\frac{m}{4}\rfloor\}$ that \[\rho_{x_i}=\frac{m-i}{U(b_1,x_i,\{x_i,\dots,x_m\})}\qquad \text{and} \qquad x_i=\arg\max_{x\in \{x_i,\dots, x_m\}} U(b_1, x, \{x_i,\dots,x_m\}).\]
\end{restatable}
\begin{proof}
    Fix a profile $R$, let $\rhd=x_1,\dots,x_m$ denote the ranking chosen by \rmes, and let $X_i=\{x_i,\dots,x_m\}$ for all $i\in \{1,\dots, m\}$. We first note that the proposition is trivial for $m\leq 3$ because $\lfloor \frac{m}{4}\rfloor=0$ in this case. Hence, assume that $m\geq 4$. We will show the proposition by induction and fix a round $i\in  \{1,\dots, \lfloor(1-\frac{m}{4}) m\rfloor\}$. We inductively suppose for all rounds $j\in \{1,\dots, i-1\}$ that $\rho_j=\frac{m-j}{U(b_1,x_j,X_j)}$ and $x_j=\arg\max_{x\in X_j} U(b_1, x, X_j)$. Clearly, when $i=1$, this assumption is true as there were no previous rounds. The central idea of our proof is to show for every input ranking ${\succ}\in\mathcal{R}$ that \[\min\left(\frac{(m-i)b_1({\succ}) u({\succ},x_i,X_i)}{U(b_1,x_i,X_i)}, b_i({\succ}), u({\succ},x_i,X_i)\right)=\frac{(m-i) b_1({\succ}) u({\succ},x_i,X_i)}{U(b_1,x_i,X_i)}.\] 
    This implies that $x_i$ can be bought for a price of $\rho_i=\frac{m-i}{U(b_1,x_i,X_i)}$. 
    Furthermore, if there was an candidate $x_k$ with $U(b_1,x_k,X_i)>U(b_1,x_i,X_i)$, this candidate could be bought for a price of $\frac{m-i}{U(b_1,x_k,X_i)}<\frac{m-i}{U(b_1,x_i,X_i)}$, which contradicts that \rmes chooses candidate $x_i$ in the $i$-th round. So, it follows from our claim also that $x_i$ is the candidate maximizing $U(b_1,x,X_i)$.  
    
    To prove the above equality, we fix an input ranking ${\succ}\in\mathcal{R}$ and let $x^*$ denote the candidate maximizing $U(b_1,x,X_i)$. We will first show that $\frac{m-i}{U(b_1,x^*, X_i)} b_1({\succ}) u({\succ}, x^*, X_i)\leq u({\succ},x^*,X_i)$. For this, let $x$ denote the top-ranked candidate among $X_i$ with respect to $\succ$. It holds that $u({\succ},x, X_i)=m-i$ as there are $m-i+1$ candidates remaining. Since $x^*$ maximizes the Borda score, it follows that $U(b_1,x^*,X_i)\geq U(b_1,x,X_i)\geq b_1({\succ}) \cdot(m-i)$. Hence, $\frac{m-i}{U(b,x^*,X_i)} b_1({\succ}) u({\succ},x^*,X_i)\leq u({\succ},x^*,X_i)$ as required.
    
    Next, we will show that $\frac{m-i}{U(b_1,x^*, X_i)}b_1({\succ}) u({\succ}, x^*, X_i)\leq b_i({\succ})$. To this end, we observe that, by the induction hypothesis, it holds for all input rankings ${\succ}\in\mathcal{R}$ that 
    \[b_i({\succ})=b_1({\succ})-\sum_{j=1}^{i-1} \frac{m-j}{U(b_1,x_j, X_j)} b_1({\succ}) ({\succ}, x_j, X_j).\] 
    In particular, each ranking $\succ$ pays $\frac{m-j}{U(b_1,x_j,X_j)} b_1({\succ}) u({\succ},x_j,X_j)$ in every round $j\in \{1,\dots, i-1\}$. If this was not the case in some round $j$, candidate $x_j$ could not have been afforded for a price of $\frac{m-j}{U(b_1,x_j,X_j)}$ as the total payments do not add up to $m-i$. By combining our insights and dividing by $b_1({\succ})$, it suffices to show that
    \[\frac{m-i}{U(b_1,x^*, X_i)} u({\succ}, x^*, X_i)\leq 1-\sum_{j=1}^{i-1} \frac{m-j}{U(b_1,x_j, X_j)} u({\succ}, x_j, X_j).\]

    For this, we observe for every round $k\in \{1,\dots, m-1\}$ that $\sum_{x\in X_k} \sum_{{\succ}\in\mathcal{R}} b_1({\succ}) u({\succ},x, X_k)=\sum_{{\succ}\in\mathcal{R}} b_1({\succ}) \frac{(m-k)(m-k+1)}{2}=\frac{m(m-1)}{2}\cdot\frac{(m-k)(m-k+1)}{2}$. Since $x^*$ maximizes the Borda score in the $i$-th round, this means that $U(b_1,x^*,  X_i)\geq \frac{m(m-1)}{4}\cdot (m-i)$ as the maximum is lower bounded by the average. In turn, it follows that $\frac{m-i}{U(b_1,x^*,X_i)}\leq \frac{4}{m(m-1)}$. Analogously, it follows for all candidates $x_j$ with $j\in \{1,\dots, i-1\}$ that $U(b_1,x_j,  X_j)\geq \frac{m(m-1)}{4}\cdot (m-j)$. Using these insights, it follows that 
    \begin{align*}
        \frac{m-i}{U(b_1,x^*, X_i)} u({\succ}, x^*, X_i)&\leq \frac{4}{m(m-1)} u({\succ}, x^*, X_i)\hspace{3cm}\text{and}\\
        1-\sum_{j=1}^{i-1} \frac{m-j}{U(b_1,x_j, X_j)} u({\succ}, x_j, X_j)&\geq 1-\sum_{j=1}^{i-1} \frac{4}{m(m-1)} u({\succ}, x_j, X_j).
    \end{align*}

    We will show that $\frac{4}{m(m-1)} u({\succ}, x^*, X_i)\leq 1-\sum_{j=1}^{i-1} \frac{4}{m(m-1)} u({\succ}, x_j, X_j)$. Equivalently, we can prove that $u({\succ}, x^*, X_i)+\sum_{j=1}^{i-1} u({\succ}, x_j, X_j)\leq \frac{m(m-1)}{4}$. To this end, we note that $u({\succ}, x^*, X_i)\leq m-i$ and $u({\succ}, x_j, X_j)\leq m-j$ for all $j\in \{1,\dots, m-1\}$. Hence, this inequality is satisfied when $\sum_{j=1}^i m-j\leq \frac{m(m-1)}{4}$. Next, we compute that $\sum_{j=1}^i m-j=\sum_{j=1}^{m-1} j - \sum_{j=1}^{m-i-1} i =\frac{m(m-1)}{2}-\frac{(m-i)(m-i-1)}{2}$. Consequently, we can further reduce our problem to showing that 
    \[\frac{m(m-1)}{2}-\frac{(m-i)(m-i-1)}{2}\leq \frac{m(m-1)}{4}.\]
    
    Finally, we will use that $i\leq \frac{m}{4}$ and that $m\geq 4$. In particular, this implies that 
    \begin{align*}
        \frac{(m-i)(m-i-1)}{2}\geq \frac{\frac{3m}{4} (\frac{3m}{4}-1)}{2} = \frac{9m^2}{32} -\frac{3m}{8}=\frac{m(m-1)}{4}+\frac{m^2}{32}-\frac{m}{8}\geq \frac{m(m-1)}{4}.
    \end{align*}
    
    Hence, it indeed holds that $\frac{m(m-1)}{2}-\frac{(m-i)(m-i-1)}{2}\leq \frac{m(m-1)}{4}$.
\end{proof}

Finally, we note that, just as for \psb and \fb, we can also show a lower bound on the average utility of any subprofile. Interestingly, we will show that \Cref{prop:rmesprices} implies a slightly better guarantee for large subprofiles compared to \psb because, roughly, this result entails a lower cost per utility ratio in the first steps. To state this bound, we fix an arbitrary profile $R$ on $m\geq 4$ candidates, the ranking $\rhd$ chosen by $\rmes$ for $R$, and we let $\xi={m-\lfloor\frac{m}{4}\rfloor\choose 2}$. Then, it holds for every subprofile $S$ of $R$ that 
    \[\frac{1}{|S|}\sum_{{\succ}\in \mathcal{R}} S({\succ}) u({\succ},\rhd)\geq \begin{cases}
    {m\choose 2}\cdot \frac{|S|}{4}-\frac{1}{8}\qquad &\text{if } {m\choose 2}|S|-0.5\leq \xi\\
    \frac{1}{2}\cdot{m\choose 2}\cdot(1-\frac{\xi}{{m\choose 2}|S|-0.5})+\frac{\xi+1}{4}\cdot \frac{\xi}{{m\choose 2}|S|}-\frac{1}{4|S|}\qquad &\text{if } {m\choose 2}|S|-0.5> \xi.
    \end{cases}
    \] 

To put this bound into perspective, we note that $\xi= \frac{9}{16}{m\choose 2}+\mathcal{O}(m)$. Thus, when ignoring lower order terms, our result recovers the average utility guarantee of \psb for the Ranked Method of Equal Shares when $|S|$ is less than $\frac{9}{16}$. By contrast, when $|S|>\frac{9}{16}$ our guarantee for \rmes is a convex combination of $\frac{\xi+1}{4}$ and $\frac{1}{2}{m\choose 2}$. In particular, if $|S|=1$, our bound shows that \rmes guarantees an utilitarian welfare of roughly ${m\choose 2}\cdot (\frac{1}{2}\cdot \frac{7}{16}+\frac{1}{4}\cdot\frac{9}{16}\cdot\frac{9}{16})=\frac{305}{1024}\cdot {m\choose 2}$ (where we ignored lower order terms).

\section{Proof of \Cref{prop:SqK}}\label{app:SqK}

In this section, we prove \Cref{prop:SqK}.

\SqK*
\begin{proof}
Fix a number of voters $m\geq 5$ and let $z=\lfloor\frac{(m-1)(m-2)}{4} \rfloor$.
We consider the following $4$ rankings: $\succ_1=x_1,x_2\dots, x_m$, $\succ_2=x_1,x_m\dots,x_2$, $\succ_3$ is a ranking that puts $x_1$ last and agrees with exactly $z$ pairs with $\succ_1$, and $\succ_4$ puts $x_1$ last and orders the remaining pairs exactly inversely to $\succ_3$. 
Based on these rankings, we define the following profile $R$. We note that, for notational simplicity, the weights in $R$ add up to $m\choose2$ instead of $1$. This is without loss of generality since we can multiply all weights with $1/{m\choose 2}$ without affecting the outcome. 
\begin{align*}
  R({\succ_1})=\frac{m}{5}\qquad \text{and} \qquad R({\succ_2})=R(\succ_3)=R(\succ_4)=\frac{1}{3}\left({m\choose 2}-\frac{m}{5}\right).
\end{align*}

Further, let $\rhd_{\mathtt{SqK}}$ denote the ranking returned by the Squared Kemeny rule for $R$. We will show that $\rhd_{\mathtt{SqK}}$  is equal to the ranking $\rhd^*=x_m,\dots, x_1$. This means that the ranking $\succ_1$ is without representation in $R$, thus proving the proposition. To show this claim, we will introduce additional notation. Specifically, we denote by $\Delta({\succ},\rhd)=|\{(x,y)\in \{C\setminus x_1\}^2\colon x\succ y\text{ and } y\rhd x\}|$ the swap distance between two rankings $\succ$ and $\rhd$ after removing $x_1$ from both rankings. Furthermore, we define by $d=1+|\{x\in C\setminus \{x_1\}\colon x\rhd_{\mathtt{SqK}} x_1\}|$ the rank of $x_1$ in $\rhd_{\mathtt{SqK}}$. Given this notation, the swap distance of $\rhd_{\mathtt{SqK}}$ to an input ranking $\succ_i$ is $\mathit{swap}(\succ_i,\rhd_{\mathtt{SqK}})=(d-1)+\Delta(\succ_i,\rhd_{\mathtt{SqK}})$ for $i\in \{1,2\}$ and $\mathit{swap}(\succ_i,\rhd_{\mathtt{SqK}})=(m-d)+\Delta(\succ_i,\rhd_{\mathtt{SqK}})$ for $i\in \{3,4\}$. We will next show the following auxiliary claims.
\begin{enumerate}[label=(\arabic*)]
    \item It holds that $\Delta(\rhd_{\mathtt{SqK}},\rhd^*)\leq\frac{1}{2}{m-1\choose 2}$.
    \item If $m\in \{5,6\}$ and $\Delta(\rhd_{\mathtt{SqK}},\rhd^*)\leq\frac{1}{2}{m-1\choose 2}$, then $d> \frac{2m}{3}$.
    \item If $m\geq 7$ and $\Delta(\rhd_{\mathtt{SqK}},\rhd^*)\leq\frac{1}{2}{m-1\choose 2}$, then $d> \frac{2m}{3}+1$.
    \item If $m\in \{5,6\}$ and $d>\frac{2m}{3}$, then $\Delta(\rhd_{\mathtt{SqK}},\rhd^*)=0$.
    \item If If $m\geq 7$ and $d>\frac{2m}{3}+1$, then $\Delta(\rhd_{\mathtt{SqK}},\rhd^*)=0$.
    \item If $\Delta(\rhd_{\mathtt{SqK}},\rhd^*)=0$, then $d=m$. 
\end{enumerate}

In combination, these observations imply that ${\rhd_{\mathtt{SqK}}}={\rhd^*}$ when $m\geq 5$. We will next prove our auxiliary claims.\medskip

\textbf{Proof of Claim (1):} We will first show that $\Delta(\rhd_{\mathtt{SqK}},\rhd^*)\leq \frac{1}{2}{m-1\choose 2}$. Assume for contradiction that this is not true and let $\bar\rhd$ denote the ranking derived from $\rhd_{\mathtt{SqK}}$ by inverting the order over the candidates $\{x_2,\dots, x_m\}$ while keeping the position of $x_1$ fixed. For instance, if $\rhd_{\mathtt{SqK}}=x_2,x_1,x_3,x_4,x_5$, then $\bar\rhd=x_5,x_1,x_4,x_3,x_2$. We will show that $\bar \rhd$ has a lower cost than $\rhd$ with respect to Squared Kemeny. To this end, we note $\succ_1$ orders the candidates in $\{x_2,\dots, x_m\}$ inversely to $\succ_2$ and that the same holds for $\succ_3$ and $\succ_4$. Hence, we compute that 
\begin{align*}
    &\Delta(\succ_1,\rhd_{\mathtt{SqK}})={m-1\choose 2}-\Delta(\succ_2,\rhd_{\mathtt{SqK}})=\Delta(\succ_2,\bar\rhd) ,\\
    &\Delta(\succ_2,\rhd_{\mathtt{SqK}})={m-1\choose 2}-\Delta(\succ_1,\rhd_{\mathtt{SqK}})=\Delta(\succ_1,\bar\rhd),\\
    &\Delta(\succ_3,\rhd_{\mathtt{SqK}})={m-1\choose 2}-\Delta(\succ_4,\rhd_{\mathtt{SqK}})=\Delta(\succ_4,\bar\rhd),\\
    &\Delta(\succ_4,\rhd_{\mathtt{SqK}})={m-1\choose 2}-\Delta(\succ_3,\rhd_{\mathtt{SqK}})=\Delta(\succ_3,\bar\rhd).\\
\end{align*}

Now, the cost of $\rhd_{\mathtt{SqK}}$ with respect to $\mathtt{SqK}$, denoted by $C_{\mathtt{SqK}}(\rhd_{\mathtt{SqK}})$, is given by 
\begin{align*}C_{\mathtt{SqK}}(\rhd_{\mathtt{SqK}})&=R({\succ_1})\cdot \left(d-1+\Delta(\succ_1, \rhd_{\mathtt{SqK}})\right)^2+R({\succ_2})\cdot \left(d-1+\Delta(\succ_2, \rhd_{\mathtt{SqK}})\right)^2\\
    &\quad+R(\succ_3)\cdot \left(m-d+\Delta(\succ_3, \rhd_{\mathtt{SqK}})\right)^2+R(\succ_4)\cdot \left(m-d+\Delta(\succ_4, \rhd_{\mathtt{SqK}})\right)^2.
\end{align*}

Furthermore, the cost of $\bar\rhd$ is 
\begin{align*}
C_{\mathtt{SqK}}(\bar\rhd)&=R({\succ_1})\cdot \left(d-1+\Delta(\succ_1, \bar \rhd)\right)^2+R({\succ_2})\cdot \left(d-1+\Delta(\succ_2, \bar\rhd)\right)^2\\
    &\quad+R(\succ_3)\cdot \left(m-d+\Delta(\succ_3, \bar\rhd)\right)^2+R(\succ_4)\cdot \left(m-d+\Delta(\succ_4, \bar\rhd)\right)^2\\
    &=R({\succ_1})\cdot \left(d-1+\Delta(\succ_2, \rhd_{\mathtt{SqK}})\right)^2+R({\succ_2})\cdot \left(d-1+\Delta(\succ_1, \rhd_{\mathtt{SqK}})\right)^2\\
    &\quad+R(\succ_3)\cdot \left(m-d+\Delta(\succ_4, \rhd_{\mathtt{SqK}})\right)^2+R(\succ_4)\cdot \left(m-d+\Delta(\succ_3, \rhd_{\mathtt{SqK}})\right)^2
\end{align*}

Next, we observe that $R(\succ_3)=R(\succ_4)$ and $R({\succ_1})<R({\succ_2})$ by the definition of $R$. Furthermore, we assumed that $\Delta(\rhd_{\mathtt{SqK}},\rhd^*)>\frac{1}{2}{m-1\choose 2}$, so $\Delta(\succ_1,\rhd_{\mathtt{SqK}})<\frac{1}{2}{m-1\choose 2}<\Delta(\succ_2,\rhd_{\mathtt{SqK}})$ as $\succ_2$ agrees with $\rhd^*$ on the order of $\{x_2,\dots, x_m\}$ and $\succ_1$ orders these candidates exactly inversely. Using these insights, we obtain that 
\begin{align*}
    &C_{\mathtt{SqK}}(\rhd_{\mathtt{SqK}})-C_{\mathtt{SqK}}(\bar\rhd)\\
    &=\left(R({\succ_1})-R({\succ_2})\right)\left(d-1+\Delta(\succ_1, \rhd_{\mathtt{SqK}})\right)^2+(R({\succ_2})-R({\succ_1}))\left(d-1+\Delta(\succ_2, \rhd_{\mathtt{SqK}})\right)^2\\
    &\quad+(R(\succ_3)-R(\succ_4))\left(m-d+\Delta(\succ_3, \rhd_{\mathtt{SqK}})\right)^2+(R(\succ_4)-R(\succ_3))\left(m-d+\Delta(\succ_4, \rhd_{\mathtt{SqK}})\right)^2\\
    &=(R({\succ_2})-R({\succ_1}))\left(\left(d-1+\Delta(\succ_2, \rhd_{\mathtt{SqK}})\right)^2-\left(d-1+\Delta(\succ_1, \rhd_{\mathtt{SqK}})\right)^2\right)\\
    &>0.
\end{align*}

This shows that $C_{\mathtt{SqK}}(\rhd_{\mathtt{SqK}})>C_{\mathtt{SqK}}(\bar\rhd)$, which contradicts that $\rhd_{\mathtt{SqK}}$ is chosen by the Squared Kemeny rule as this rule chooses the ranking minimizing $C_{\mathtt{SqK}}$.\medskip

\textbf{Proofs of Claims (2), (3), and (6)}: We will next show that $x_1$ cannot be ranked to highly. For each of our three claims, we will show that we can decrease the Squared Kemeny cost when moving $x_1$ further down if its not already in the permitted interval. Throughout the proofs of these statements, we will thus assume that $d\leq m-1$, i.e., that $x_1$ is not bottom-ranked in $\rhd_{\mathtt{SqK}}$. For Claims (2) and (3), we will refine this assumption later. 

Now, we first recall that 
\begin{align*}
    C_{\mathtt{SqK}}(\rhd_{\mathtt{SqK}})&=R({\succ_1})\cdot \left(d-1+\Delta(\succ_1, \rhd_{\mathtt{SqK}})\right)^2+R({\succ_2})\cdot \left(d-1+\Delta(\succ_2, \rhd_{\mathtt{SqK}})\right)^2\\
    &\quad+R(\succ_3)\cdot \left(m-d+\Delta(\succ_3, \rhd_{\mathtt{SqK}})\right)^2+R(\succ_4)\cdot \left(m-d+\Delta(\succ_4, \rhd_{\mathtt{SqK}})\right)^2.
\end{align*}

Next, let $\rhd$ denote the ranking derived from $\rhd_{\mathtt{SqK}}$ by moving $x_1$ one position down without reordering any other candidates. This means that $x_1$ is now the $(d+1)$-th best candidate and that $\Delta({\succ},\rhd)=\Delta({\succ},\rhd_{\mathtt{SqK}})$ for all rankings ${\succ}\in\mathcal{R}$. Hence, the cost of $\rhd$ is
\begin{align*}
    C_{\mathtt{SqK}}(\rhd)&=R({\succ_1})\cdot \left(d+\Delta(\succ_1, \rhd_{\mathtt{SqK}})\right)^2
    +R({\succ_2})\cdot \left(d+\Delta(\succ_2, \rhd_{\mathtt{SqK}})\right)^2\\
    &\quad+R(\succ_3)\cdot \left(m-d-1+\Delta(\succ_3, \rhd_{\mathtt{SqK}})\right)^2
    +R(\succ_4)\cdot \left(m-d-1+\Delta(\succ_4, \rhd_{\mathtt{SqK}})\right)^2.
\end{align*}

We aim to show that $C_{\mathtt{SqK}}(\rhd_{\mathtt{SqK}})-C_{\mathtt{SqK}}(\rhd)>0$. This means that $\rhd_{\mathtt{SqK}}$ cannot be chosen by the Squared Kemeny rule as $\rhd$ has a lower cost. Based on simple calculus, we infer that
{\thinmuskip=1mu
\medmuskip=2mu
\thickmuskip=3mu
\begin{align*}
    C_{\mathtt{SqK}}(\rhd_{\mathtt{SqK}})-C_{\mathtt{SqK}}(\rhd)&=R({\succ_1})\cdot (-2d-2\Delta(\succ_1,\rhd_{\mathtt{SqK}})+1)
    +R({\succ_2})\cdot (-2d-2\Delta(\succ_2,\rhd_{\mathtt{SqK}})+1)\\
    &\quad+R(\succ_3)\cdot (2(m-d) +2\Delta(\succ_3,\rhd_{\mathtt{SqK}})-1)
    +R(\succ_4)\cdot (2(m-d) +2\Delta(\succ_4,\rhd_{\mathtt{SqK}})-1).
\end{align*}
}

Next, we note that $\Delta(\succ_1,\rhd_{\mathtt{SqK}})+\Delta(\succ_2,\rhd_{\mathtt{SqK}})={m-1\choose 2}$ and $\Delta(\succ_3,\rhd_{\mathtt{SqK}})+\Delta(\succ_4,\rhd_{\mathtt{SqK}})={m-1\choose 2}$, because $\succ_1$ and $\succ_2$ (resp. $\succ_3$ and $\succ_4$) order the candidates in $\{x_2,\dots,x_m\}$ inverse to each other. Furthermore, using the definition of $R$, we derive that 
\begin{align*}
    C_{\mathtt{SqK}}(\rhd_{\mathtt{SqK}})-C_{\mathtt{SqK}}(\rhd)=&\frac{m}{5}\cdot (-2d-2\Delta(\succ_1,\rhd_{\mathtt{SqK}})+1)\\
    &+\frac{1}{3}\left({m\choose 2} - \frac{m}{5} \right)\cdot (-2d-2\Delta(\succ_2,\rhd_{\mathtt{SqK}})+1)\\
    &+\frac{1}{3}\left({m\choose 2} - \frac{m}{5} \right)\cdot (2(m-d) +2\Delta(\succ_3,\rhd_{\mathtt{SqK}})-1)\\
    &+\frac{1}{3}\left({m\choose 2} - \frac{m}{5} \right)\cdot (2(m-d) +2\Delta(\succ_4,\rhd_{\mathtt{SqK}})-1)\\
    =&\frac{4m}{3}\left({m\choose 2} - \frac{m}{5} \right)-2d\cdot{m\choose 2}-\frac{1}{3}{m\choose 2} + \frac{4m}{15}\\
    &-\frac{2m}{5}\Delta(\succ_1,\rhd_\mathtt{SqK})-\frac{2}{3}\left({m\choose 2}-\frac{m}{5}\right)\Delta(\succ_2,\rhd_\mathtt{SqK})\\
    &+\frac{2}{3}\left({m\choose 2}-\frac{m}{5}\right)\Delta(\succ_3,\rhd_\mathtt{SqK})
    +\frac{2}{3}\left({m\choose 2}-\frac{m}{5}\right)\Delta(\succ_4,\rhd_\mathtt{SqK})\\
    =&\frac{4m}{3}\left({m\choose 2} - \frac{m}{5} \right)-2d\cdot{m\choose 2}-\frac{1}{3}{m\choose 2} + \frac{4m}{15}\\
    &-2\left(\frac{1}{3}{m\choose 2} - \frac{4m}{15}\right) \Delta(\succ_2,\rhd_{\mathtt{SqK}}) - \frac{2m}{5}{m-1\choose 2}\\
    &+2\left(\frac{1}{3}{m\choose 2} - \frac{m}{15}\right){m-1\choose 2}\\
    =&\frac{4m}{3}\left({m\choose 2} - \frac{m}{5} \right)-2d\cdot{m\choose 2}-\frac{1}{3}{m\choose 2} + \frac{4m}{15}\\
    &+2\left(\frac{1}{3}{m\choose 2} - \frac{4m}{15}\right)\left({m-1\choose 2}-\Delta(\succ_2,\rhd_{\mathtt{SqK}})\right)
\end{align*}

In the first equality, we substitute the definition of $R(\succ_i)$ for $i\in \{1,2,3,4\}$. In the next equality, we rearrange the terms to isolate the terms that depend on $\Delta$ and simplify the other terms as much as possible. In the third inequality, we then use  $\frac{1}{3}({m\choose 2}-\frac{m}{5})-\frac{m}{5}=\frac{1}{3}{m\choose 2}-\frac{4m}{15}$ and $\Delta({\succ_1},\rhd_{\mathtt{SqK}})+\Delta({\succ_2},\rhd_{\mathtt{SqK}})=\Delta(\succ_3,\rhd_{\mathtt{SqK}})+\Delta(\succ_4,\rhd_{\mathtt{SqK}})={m-1\choose 2}$ as well as $\Delta(\succ_3,\rhd)$ and $\Delta(\succ_4,\rhd)$. Finally, the last line follows by rearranging our terms. 

We now process with a case distinction with respect to $m$ and $\Delta(\rhd,\rhd^*)$ to prove our three claims. To this end, we note that $\Delta(\succ_2,\rhd^*)=0$ as $\succ_2$ and $\rhd^*$ agree on the order ov $\{x_2,\dots, x_m\}$. This means that $\Delta(\rhd_{\mathtt{SqK}},\rhd^*)=\Delta(\succ_2,\rhd_{\mathtt{SqK}})$.\medskip

\emph{Claim (2):} We assume that $m\in \{5,6\}$ and $\Delta(\succ_2,\rhd_\mathtt{SqK})=\Delta(\rhd_\mathtt{SqK},\rhd^*)\leq \frac{1}{2}{m-1\choose 2}$ and aim to show that $C_{\mathtt{SqK}}(\rhd_\mathtt{SqK})-C_{\mathtt{SqK}}(\rhd)>0$ if $d\leq \frac{2m}{3}$. To this end, we observe that the assumptions that $\Delta(\succ_2,\rhd_\mathtt{SqK})\leq \frac{1}{2}{m-1\choose 2}$ and $d\leq \frac{2m}{3}$ imply that 
\begin{align*}
    &C_{\mathtt{SqK}}(\rhd_\mathtt{SqK})-C_{\mathtt{SqK}}(\rhd)\\
    &\geq \frac{4m}{3}\left({m\choose 2} - \frac{m}{5} \right)-\frac{4m}{3}\cdot{m\choose 2}-\frac{1}{3}{m\choose 2} + \frac{4m}{15}
    +\left(\frac{1}{3}{m\choose 2} - \frac{4m}{15}\right)\cdot {m-1\choose 2}\\
    &= \left(\frac{1}{3}{m\choose 2} - \frac{4m}{15}\right)\cdot {m-1\choose 2}+\frac{4m}{15} - \frac{1}{3}{m\choose 2}-\frac{4m^2}{15}.
\end{align*}

Since $m\geq 5$, it holds that ${m-1\choose 2}\geq 6$, so our formula further simplifies to
\begin{align*}
    C_{\mathtt{SqK}}(\rhd_\mathtt{SqK})-C_{\mathtt{SqK}}(\rhd)&\geq 2 {m\choose 2} - \frac{20m}{15} - \frac{1}{3}{m\choose 2}-\frac{4m^2}{15}
\end{align*}

Finally, for $m=5$, this term evaluates $2{5\choose 2}-\frac{100}{15}-\frac{1}{3}{m\choose 2}-\frac{100}{15}=\frac{10}{3}$. Moreover, for $m=6$, we derive that  $2{6\choose 2} - \frac{120}{15} - \frac{1}{3}{6\choose 2}-\frac{144}{15}=\frac{375}{15}-\frac{266}{15}>0$. Hence, in both cases, $\rhd$ has a lower cost than $\rhd_\mathtt{SqK}$, contradicting that $\rhd_\mathtt{SqK}$ is chosen by the Squared Kemeny rule.\medskip

\emph{Claim (3):} Next, we assume that $m\geq 7$ and $\Delta(\succ_2,\rhd_\mathtt{SqK})=\Delta(\rhd_\mathtt{SqK},\rhd^*)\leq \frac{1}{2}{m-1\choose 2}$. This time, our goal is to show that $C_{\mathtt{SqK}}(\rhd_\mathtt{SqK})-C_{\mathtt{SqK}}(\rhd)>0$ if $d\leq \frac{2m}{3}+1$. Analogous to Claim (2), we derive that 
\begin{align*}
    &C_{\mathtt{SqK}}(\rhd_\mathtt{SqK})-C_{\mathtt{SqK}}(\rhd)\\
    &\geq \frac{4m}{3}\left({m\choose 2} - \frac{m}{5} \right)-\frac{4m}{3}{m\choose 2}-2{m\choose 2}-\frac{1}{3}{m\choose 2} + \frac{4m}{15}
    +\left(\frac{1}{3}{m\choose 2} - \frac{4m}{15}\right) {m-1\choose 2}\\
    &= \left(\frac{1}{3}{m\choose 2} - \frac{4m}{15}\right)\cdot {m-1\choose 2}+\frac{4m}{15} - 2{m\choose 2} - \frac{1}{3}{m\choose 2}-\frac{4m^2}{15}.
\end{align*}

Furthermore, it holds that ${m-1\choose 2}\geq 15$ as $m\geq 7$, which implies that
\begin{align*}
    C_{\mathtt{SqK}}(\rhd_\mathtt{SqK})-C_{\mathtt{SqK}}(\rhd)&\geq 5{m\choose 2} - 4m+\frac{4m}{15} - 2{m\choose 2} - \frac{1}{3}{m\choose 2}-\frac{4m^2}{15}.\\
    &=\frac{8}{3}\cdot \frac{m(m-1)}{2} -4m +\frac{4m}{15}-\frac{4m^2}{15}\\
    &\geq \frac{16}{15}m^2 - 6m\\
    &>0.
\end{align*}

The last inequality here use the fact that $m\geq 7$. This proves again that $C_{\mathtt{SqK}}(\rhd)-C_{\mathtt{SqK}}(\rhd')$, so $d>\frac{2m}{3}+1$ in this case.\medskip

\emph{Claim (6):} Finally, we suppose that $m\geq 5$ is arbitrary and that $\Delta(\rhd_\mathtt{SqK},\rhd^*)=0$. In this case, we will show that $d=m$. To this end, we assume that $d\leq m-1$ and show that $C_{\mathtt{SqK}}(\rhd_\mathtt{SqK})>C_{\mathtt{SqK}}(\rhd)$. Our assumptions imply that 
\begin{align*}
    &C_{\mathtt{SqK}}(\rhd_{\mathtt{SqK}})-C_{\mathtt{SqK}}(\rhd)\\
    &\geq \frac{4m}{3}\left({m\choose 2} - \frac{m}{5} \right)-2m{m\choose 2}+2{m\choose 2} -\frac{1}{3}{m\choose 2} + \frac{4m}{15}
    +2\left(\frac{1}{3}{m\choose 2} - \frac{4m}{15}\right) {m-1\choose 2}\\
    &= \left(\frac{2}{3}{m\choose 2} - \frac{8m}{15}\right)\cdot {m-1\choose 2}-\frac{2m}{3}{m\choose 2} +2{m\choose 2}- \frac{4m^2}{15}+\frac{4m}{15} - \frac{1}{3}{m\choose 2}.
\end{align*}

Now, for $m=5$, this term evaluates to 
\begin{align*}
    \left(\frac{2}{3}\cdot 10 - \frac{8\cdot 5}{15}\right)\cdot 6-\frac{2\cdot 5}{3}\cdot 10 +2\cdot 10- \frac{4\cdot 5^2}{15}+\frac{4\cdot 5}{15} - \frac{1}{3}\cdot 10 = 44 - \frac{126}{3} =2.
\end{align*}

Further for $m=6$, we get that
\begin{align*}
\left(\frac{2}{3}\cdot 15 - \frac{8\cdot 6}{15}\right)\cdot 10-\frac{2\cdot 6}{3}\cdot 15+2\cdot 15- \frac{4\cdot 6^2}{15}+\frac{4\cdot 6}{15} - \frac{1}{3}\cdot 15 = 98 - 73  =25.
\end{align*}

Finally, for $m\geq 7$, we observe that $2{m\choose 2} - \frac{4m^2}{15}+\frac{4m}{15}-\frac{1}{3}{m\choose 2}>0$. Hence, we have that 
\begin{align*}
    C_{\mathtt{SqK}}(\rhd_\mathtt{SqK})-C_{\mathtt{SqK}}(\rhd)&> \left(\frac{2}{3}{m\choose 2} - \frac{8m}{15}\right)\cdot {m-1\choose 2}-\frac{2m}{3}{m\choose 2}\\
    &>\frac{2}{3}\cdot {m\choose 2}\cdot {m-1\choose 2}-\frac{6m}{5}{m\choose 2}. 
\end{align*}
For the second inequality, we replace the term $-\frac{8}{15}m {m-1\choose 2}$ with $-\frac{8}{15}m {m\choose 2}$. Further, it holds that $\frac{2}{3}{m-1\choose 2}-\frac{6m}{5}>0$ if $m\geq 7$. Specifically, for $m=7$, this can be straightforwardly verified and the term is increasing in $m$ when $m\geq 7$. Hence, it holds for all $m\geq 5$ that $C_{\mathtt{SqK}}(\rhd_\mathtt{SqK})>C_{\mathtt{SqK}}(\rhd)$ if $d\leq m-1$ and $\Delta(\rhd_\mathtt{SqK},\rhd^*)=0$, which shows that $d=m$ under these assumptions.
\medskip

\textbf{Proofs of Claims (4) and (5):} Finally, we will show that, when $x_1$ is placed low in $\rhd_\mathtt{SqK}$, then $\Delta(\rhd_\mathtt{SqK},\rhd^*)=0$. To this end, assume for contradiction that this is not true and let $\rhd$ denote the ranking derived from $\rhd_\mathtt{SqK}$ by ordering all candidates in $\{x_2,\dots, x_m\}$ according to $\rhd^*$ without changing the position of $x_1$. We will again show that $C_{\mathtt{SqK}}(\rhd)>C_{\mathtt{SqK}}(\rhd')$. 

We first consider the cost caused by $\succ_3$ and $\succ_4$ for an arbitrary ranking $\rhd'$ that puts $x_1$ at position $d$. Because $\succ_3$ and $\succ_4$ disagree on the order over all candidates in $\{x_2,\dots, x_m\}$, it holds for all rankings $\rhd'$ with $1+|\{x\in C\setminus \{x_1\}\colon x\rhd' x_1\}|=d$ that the cost caused by $\succ_3$ and $\succ_4$ is
\begin{align*}
    R(\succ_3)\cdot (m-d+\Delta(\succ_3,\rhd'))^2+R(\succ_4)\cdot (m-d+{m-1\choose 2}-\Delta(\succ_3,\rhd'))^2.
\end{align*}

Because $R(\succ_3)=R(\succ_4)$, this is equivalent to 
\begin{align*}
&R(\succ_3)\cdot \left(\left(m-d+\Delta(\succ_3,\rhd')\right)^2+ \left(m-d+{m-1\choose 2}-\Delta(\succ_3,\rhd')\right)^2\right).\\
=&R(\succ_3)\cdot \bigg((m-d)^2 + 2(m-d)\Delta(\succ_3,\rhd') + \Delta(\succ_3,\rhd')^2+(m-d)^2\\
&+2(m-d)\left({m-1\choose2 }-\Delta(\succ_3,\rhd')\right)+\left({m-1\choose2 }-\Delta(\succ_3,\rhd')\right)^2\bigg)\\
=&R(\succ_3) \bigg(2(m-d)^2 + 2(m-d){m-1\choose 2}+2\Delta(\succ_3,\rhd')^2 - 2\Delta(\succ_3,\rhd'){m-1\choose 2}+{m-1\choose 2}^2\bigg)
\end{align*}

By considering the first order condition with respect to $\Delta(\succ_3,\rhd')$, it is easy to see that, for every fixed $d$, this term is minimized when $\Delta(\succ_3,\rhd')=\frac{1}{2}{m-1\choose 2}$. Since quadratic functions grow symmetrically from their minimum, this means that the cost caused by $\succ_3$ and $\succ_4$ is minimal if $\Delta(\succ_3,\rhd')=\lfloor\frac{1}{2}{m-1\choose 2}\rfloor$ or $\Delta(\succ_3,\rhd')=\lceil\frac{1}{2}{m-1\choose 2}\rceil$. We finally note that $\Delta(\succ_3,\rhd)=\Delta(\succ_3,\rhd^*)=\lfloor\frac{1}{2}{m-1\choose 2}\rfloor$. The first equality here follows because $\rhd$ and $\rhd^*$ agree on the order of the candidates $\{x_2,\dots x_m\}$ by definition. The second equality holds because $\succ_3$ is chosen such that it agrees with $\succ_1$ on exactly $\lfloor\frac{1}{2}{m-1\choose 2}\rfloor$ pairs over $\{x_2,\dots, x_m\}$, and because $\rhd^*$ and $\succ_1$ order these candidates exactly inversely. 
This means that the cost caused by $\succ_3$ and $\succ_4$ is weakly less for $\rhd$ than for $\rhd_{\mathtt{SqK}}$ because both rankings put $x_1$ at the same rank $d$ and $\rhd$ minimizes the cost of $\succ_3$ and $\succ_4 $ among all such rankings.

Next, we turn to $\succ_1$ and $\succ_2$. The cost caused by these rankings for a ranking $\rhd'$ that puts $x_1$ at rank $d$ is 

{\medmuskip=2mu
\thickmuskip=3mu
\begin{align*}
    &R({\succ_1})\cdot \bigg(d-1+\Delta(\succ_1,\rhd')\bigg)^2+R({\succ_2})\cdot \bigg(d-1+{m-1\choose 2}-\Delta(\succ_1,\rhd')\bigg)^2\\
    &=R({\succ_1})\cdot \bigg((d-1)^2+2(d-1)\Delta(\succ_1,\rhd_1)+\Delta(\succ_1,\rhd_1)^2\bigg)\\
    &\quad+R({\succ_2})\cdot \bigg((d-1)^2+2(d-1)({m-1\choose 2}-\Delta(\succ_1,\rhd_1)) + {m-1\choose 2}^2-2{m-1\choose 2}\Delta(\succ_1,\rhd_1)+\Delta(\succ_1,\rhd_1)^2\bigg).
\end{align*}
}

We next consider the function $f$ that interprets the above term as a function in $\Delta(\succ_1,\rhd_1)$ and ignores all terms that are independent of this swap distance. Specifically, 
\[f(x)=R({\succ_1})\cdot \bigg(2(d-1)x+x^2\bigg)
    +R({\succ_2})\cdot \bigg(-2(d-1)x - 2{m-1\choose 2}x+x^2\bigg)\]

We next aim to analyze the minimum of $f(x)$, which then gives insight into the optimal swap distance for our above expression. To this end, we first note that the second derivative of $f$ is a positive constant, so the value of $f$ is strictly decreasing until we reach the minimum. Next, we compute the derivative of $f$: 
\begin{align*}
    f'(x)=R({\succ_1})\cdot \bigg(2(d-1) + 2x\bigg) + R({\succ_2})\cdot \bigg(-2(d-1)-2{m-1\choose2}+2x\bigg).
\end{align*}

We aim to show that $f'({m-1\choose 2})\leq 0$. This implies that the unique optimal value of $\Delta(\succ_1,\rhd')$ is ${m-1\choose 2}$ because $\Delta(\succ_1,\rhd')\leq {m-1\choose 2}$ for every ranking $\rhd_1$. To this end, we observe that $f'({m-1\choose 2})=2 R({\succ_1}) {m-1\choose 2} - 2(R({\succ_2})-R({\succ_1}))(d-1)$. We next consider Claims (4) and (5) separately.\medskip

\emph{Claim (4)}: First, we assume that $m\in \{5,6\}$ and $d>\frac{2m}{3}$. Now, if $m=5$, this means that $d>10/3$. Furthermore, as $d$ is an integer, we derive that $d\geq 4$. By using the definition of $R({\succ_1})$ and $R({\succ_2})$, we now compute that 
\begin{align*}
    f'({m-1\choose 2})&=\frac{2m}{5} {m-1\choose 2}  - 2\left(\frac{1}{3}\left({m\choose 2}-\frac{m}{5}\right)-\frac{m}{5}\right)(d-1)\\
    &\leq 2 \cdot 6  - 2
    \cdot\left(\frac{1}{3}\left(10-1\right)-1\right)\cdot 3\\
    &=0
\end{align*}

Similarly, for $m=6$, the condition that $d>\frac{2m}{3}$ means that $d>4$. Using again that $d$ is an integer, we get that $d\geq 5$. Hence, we compute in this case that 
\begin{align*}
    f'({m-1\choose 2})&=\frac{2m}{5} {m-1\choose 2}  - 2\left(\frac{1}{3}\left({m\choose 2}-\frac{m}{5}\right)-\frac{m}{5}\right)(d-1)\\
    &\leq \frac{12}{5}\cdot 10 - 2\cdot \left(\frac{1}{3}(15-\frac{6}{5})-\frac{6}{5}\right)\cdot 4\\
    &=\frac{-16}{5}
\end{align*}

Hence, in both cases, we get that the minimum of $f$ is reached for $x\geq {m-1\choose 2}$. Since $f$ differs from the cost of the ranking $\rhd'$ only in a constant, a ranking $\rhd'$ that puts $x_1$ at position $d$ minimizes the cost caused by $\succ_1$ and $\succ_2$ if it is inverse to $\succ_1$ for the alternatives $\{x_2,\dots,x_m\}$. Since our ranking $\rhd$ satisfies this condition and is the only ranking that satisfies it,
it hence is the unique minimizer for the cost caused by $\succ_1$ and $\succ_2$. Since it also minimizes the cost caused by $\succ_3$ and $\succ_4$, we conclude that $C_{\mathtt{SqK}}(\rhd)<C_{\mathtt{SqK}}(\rhd_{\mathtt{SqK}})$. Put differently, this contradiction means that if $m\in \{5,6\}$ and $d>\frac{2m}{3}$, then $\Delta(\rhd_{\mathtt{SqK}}, \rhd^*)=0$.\medskip

\emph{Claim (5):} As the second case, we suppose that $m\geq 7$ and $d>\frac{2m}{3}+1$. In this case, we get that 
\begin{align*}
    f'({m-1\choose 2})&\leq \frac{2m}{5} {m-1\choose 2}  - 2\left(\frac{1}{3}\left({m\choose 2}-\frac{m}{5}\right)-\frac{m}{5}\right)\cdot \frac{2m}{3}\\
    &=\frac{2m}{5}{m-1\choose2} - \frac{4m}{9}{m\choose 2}+\frac{16m^2}{45}\\
    &=\left(\frac{2m}{5}-\frac{4m}{9}\right){m-1\choose 2}-\frac{4m(m-1)}{9}+\frac{16m^2}{45}\\
    &=-\frac{2m}{45}{m-1\choose 2}-\frac{4m^2}{45}+\frac{4m}{9}\\
    &< 0.
\end{align*}

In the last inequality, we use that $\frac{4m^2}{45}> \frac{4m}{9}$ because $m\geq 7$. Hence, we have also in this case that $f$ is minimized for some value $x\geq {m-1\choose 2}$. Hence, analogous to the last case, we can now conclude that $C_{\mathtt{SqK}}(\rhd)<C_{\mathtt{SqK}}(\rhd_{\mathtt{SqK}})$, which again contradicts the definition of Squared Kemeny. This completes the proof of our last auxiliary claim. 
\end{proof}

\section{Proof of \Cref{thm:propBordaAR}}\label{app:propBordaAR}

\propBordaAR*
\begin{proof}
    Fix a profile $R$ and an arbitrary subprofile $S$ of $R$. Furthermore, let $\rhd=x_1,\dots, x_m$ denote the ranking chosen by \psb and let $b_i({\succ})$ denote the budgets of the input rankings in the $i$-th round. To simplify the notation, we will assume throughout this proof that $\frac{0}{0}=0$. This assumption removes the need to separately discuss rankings $\succ$ with $R({\succ})=0$, which do not have any influence on \psb.
    Our proof will focus on the payments made by the rankings in $S$. We thus define by $b_i^S({\succ})=\frac{S({\succ})}{R({\succ})} b_i({\succ})$ for all rankings ${\succ}\in\mathcal{R}$ and $i\in \{1,\dots, m\}$ the budget of $b_i({\succ})$ that is due to $S$. 
    Moreover, we let $c_i^S({\succ})=b_i^S({\succ})-b^S_{i+1}({\succ})=\frac{S({\succ})}{R({\succ})}(b_i({\succ})-b_{i+1}({\succ}))$ denote the payment made by $\succ$ in the $i$-th round with respect to $S$, and by $C_i^S=\sum_{{\succ}\in\mathcal{R}} c_i^S({\succ})$ denote the total payment made by the subprofile $S$ in step $i$. 

    Now, fix a round $i\in \{1,\dots, m-1\}$ and let $X_i=\{x_i,\dots,x_m\}$. It holds for all ranking ${\succ}\in\mathcal{R}$ that
    \begin{align*}
        b_i({\succ})-b_{i+1}({\succ})&=\min\left(\frac{(m-i)b_i({\succ})u({\succ}, x_i, X_i)}{U(b_i, x_i, X_i)}, b_i({\succ})\right)
        \leq \frac{(m-i)b_i({\succ})u({\succ}, x_i, X_i)}{U(b_i, x_i, X_i)}.
    \end{align*}

    This means that $c_i^S({\succ})\cdot \frac{U(b_i, x_i, X_i)}{m-i}\leq \frac{S({\succ})}{R({\succ})} \cdot b_i({\succ})\cdot  u({\succ}, x_i, \dots,X_i)$ for all rankings ${\succ}\in\mathcal{R}$. Further, we note that $b_i({\succ})\leq b_1({\succ})$ as our budgets are non-increasing and that $b_1({\succ})=R({\succ})\cdot {m\choose 2}$. Hence, we derive that $c_i^S({\succ})\cdot \frac{U(b_i, x_i, X_i)}{m-i}\leq S({\succ})\cdot {m\choose 2} \cdot  u({\succ}, x_i, \{x_i, \dots, x_m\})$. By summing over all rankings, it follows that
    \[C_i^S\cdot \frac{U(b_i, x_i, X_i)}{m-i}\leq \sum_{{\succ}\in\mathcal{R}} S({\succ})\cdot  {m\choose 2} \cdot u({\succ}, x_i, X_i).\]

    Next, we recall that $\sum_{i=1}^{m-1} u({\succ},x_i, X_i)=u({\succ},\rhd)$. Hence, we derive that
    \begin{align*}
        \sum_{i=1}^{m-1} C_i^S\cdot \frac{U(b_i, x_i, X_i)}{m-i}&\leq \sum_{i=1}^{m-1}\sum_{{\succ}\in\mathcal{R}} S({\succ})\cdot  {m\choose 2} \cdot u({\succ}, x_i, X_i)=\sum_{{\succ}\in\mathcal{R}} S({\succ})\cdot  {m\choose 2} \cdot u({\succ}, \rhd)
    \end{align*}

    We will next work towards inferring a lower bound on $U(b_i, x_i, X_i)$. For this, we first recall that $\sum_{{\succ}\in\mathcal{R}} b_1({\succ})=\frac{m(m-1)}{2}$ by definition. Moreover, it holds that $\frac{m(m-1)}{2}=\sum_{j=1}^{m-1} m-j$. Since we decrease the total budget by at most $m-i$ in each round $i$, it follows for all $i\in \{1,\dots, m-1\}$ that 
    \[\sum_{{\succ}\in\mathcal{R}} b_i({\succ})\geq \sum_{j=1}^{m-1} m-j - \sum_{j=1}^{i-1} m-j=\sum_{j=i}^{m-1} m-j=\frac{(m-i)(m-i+1)}{2}.\]

    We next observe that $\sum_{x\in X_i} u({\succ}, x, X_i)=\sum_{j=0}^{m-i} j=\frac{(m-i)(m-i+1)}{2}$ for every ranking ${\succ}\in\mathcal{R}$. Consequently, the Borda score of all candidates in the $i$-th round is 
    \begin{align*}
        \sum_{{\succ}\in\mathcal{R}} \sum_{x\in X_i} b_i({\succ}) u({\succ}, x, X_i)&= \sum_{{\succ}\in\mathcal{R}}  b_i({\succ}) \cdot \frac{(m-i)(m-i+1)}{2} \geq \frac{(m-i)^2(m-i+1)^2}{4}.
    \end{align*}

    Since there are $m-i+1$ candidates remaining in the $i$-th round, this means that the average Borda score is at least $\frac{(m-i)^2(m-i+1)}{4}$. We thus infer that $U(b_i,x_i, X_i)\geq \frac{(m-i)^2(m-i+1)}{4}$ because $x_i$ maximizes the Borda score in the $i$-th round. By substituting this lower bound in our previous inequality, we derive that
    \begin{align*}
        \sum_{i=1}^{m-1} C_i^S\cdot \frac{(m-i)(m-i+1)}{4}\leq \sum_{{\succ}\in\mathcal{R}} S({\succ})\cdot  {m\choose 2} \cdot u({\succ}, \rhd).
    \end{align*}

    Next, we focus on the payments $C_i^S$. For this, let $C^S=\sum_{i\in S} C_i^S$ denote the total payment made by our subprofile $S$ and let $k$ denote the maximal integer such that $C^S\geq \frac{k(k+1)}{2}$. We note that the term $\frac{(m-i)(m-i+1)}{4}$ is decreasing as $i$ increases, so we minimize the left-hand sum if we pay only in late rounds. Moreover, it holds for all $i$ that $C_i^S\leq m-i$ because the total budget reduction in the $i$-th step is upper bounded by this value. Since $\frac{k(k+1)}{2} \leq C< \frac{(k+1)(k+2)}{2}$, we thus minimize our sum when $C_i^S=m-i$ for all $i\in \{m-k,\dots, m-1\}$ and $C_{m-k-1}^S=C^S-\frac{k(k+1)}{2}<\frac{(k+1)(k+2)}{2}-\frac{k(k+1)}{2}=k+1$. For a simple notation, we let $\ell=C^S-\frac{k(k+1)}{2}$ and conclude that 
    \begin{align*}
        &\sum_{i=1}^{m-1} C_i^S\cdot \frac{(m-i)(m-i+1)}{4}\\
        &\geq \sum_{i=m-k}^{m-1} \frac{(m-i)(m-i)(m-i+1)}{4} + \frac{\ell(m-(m-k-1))(m-(m-k-1)+1)}{4}\\
        &=\sum_{i=1}^{k} \frac{i^2(i+1)}{4} + \frac{\ell(k+1)(k+2)}{4}
    \end{align*}

    We will next show that this term is lower bounded by $\frac{C^S(C^S+1)}{4}$. 
    For this, we note that a simple induction shows that $\sum_{i=1}^{k} i^2(i+1)=\frac{k^4}{4}+\frac{5k^3}{6}+\frac{3k^2}{4}+\frac{k}{6}$. Hence, we have that $\sum_{i=1}^{k} i^2(i+1) + \ell(k+1)(k+2)= \frac{k^4}{4}+\frac{5k^3}{6}+\frac{3k^2}{4}+\frac{k}{6}+ \ell(k+1)(k+2)$. Now, if $k=0$, it holds that $\ell=C^S<1$. On the other hand, our sum evaluates to $2\ell=2C^S>C^S (1+C^S)$. Next, suppose that $k\geq1$. In this case, we observe that 
    \begin{align*}
        &\frac{k^4}{4}+\frac{5k^3}{6}+\frac{3k^2}{4}+\frac{k}{6}+\ell (k+1)(k+2)\\
        &=\left(\frac{k^4}{4}+\frac{2k^3}{4} + \frac{k^2}{4} + \ell k(k+1) + \ell(k+1)\right) + \left(\frac{k^3}{3}+\frac{k^2}{2}+\frac{k}{6} + \ell(k+1)\right)
    \end{align*}

    Next, we recall that $\ell<k+1$, so $\ell(k+1)>\ell^2$. Further, it is easy to check that $\frac{k^3}{3}+\frac{k}{6}\geq \frac{k^2}{2}$ and that $k^2\geq \frac{k(k+1)}{2}$. Hence, we further simply our sum to 
    \begin{align*}
        \frac{k^4}{4}+\frac{5k^3}{6}+\frac{3k^2}{4}+\frac{k}{6}+\ell (k+1)(k+2)
        &\geq\left(\frac{k^4}{4}+\frac{2k^3}{4} + \frac{k^2}{4} + \ell k(k+1)) + \ell^2\right) + \left(k^2 +  \ell(k+1)\right)\\
        &\geq\left(\frac{k(k+1)}{2}+\ell\right)^2 + \left(\frac{k(k+1)}{2} + \ell\right)\\
        &=C^S(C^S+1)
    \end{align*}

This proves our lower bound, so we conclude that $\frac{C^S(C^S+1)}{4}\leq \sum_{{\succ}\in\mathcal{R}} S({\succ}) \cdot {m\choose 2} \cdot u({\succ}, \rhd)$. 
Finally, we note that  proof of \Cref{thm:propBordaRP} shows that the total remaining budget of \psb is at most $\frac{3}{4}$. In particular, this means that $C^S\geq \sum_{\succ \in\mathcal{R}} S({\succ})\cdot {m\choose 2}-\frac{3}{4}$. We hence infer that
\[\frac{(\sum_{{\succ}\in\mathcal{R}} S({\succ})\cdot {m\choose 2}-\frac{3}{4})(\sum_{\succ S} S({\succ})\cdot {m\choose 2}+\frac{1}{4})}{4}\leq \sum_{{\succ}\in\mathcal{R}} S({\succ}) \cdot {m\choose 2} \cdot u({\succ}, \rhd).\]

Equivalently, this means that 
\begin{align*}{m\choose 2}\cdot \frac{\sum_{{\succ}\in\mathcal{R}} S({\succ})}{4}-\frac{3}{16} \leq \frac{\sum_{{\succ}\in\mathcal{R}} S({\succ}) \cdot {m\choose 2} \cdot u({\succ}, \rhd)}{\sum_{{\succ}\in\mathcal{R}} S({\succ}){m\choose 2} +\frac{1}{4}}\leq \frac{\sum_{{\succ}\in\mathcal{R}} S({\succ}) \cdot u({\succ}, \rhd)}{\sum_{{\succ}\in\mathcal{R}} S({\succ})}.
\end{align*}

Finally, by noting that $\sum_{{\succ}\in\mathcal{R}} S({\succ})=|S|$, our theorem follows. 
\end{proof}

\section{Proofs Omitted from \Cref{subsec:FB}}\label{app:FB}

\pairpriceable*
\begin{proof}
    Fix a profile $R$ on $m$ candidates and suppose that $\rhd$ is a pair-priceable ranking for~$R$. Moreover, we denote by $S$ an arbitrary subprofile of $R$ and aim to show that $|A(\rhd)\cap\bigcup_{{\succ}\in\mathcal{R}\colon S({\succ})>0} A({\succ})|\geq \lfloor S({\succ})\cdot {m\choose 2}\rfloor$. To this end, let $\pi$ denote a payment scheme that verifies the pair-priceability of~$\rhd$. By Condition (4) of pair-priceability, we have that $\sum_{{\succ}\in\mathcal{R}} \sum_{(x_i,x_j)\in A(\rhd)} \pi({\succ}, (x_i,x_j))>{m\choose 2}-1$. Since the total budget of all rankings is ${m\choose 2}$, this implies that 
    \begin{align*}\sum_{{\succ}\in\mathcal{R}\colon S({\succ})>0} \sum_{(x_i,x_j)\in A(\rhd)} \pi({\succ}, (x_i,x_j))>{m\choose 2}\cdot \sum_{{\succ}\in\mathcal{R}\colon S({\succ})>0} R({\succ})-1
    \geq |S|\cdot {m\choose 2}-1.
    \end{align*}
    
    Next, by Conditions (1) and (3), each ranking $\succ$ only pays for pairs $(x_i,x_j)\in A({\succ})$ and we can pay at most $1$ to each such pair. Put differently, this means that the rankings with positive weight in $S$ can only pay for the pairs in $A(\rhd)\cap \bigcup_{{\succ}\in\mathcal{R}\colon S({\succ})>0} A({\succ})$ and at most $1$ for each such pair. Hence, it holds that \[\sum_{{\succ}\in\mathcal{R}\colon S({\succ})>0} \sum_{(x_i,x_j)\in A(\rhd)} \pi({\succ}, (x_i,x_j))\leq |A(\rhd)\cap \bigcup_{{\succ}\in\mathcal{R}\colon S({\succ})>0} A({\succ})|.\]

    By combining our two inequalities, we conclude that 
    \[|S|\cdot {m\choose 2}-1< |A(\rhd)\cap \bigcup_{{\succ}\in\mathcal{R}\colon S({\succ})>0} A({\succ})|.\]

    Finally, since the right side of this inequality is an integer, it follows that $\lfloor|S|\cdot {m\choose 2}\rfloor\leq |A(\rhd)\cap \bigcup_{{\succ}\in\mathcal{R}\colon S({\succ})>0} A({\succ})|$, thus proving that sPJR holds.
\end{proof}

\flowBordaPR*
\begin{proof}
    Fix some profile $R$ and let $\rhd=x_1,\dots,x_m$ denote the ranking selected by \fb. Further, let $b_i({\succ})$ denote the budgets of the ranking $\succ$ in the $i$-th round of \fb and let $X_i=\{x_i,\dots, x_m\}$ denote the remaining candidates. Moreover, let $G_{x_i}$ denote the flow network of \fb in the $i$-th round and let $f_i$ be the maximum flow chosen in this network. We will show that the payment scheme~$\pi$ given by $\pi({\succ}, (x_i,x_j))=f_i(v_{\succ},v_{x_j})$ if $x_i\succ x_j$ and $\pi({\succ}, (x_i,x_j))=0$ otherwise satisfies all conditions of pair-priceability.\medskip

    \emph{Condition (1):} We first note that $\pi({\succ}, (x_i,x_j))=0$ by definition and that there is no edge fro $v_{\succ}$ to $x_j$ in $G_{x_i}$ in this case. Hence, the claim holds in this case. Furthermore, if $x_i\succ x_j$, it holds that $f_i(v_{\succ}, v_{x_j})\leq 1$ because the capacity from $v_{x_j}$ to the source is $1$. Hence, it follows for all rankings $\succ$ and pairs of candidates $(x_i,x_j)$ that $\pi({\succ}, (x_i,x_j))\leq 1=u({\succ}, x_i, \{x_i, x_j\})$. This proves that $0\leq \pi({\succ}, (x_i,x_j))\leq u({\succ},x_i, \{x_i,x_j\})$ for all ${\succ}\in\mathcal{R}$ and $(x_i,x_j)\in A(\rhd)$. \medskip
    
    \emph{Condition (2):} This condition is satisfied because $b_1({\succ})=R({\succ})\cdot{m\choose 2}$ for all ${\succ}\in\mathcal{R}$. Moreover, we note that $b_i({\succ})\geq 0$ for all rankings $\succ$ and rounds $i$ because it is never possible to decrease the budget of a ranking by more than $b_i({\succ})$. This is encoded in our flow network as the capacity from the source to a ranking vertex $v_{\succ}$ is $b_i({\succ})$. Next, using flow conversation, we note that for all $i\in \{1,\dots, m-1\}$ and ${\succ}\in\mathcal{R}$ that $f(s,v_{\succ})=\sum_{x_j\in X_i\setminus \{x_i\}\colon x_i\succ x_j} f(v_{\succ},v_{x_j})=\sum_{x_j\in X_i\setminus \{x_i\}} \pi({\succ}, (x_i,x_j))$. Further, by definition of \fb, it holds that $f_i(s,v_{\succ})=b_i({\succ})-b_{i+1}({\succ})$. By combining these insights, it follows for all ${\succ}\in\mathcal{R}$ that
    \begin{align*}
     \sum_{(x_i,x_j)\in A(\rhd)} \pi({\succ},(x_i,x_j))
    &=\sum_{i=1}^{m-1}\sum_{x_j\in X_i\setminus \{x_i\}} \pi({\succ},(x_i,x_j))
    =\sum_{i=1}^{m-1} f_i(s,v_{\succ})\\
    &=\sum_{i=1}^{m-1} b_i({\succ})-b_{i+1}({\succ})
    =b_1({\succ})-b_m({\succ})
    \leq {m\choose2}\cdot R({\succ}).
    \end{align*}
    This proves this condition.\medskip

    \emph{Condition (3):} Fix a pair of candidate $(x_i,x_j)\in A(\rhd)$; we will show that $\sum_{{\succ}\in\mathcal{R}} \pi({\succ},(x_i,x_j))\leq 1$. To this end, we note that $\sum_{{\succ}\in\mathcal{R}} \pi({\succ},(x_i,x_j))=\sum_{\succ\in\mathcal{R}\colon x_i\succ x_j} f_i(v_{\succ}, v_{x_j})=f_i(v_{x_j},t)\leq 1$. Here, the first equality follows by the definition of $\pi$, the second one by flow conversation (since $v_{x_j}$ has only a single outwards edge), and the final inequality by the capacity constraint of the edge $(v_{x_j},t)$.\medskip

    \emph{Condition (4):} As the last condition, we need to show that $\sum_{\succ\in\mathcal{R}}\sum_{(x_i,x_j)\in A(\rhd)} \pi({\succ},(x_i,x_j))>{m\choose 2}-1$. Just as for \psb, we will again show that the total remaining budget after the execution of \fb, $\sum_{\succ\in\mathcal{R}} b_m({\succ})$, is at most $0.75$. This proves our claim because, as we have shown for Condition (2), it holds for each ranking $\succ$ that $\sum_{(x_i,x_j)\in A(\rhd)} \pi({\succ}, (x_i,x_j)= b_1({\succ})-b_m({\succ})$ and therefore $\sum_{\succ\in\mathcal{R}}\sum_{(x_i,x_j)\in A(\rhd)} \pi({\succ}, (x_i,x_j)=\sum_{\succ\in\mathcal{R}} b_1({\succ})-b_m({\succ})={m\choose 2}-\sum \sum_{\succ\in\mathcal{R}} b_m({\succ})$.
    
    Now, to prove this claim, we will show for all rounds $i\in \{1,\dots, m-3\}$ that the maximum flow in $G_{x_i}$ has value $m-i$. 
    To this end, we fix such a round $i$ and suppose that our claim is true for all rounds $j<i$. If $i=1$, this assumption is true as there were no previous rounds and it will hold inductively for $i>1$. We first note that the total budget in this round is 
    \[\sum_{{\succ}\in\mathcal{R}} b_i({\succ})=\frac{(m-1)m}{2}-\sum_{j=1}^{i-1} m-j = \sum_{j=1}^{m-i} j = \frac{(m-i)(m-i+1)}{2}.\]

    Next, assume for contradiction that the maximum flow in $G_{x_i}=(V,E,c)$ has a value not equal to $m-i$. We note that no flow in $G_{x_i}$ can have a value of more than $m-i$ because the capacities of all edges pointing to the sink $t$ is $\sum_{y\in X_i\setminus \{x_i\}} c(v_y,t)=m-i$. Hence, our assumption means that the maximum flow has a value strictly less than $m-i$. By the MaxFlow-MinCut equivalence, this means that there is an $(s,t)$-cut $S=(T, V\setminus T)$ in $G_{x_i}$ such that $\sum_{(v,w)\in E\colon v\in T, w\in V\setminus T}c(v,w)<m-i$. It follows from this insight that $S$ does not cut any edge connecting a ranking vertex and a candidate vertex because all of these edges have unbounded capacity. Furthermore, let $Z=\{x_j\in X_i\setminus \{x_i\}\colon v_{x_j}\in T\iff t\in T \}$ denote the set of candidates such that the edge from the corresponding vertex to the sink is not separated by $S$. We note that $Z\neq\emptyset$ because otherwise, $S$ cuts all edges from candidate vertices to the sink and thus has a weight of $m-i$. Next, let $\mathcal {\bar R}$ denote the set of rankings $\succ$ such that there is an edge from $v_{\succ}$ to a candidate vertex $v_y$ with $y\in Z$. All edges from the source to the ranking vertices $v_{\succ}$ with ${\succ}\in \mathcal {\bar R}$ have been cut as $S$ does otherwise not disconnect $s$ and $t$. Furthermore, the total cost for cutting these edges is less than $|Z|$ because $S$ would have a value of at least $m-i$ otherwise. Put differently, there is a set of candidates $Z$ such that the rankings in $\mathcal {\bar R}=\{{\succ}\in\mathcal{R}\colon \exists y\in Z\colon x_i\succ y\}$ have a total budget of less than $|Z|$. 
        
    We will show that this observation contradicts that $x_i$ maximizes $U(b_i,x,X_i)$ when $i\geq m-3$. To this end, we define by $\mathcal R^{Z\succ x_i}$ the set of rankings that prefer all candidates in $Z$ to $x_i$. Since $\mathcal R^{Z\succ x_i}=\mathcal{R}\setminus \bar{\mathcal{R}}$, our previous insight implies that \[\sum_{{\succ}\in\mathcal{R}^{Z\succ x}} b_i({\succ})>\frac{(m-i)(m-i+1)}{2}-z,\] where we use $z=|Z|$. In particular, note here that  $\sum_{{\succ}\in\mathcal{R}^{Z\succ x_i}} b_i({\succ})+\sum_{{\succ}\in\mathcal{\bar R}} b_i({\succ})=\frac{(m-i)(m-i+1)}{2}$ by assumption. Further, because all rankings in $\mathcal{R}^{Z\succ x_i}$ prefer all candidates in $Z$ to $x_i$, it holds that $\sum_{y\in Z} u({\succ},y,X_i)-u({\succ},x_i,X_i)\geq \sum_{j=1}^z j=\frac{z(z+1)}{2}$ for all ${\succ}\in\mathcal{R}^{Z\succ x_i}$. By summing over all ${\succ}\in\mathcal{R}^{Z\succ x_i}$, we hence conclude that 
    \begin{align*}\sum_{y\in Z} \sum_{{\succ}\in\mathcal{R}^{Z\succ x_i}} b_i({\succ}) \left(u({\succ}, x_i, X_i)-u({\succ}, y, X_i)\right)&
    \leq -\sum_{{\succ}\in\mathcal{R}^{Z\succ x}}  b_i({\succ}) \frac{z(z+1)}{2}\\
    &<-\left(\frac{(m-i)(m-i+1)}{2}-z\right)\cdot \frac{z(z+1)}{2}.
    \end{align*}

    On the other hand, in the best case, it holds for all rankings ${\succ}\in \mathcal{\bar R}$ with $b_i({\succ})>0$ that $x_i$ is top-ranked and the candidates in $Z$ are bottom-ranked. Since 
 the maximal Borda score with $m-i+1$ candidates is $m-i$, we derive for all such rankings that $\sum_{y\in Z} u({\succ}, x_i, X_i)-u({\succ}, y, X_i)=\sum_{j=0}^{z-1} (m-i-j)=\sum_{j=1}^z (m-i+1-j)=z(m-i+1)-\frac{z(z+1)}{2}$. By summing over all such rankings, we get that 
    \begin{align*}
    \sum_{y\in Z} \sum_{{\succ}\in\mathcal{\bar R}} b_i({\succ}) \left(u({\succ}, x_i, X_i)-u({\succ}, y, X_i)\right)&\leq \sum_{{\succ}\in\mathcal{\bar R}}  b_i({\succ}) \left((m-i+1)z-\frac{z(z+1)}{2}\right)\\
    &<z\cdot \left((m-i+1)z-\frac{z(z+1)}{2}\right).
    \end{align*}
    Finally, by combining our two inequalities, we derive that
    \begin{align*}
        & \sum_{y\in Z} \sum_{{\succ}\in\mathcal{R}} b_i({\succ}) (u({\succ}, x_i, X_i)-u({\succ}, y, X))\\
        &=\sum_{y\in Z} \sum_{{\succ}\in\mathcal{R}^{Z\succ x}} b_i({\succ}) (u({\succ}, x_i, X_i)-u({\succ}, y, X_i)) + \sum_{y\in Z} \sum_{{\succ}\in\mathcal{\bar R}} b_i({\succ}) (u({\succ}, x_i, X_i)-u({\succ}, y, X_i))\\
        &<-\left(\frac{(m-i)(m-i+1)}{2}-z\right)\cdot \frac{z(z+1)}{2}+z\left((m-i+1)z-\frac{z(z+1)}{2}\right)\\
        &=(m-i+1)z^2 - \frac{(m-i)(m-i+1)z(z+1)}{4}\\
        &=(m-i+1)z\left(z-\frac{(m-i)(z+1)}{4}\right)
    \end{align*}
    
    Now, if $i\leq m-4$ and thus $m-i\geq 4$, it is clear that $z-\frac{(m-i)(z+1)}{4}<0$, which shows that $x_i$ cannot be the Borda winner. Next, if $i=m-3$, we have that $z-\frac{(m-i)(z+1)}{4}=z-\frac{3(z+1)}{4}=\frac{z}{4}-\frac{3}{4}\leq 0$ because $z\leq 3$ if only four candidates are remaining. Combined with our previous inequality, we get again a contradiction to the fact that $x_i$ maximizes the Borda score. This shows that our initial assumption is wrong and there is indeed a maximum flow of value $m-i$ when $i\in \{1,\dots, m-3\}$. 

    By our analysis so far, we conclude that $\sum_{{\succ}\in\mathcal{R}} b_i({\succ})=3$ if $i=m-2$. Furthermore, we are left with three candidates $X_{m-2}=\{x_{m-2},x_{m-1},x_m\}$ when $i=m-2$. Now, just as for \psb, it can be shown that the average Borda score in this round is $\frac{1}{3}\sum_{\succ\in\mathcal{R}}\sum_{x\in X_{m-2}} b_{m-2}({\succ})u({\succ},x_{m-2}, X_{m-2})=3$. Hence, it holds for the Borda winner $x_{m-2}$ that $U(b_{m-2}, x_{m-2}, X_{m-2})\geq 3$. We will show that this implies that the maximum flow in $G_{x_{m-2}}$ has value of at least $1.5$. If this was not true, we could also find a minimum $(s,t)$-cut $S=(T, V\setminus T)$ of value less than $1.5$. Consequently, this cut cannot disconnect both candidate vertices from the sink as its value would be $2$ otherwise. 
    
    First assume that $S$ disconnects $x_{m-1}$ from the sink but not $x_{m}$ (and note that the case where $x_{m}$ is disconnected but $x_{m-1}$ not is symmetrical). In this case, the cut needs to disconnect all rankings that put $x_{m-2}$ ahead of $x_m$ from the source. Moreover, as the value of the cut is less than $1.5$, we get that $\sum_{\succ\in\mathcal{R}\colon x_{m-2}\succ x_{m}}b_{m-2}({\succ})<0.5$. Lastly, note that the only other rankings that contribute to the Borda score of $x_{m-2}$ are those with $x_{m}\succ x_{m-2}\succ x_{m-1}$, which satisfy that $u({\succ}, x_{m-2}, X_{m-2})=1$. If the total budget of these rankings is $\sum_{\succ\in\mathcal{R}\colon x_{m}\succ x_{m-2}\succ x_{m-1}}b_{m-2}({\succ})\leq 2$, we get that $U(b_{m-2}, x_{m-2}, X_{m-2})\leq 2\cdot \sum_{\succ\in\mathcal{R}\colon x_{m-2}\succ x_m}b_{m-2}({\succ})+ \sum_{\succ\in\mathcal{R}\colon x_m\succ x_{m-2}\succ x_{m-1}}b_{m-2}({\succ})<3$, which contradicts our previous observation. On the other hand, if $\sum_{\succ\in\mathcal{R}\colon x_{m}\succ x_{m-2}\succ x_{m-1}}b_{m-2}({\succ})> 2$, we have that $U(b_{m-2}, x_{m}, X_{m-2})\geq 2\sum_{\succ\in\mathcal{R}\colon x_{m}\succ x_{m-2}\succ x_{m-1}}b_{m-2}({\succ})>4$ but $U(b_{m-2}, x_{m-2}, X_{m-2})\leq 2\cdot \sum_{\succ\in\mathcal{R}\colon x_{m-2}\succ x_m}b_{m-2}({\succ})+ 3- \sum_{\succ\in\mathcal{R}\colon x_{m-2}\succ x_m}b_{m-2}({\succ})<3.5$. Hence, in both cases, we get a contradiction to the fact that $x_{m-2}$ is a Borda winner.

    As the second case, suppose that the cut $S$ disconnects neither $x_{m-1}$ nor $x_{m}$ from the sink. Hence, it needs to disconnect all rankings that do not rank $x_{m-2}$ last among $\{x_{m-2}, x_{m-1}, x_{m}\}$ from the source. This implies that $\sum_{\succ\in\mathcal{R}\colon x_{m-2}\succ x_m\lor x_{m-2}\succ x_{m-1}} b_{m-2}({\succ})<1.5$ since the cut has a value of less than $1.5$. However, this means that $U(b_{m-2}, x_{m-2},X_{m-2})\leq 2\sum_{\succ\in\mathcal{R}\colon x_{m-2}\succ x_m\lor x_{m-2}\succ x_{m-1}} b_{m-2}({\succ})<3$, which again contradicts that $x_{m-2}$ is the Borda winner. Hence, the maximum flow $f_{m-2}$ has a value of at least $1.5$ and the remaining budget for the last round is at most $\sum_{\succ\in\mathcal{R}} b_{m-1}(\succ)\leq 1.5$.
    
    Finally, in the last round, the candidate maximizing $U(b_{m-1}, x, X_{m-1})$ is the winner of the majority vote between $x_{m-1}$ and $x_m$ with respect to $b_{m-1}$. It is thus straightforward that we reduce the total budget by at least half and at most $1$, so the total remaining budget after the execution is at most $\frac{3}{4}$. This proves the last condition of pair-priceability. 
\end{proof}

\flowBordaAR*
\begin{proof}
    Fix a profile $R$ and let $\rhd=x_1,\dots, x_m$ denote the ranking chosen by \fb. 
    We will closely follow the proof of \Cref{thm:propBordaAR} and thus define by $b_i({\succ})$ the budget of ranking $\succ$ in the $i$-th round of the Flow-adjusting Borda rule and let $X_i=\{x_i,\dots, x_m\}$ denote the remaining candidates. In particular, $b_1({\succ})=R({\succ})\cdot{m\choose2}$ for all rankings ${\succ}\in\mathcal{R}$. Furthermore, let $f_i$ denote the maximum flow chosen in the $i$-th round of \fb and define the cost per utility ratio $\rho_i$ by $\max_{{\succ}\in\mathcal{R}} \frac{f_i(s,v_{\succ})}{b_i({\succ})\cdot u({\succ},x_i, \{x_i,\dots, x_m\})}$. Lastly, we will throughout the proof assume that $\frac{0}{0}=0$ to avoid trivial corner cases.

We will show this theorem in two steps. First, we will show that $\rho_i\leq \frac{4}{(m-i)(m-i+1)}$ for all $i\in \{1,\dots,m-3\}$, $\rho_{m-2}\leq 1$, and $\rho_{m-1}\leq 1$. Based on this insight, we will prove the theorem in a second step.\medskip

    \textbf{Step 1:} We start by showing our upper bounds on $\rho_i$. To this end, we fix a round $i$ and let $G_{x_i}=(V,E,c)$ denote the flow network used by \fb in this round. It holds for every ranking $\succ$ that $f_i(s,v_{\succ})\leq b_i({\succ})$ since $c(s,v_{\succ})=b_i({\succ})$. Hence, we derive that $\frac{f_i(s,v_{\succ})}{b_i({\succ})\cdot u({\succ}, x_i, X_i)}\leq \frac{1}{u({\succ}, x_i, X_i)}\leq 1$ for every $\succ\in\mathcal{R}$ with $u({\succ},x_i,X_i)>0$ and $b_i(\succ)>0$. Secondly, if $u({\succ}, x_i, X_i)=0$, then $\succ$ bottom-ranks $x_i$ among $X_i$. Thus, $v_\succ$ has no outgoing edge in $G_{x_i}$ and $f(s,v_{\succ})=0$. Similarly, if $b_i(\succ)=0$, then $f(s,v_{\succ})=0$ by the capacity constraint. In both cases, it follows  $\frac{f_i(s,v_{\succ})}{b_i({\succ})\cdot u({\succ}, x_i, \{x_i,\dots, x_m\})}=0$ by our assumption that $\frac{0}{0}=0$ (or, put differently, we can ignore rankings with no budget or that bottom-rank the current Borda winner as they do not pay anything). This proves that $\rho_i\leq 1$ for all rounds $i$, including $\rho_{m-2}$ and $\rho_{m-1}$. 

Next, we suppose that $i\in \{1,\dots,m-3\}$ and aim to show that $\frac{f_i(s,v_{\succ})}{b_i({\succ})\cdot u({\succ}, x_i, X_i)}\leq \frac{4}{(m-i)(m-i+1)}$. To this end, we will prove that the flow network $G_{x_i}$ admits a flow $f^*_i$ with value $m-i$ such that $f_i^*(s,v_{\succ})\leq \frac{4 b_i({\succ})u({\succ},x_i, X_i)}{(m-i)(m-i+1)}$ for all ${\succ}\in\mathcal{R}$. For this, consider the modified flow network $G_{x_i'}'=(V,E,c')$, which uses the same vertices and edges as $G_{x_i}$ but has different capacities. Specifically, we set $c'(s,v_{\succ})=b_i({\succ})\cdot u({\succ}, x_i, X_i)$ for all ranking vertices, $c'(v_{\succ}, v_x)$ is still unbounded for all rankings $\succ$ and candidates $x$ with $x_i\succ x$, and $c'(v_x,t)=\frac{(m-i)(m-i+1)}{4}$ for all candidates $x\in \{x_{i+1},\dots,x_m\}$. We will show that $G'_{x_i}$ permits a flow $f'_i$ of value $\frac{(m-i)^2(m-i+1)}{4}$. Based on $f'_i$, we will then define the flow $f^*_i$ by $f^*_i(e)=\frac{4f'_i(e)}{(m-i)(m-i+1)}$ for all $e\in E$. We claim that $f^*_i$ is feasible for the original network $G_{x_i}$ and satisfies that $\frac{f_i^*(s,v_{\succ})}{b_i({\succ})\cdot u({\succ},x_i,X_i)}\leq \frac{4}{(m-i)(m-i+1)}$ for all ${\succ}\in\mathcal{R}$. For the feasibility, we observe for every edge $e\in E$ that $f_i^*(e)=\frac{4f'_i(e)}{(m-i)(m-i+1)}\leq \frac{4c'(e)}{(m-i)(m-i+1)}\leq c(e)$. In more detail, for the edges from candidates vertices $v_x$ to the sink $t$, this holds as $c'(v_x,t)=\frac{(m-i)(m-i+1)}{4}$ and $c(v_x,t)=1$. For edges from the source $s$ to ranking vertices $v_{\succ}$, our claim holds as $\frac{4c'(s,v_{\succ})}{(m-i)(m-i+1)}=\frac{4b_i({\succ}) u({\succ},x_i,X_i)}{(m-i)(m-i+1)}\leq b_i({\succ})=c(s,v_{\succ})$. For the second step, we use that $u({\succ},x_i,X_i)\leq m-i$ and that $4\leq m-i+1$ as $i\leq m-3$. Furthermore, it holds by definition that $f_i^*(x,v_{\succ})\leq \frac{4b_i({\succ}) u({\succ},x_i,X_i)}{(m-i)(m-i+1)}$ for all rankings~$\succ$, so $\frac{f_i^*(s,v_{\succ})}{b_i({\succ})\cdot u({\succ},x_i,X_i)}\leq \frac{4}{(m-i)(m-i+1)}$ for all ${\succ}\in\mathcal{R}$.

Now, we assume for contradiction that our modified flow network $G_{x_i}'$ does not permit a flow of value $\frac{(m-i)^2(m-i+1)}{4}$. By the MaxFlow-MinCut equivalence, this means that there is an $(s,t)$-cut $(T, V\setminus T)$ in $G_{x_i}'$ whose total weight is less than $\frac{(m-i+1)(m-i)^2}{4}$. Let $Z$ denote the set of candidates for which the candidate vertex is not separated from the sink by $(T, V\setminus T)$, i.e., $Z$ is the set of candidates $x$ such that $v_x\in T$ if and only if $t\in T$. Since there are $(m-i)$ candidate vertices and all edges $(v_x,t)$ have a capacity of $\frac{(m-i)(m-i+1)}{4}$, we derive that $Z\neq\emptyset$ as $(T, V\setminus T)$ would have a value of at least $\frac{(m-i+1)(m-i)^2}{4}$ otherwise. Moreover, $(T, V\setminus T)$ cannot disconnect any edge from a ranking vertex to a candidate vertex as these have unbounded capacity. Finally, let $\mathcal{\bar R}$ denote the set of rankings $\succ$ such that $(v_{\succ}, v_x)\in E$ for a candidate $x\in Z$. All edges from the source to the ranking vertices $v_{\succ}$ for ${\succ}\in\mathcal{\bar R}$ must be disconnected as there is otherwise still a path from $s$ to $t$ in $G_{x_i}'$. Moreover, the total capacities of these edges is less than $|Z|\frac{(m-i)(m-i+1)}{4}$. Otherwise, the weight of $(T, V\setminus T)$ is at least $\frac{(m-i+1)(m-i)^2}{4}$ because we cut $(m-i-|Z|)$ edges from candidate vertices to the sink, each of which has a weight of $\frac{(m-i+1)(m-i)}{4}$. 

    In summary, this analysis shows that there is a set of candidates $Z\subseteq \{x_{i+1},\dots,x_m\}$ such that $\sum_{{\succ}\in \mathcal{\bar R}} b_i({\succ})u({\succ},x_i,X_i)<z\frac{(m-i+1)(m-i)}{4}$, where $z=|Z|$ and $\mathcal{\bar R}=\{{\succ}\in\mathcal{R}\colon \exists y\in Z\colon x_i\succ y\}$. We will show that this contradicts the fact that $x_i$ maximizes the $U(b_i,x_i, X_i)$. To this end, let $\varphi=\sum_{{\succ}\in \mathcal{\bar R}} b_i({\succ})$ and first assume that $\varphi\geq \frac{(m-i+1)(m-i)}{4}$. Because each ranking ${\succ}\in\mathcal{R}\setminus \mathcal{\bar R}$ prefers all candidates in $Z$ to $x_i$, $x_i$ receives a score of at most $m-z-i$ from each of these rankings. Since we have shown in the proof of \Cref{thm:flowBordaPR} that the total remaining budget is $\sum_{{\succ}\in\mathcal{R}} b_i({\succ})=\frac{(m-i)(m-i+1)}{2}$ when $i\in \{1,\dots, m-2\}$, this means that $\sum_{{\succ}\in\mathcal{\mathcal{R}\setminus \bar R}} b_i({\succ})u({\succ},x_i,X_i)\leq (\frac{(m-i)(m-i+1)}{2}-\varphi)(m-z-i)$. Combined with the observation that 
    ${\sum_{{\succ}\in \mathcal{\bar R}} b_i({\succ})u({\succ},x_i,X_i)<z\frac{(m-i)(m-i+1)}{4}}$, we now derive that
    \begin{align*}
        \sum_{{\succ}\in\mathcal{R}} b_i({\succ})u({\succ},x_i,X_i)
        &=\sum_{{\succ}\in\mathcal{R}\setminus \mathcal{\bar R}} b_i({\succ})u({\succ},x_i,X_i)+\sum_{{\succ}\in \mathcal{\bar R}} b_i({\succ})u({\succ},x_i,X_i)\\
        &<\left(\frac{(m-i)(m-i+1)}{2}-\varphi\right) (m-z-i) + z\frac{(m-i)(m-i+1)}{4}\\
        &\leq \frac{(m-i)^2(m-i+1)}{4}.
    \end{align*}

    In the last step, we use that $\varphi\geq \frac{(m-i)(m-i+1)}{4}$ and thus $\frac{(m-i)(m-i+1)}{2}-\varphi\leq \frac{(m-i)(m-i+1)}{4}$. We next note that the average Borda score with respect to $b_i$ is 
    \begin{align*}
        \frac{1}{m-i+1}\sum_{x_j\in X_i} \sum_{{\succ}\in\mathcal{R}} b_i({\succ})u({\succ},x_j,X_i)&=\frac{1}{m-i+1} \sum_{{\succ}\in\mathcal{R}} b_i({\succ})\sum_{j=0}^{m-i} j
        =\frac{(m-i)^2(m-i+1)}{4}.
    \end{align*} 
    
    However, this means that there is a candidate $x_j$ with $\sum_{{\succ}\in\mathcal{R}} b_i({\succ})u({\succ},x_j,X_i)\geq \frac{(m-i)^2(m-i+1)}{4}>U({\succ},x_i,X_i)$, which contradicts that $x_i$ is selected by \fb in this step. 

    As the second case, suppose that $\varphi<\frac{(m-i)(m-i+1)}{4}$. Since $x_i$ maximizes the total Borda score with respect to $x_i$, it holds that 
    \begin{align*}
        0&\leq \sum_{{\succ}\in \mathcal{\bar R}} \sum_{y\in Z}b_i({\succ})(u({\succ},x_i,X_i)-u({\succ},y,X_i))+\sum_{{\succ}\in \mathcal{R}\setminus \mathcal{\bar R}} \sum_{y\in Z} b_i({\succ})(u({\succ},x_i,X_i)-u({\succ},y,X_i)). 
    \end{align*}

    Equivalently, this means that 
    \begin{align*}
  \sum_{{\succ}\in \mathcal{\bar R}} b_i({\succ})u({\succ},x_i,X_i)&\geq \frac{1}{z}\sum_{{\succ}\in \mathcal{\bar R}} \sum_{y\in Z}b_i({\succ}) u({\succ},y,X_i) \\
    &+ \frac{1}{z}\sum_{{\succ}\in \mathcal{R}\setminus \mathcal{\bar R}} \sum_{y\in Z} b_i({\succ})(u({\succ},y,X_i)-u({\succ},x_i,X_i)).
    \end{align*}

    Next, since $\sum_{y\in Z} u({\succ}, y, X_i)\geq \sum_{j=0}^{z-1} j$ for every ranking $\succ$, it follows that
    \[\frac{1}{z}\sum_{{\succ}\in \mathcal{\bar R}} \sum_{y\in Z}b_i({\succ}) u({\succ},y,X_i)\geq \frac{1}{z}\sum_{{\succ}\in \mathcal{\bar R}} b_i({\succ}) \sum_{j=0}^{z-1} j = \varphi \frac{z-1}{2}.\] 
    
    Further, because every ranking ${\succ}\in\mathcal{R}\setminus\mathcal{\bar R}$ ranks all candidates in $Z$ ahead of $x_i$, we infer that $\sum_{y\in Z} u({\succ},y,X_i)-u({\succ},x_i,X_i)\geq \sum_{j=1}^z j=\frac{z(z+1)}{2}$ for each ${\succ}\in\mathcal{R}\setminus\mathcal{\bar R}$. Based on this insight and the fact that $\sum_{{\succ}\in\mathcal{R}} b_i({\succ})=\frac{(m-i)(m-i+1)}{2}$, we compute that
    \begin{align*}
        \frac{1}{z}\sum_{{\succ}\in \mathcal{R}\setminus \mathcal{\bar R}} \sum_{y\in Z} b_i({\succ})(u({\succ},y,X_i)-u({\succ},x_i,X_i))
        &\geq \frac{1}{z}\sum_{{\succ}\in\mathcal{R}\setminus \mathcal{\bar R}} b_i({\succ}) \sum_{j=1}^z j\\
        &=\left(\frac{(m-i)(m-i+1)}{2}-\varphi\right)\frac{z+1}{2}
    \end{align*}

    Putting our inequalities together, this means that 
    \begin{align*}
        \sum_{{\succ}\in \mathcal{\bar R}} b_i({\succ})u({\succ},x_i,X_i)
        &\geq \varphi \frac{z-1}{2}+\left(\frac{(m-i)(m-i+1)}{2}-\varphi\right)\cdot \frac{z+1}{2}\\
        &=\frac{(m-i)(m-i+1)}{2}\cdot \frac{z-1}{2}+\left(\frac{(m-i)(m-i+1)}{2}-\varphi\right)
        \\
        &\geq \frac{(m-i)(m-i+1)}{2}\cdot \frac{z-1}{2} + \frac{(m-i)(m-i+1)}{4}\\
        &=\frac{(m-i)(m-i+1)}{4}\cdot z
    \end{align*}

    Here, the inequality in the third step follows by using that $\varphi\leq \frac{(m-i)(m-i+1)}{4}$. This directly disproves the assumption that $\sum_{{\succ}\in \mathcal{\bar R}} b_i({\succ})u({\succ},x_i,X_i)<\frac{(m-i)(m-i+1)}{4}z$, so we also showed our claim in this case.\medskip

    \textbf{Step 2:} We are now ready to prove the theorem and thus fix an arbitrary subprofile $S$ of $R$. We will closely follow the proof of \Cref{thm:propBordaAR} and thus define $b_i^S({\succ})=\frac{S({\succ})}{R({\succ})}b_i({\succ})=S({\succ})\cdot{m\choose 2}$ for all $i\in \{1,\dots, m\}$ and ${\succ}\in\mathcal{R}$. Next, we let $c_i^S({\succ})=b_i^S({\succ})-b_{i+1}^S({\succ})$ denote the payment made by the ranking $\succ$ with respect to its budget in $S$ and let $C_i^S=\sum_{{\succ}\in\mathcal{R}} c_i^S({\succ})$ be the total payment made by the subprofile $S$ in the $i$-th round. 

    Now, fix a round $i\in \{1,\dots, m-1\}$ in the execution of the Flow-adjusting Borda rule. We first observe that for every ranking $\succ$, it holds that $b_i({\succ})-b_{i+1}({\succ})=f_i({\succ})$ as the computed flow determines the payment. Furthermore, by the definition of $\rho_i$, we have that $\frac{f_i(s,v_{\succ})}{b_i({\succ})\cdot u({\succ},x_i, X_i)}\leq \rho_i$. Equivalently, this means that $\frac{b_i({\succ})\cdot u({\succ},x_i, X_i)}{f_i(s,v_{\succ})}\geq \frac{1}{\rho_i}$. By combining our insight, it follows that $b_i({\succ})\cdot u({\succ},x_i,X_i)\geq  \frac{b_i({\succ})-b_{i+1}({\succ}))}{\rho_i}$. Furthermore, by multiplying both sides with $\frac{S({\succ})}{R({\succ})}$, we get that $b_i^S({\succ})\cdot u({\succ},x_i,X_i)\geq \frac{c_i^S({\succ})}{\rho_i}$. Finally, summing over all the rankings ${\succ}\in\mathcal{R}$, we conclude that 
    \[\frac{C_i^S}{\rho_i}\leq \sum_{{\succ}\in\mathcal{R}} b_i^S({\succ})u({\succ},x_i,X_i)\leq \sum_{{\succ}\in\mathcal{R}} b_1^S({\succ})u({\succ},x_i,X_i).\]

    Moreover, by summing over all rounds $i\in \{1,\dots, m-1\}$, we infer that 
    \begin{align*}
        \sum_{i=1}^{m-1} \frac{C_i^S}{\rho_i}\leq \sum_{i=1}^{m-1} \sum_{{\succ}\in\mathcal{R}} b_1^S({\succ})u({\succ},x_i,X_i)=\sum_{{\succ}\in\mathcal{R}} b_1^S({\succ})u({\succ},\rhd).
    \end{align*}

   In turn, our upper bounds on $\rho_i$ proven in Step 1 show that
    \begin{align*}
        \sum_{i=1}^{m-3} \frac{C_i^S(m-i)(m-i+1)}{4} + C^S_{m-2}+C^S_{m-1}\leq \sum_{{\succ}\in\mathcal{R}} b_1^S({\succ})u({\succ},\rhd).
    \end{align*}

    Next, let $C^S=\sum_{i=1}^{m-1} C_i^S$ denote the total payments made by $S$ and note that $C^S\geq \sum_{{\succ}\in\mathcal{R}} b_1^S({\succ}) - \frac{3}{4}$ because we have shown in the proof of \Cref{thm:flowBordaPR} that the total remaining budget of \fb is at most $\frac{3}{4}$. Moreover, let $k$ denote the largest integer such that $C^S\geq\frac{k(k+1)}{2}$. Now, since the coefficients of our sum are weakly decreasing as $i$ increases, we minimize this term by assuming that $C$ is only distributed at the late payments. However, in each round $i$, it is only possible to pay at most $m-i$. Hence, the left-hand sum of the above inequality is minimized if $C^S_{i}=m-i$ for all $i\in \{m-1,\dots, m-k\}$ and $C_{m-k-1}^S=C^S-\frac{k(k+1)}{2}$. 
    
    We next proceed with a case distinction regarding $C^S$ (resp. $k$) and first suppose that $C^S\leq 3$. In this case, we have in the worst-case that $C_i^S=0$ for all $i\in \{1,\dots, m-3\}$, so we can simplify our inequality to $C^S\leq \sum_{{\succ}\in\mathcal{R}} b_1^S({\succ})u({\succ},\rhd)$. Further, by using that $C^S\geq \sum_{\succ\in\mathcal{R}} b_1^S(\succ)-\frac{3}{4}=|S|\cdot{m\choose 2}-\frac{3}{4}$ and dividing by $|S|\cdot{m\choose 2}$, this means that 
    \begin{align*}
        1-\frac{3}{2m(m-1) \cdot |S|}\leq \frac{1}{|S|}\cdot{m\choose 2}^{-1}\cdot\sum_{{\succ}\in\mathcal{R}} b_1^S({\succ}) u({\succ},\rhd)=\frac{1}{|S|}\sum_{{\succ}\in\mathcal{R}} S({\succ}) u({\succ},\rhd).
    \end{align*}

    Now, we first note that the bound of the theorem is trivial if $|S|\leq \frac{3}{4}\cdot {m\choose 2}^{-1}$ because then ${m\choose 2}\cdot \frac{|S|}{4}-\frac{3}{16}\leq 0$. We hence assume that $|S|>\frac{3}{4}\cdot {m\choose 2}^{-1}$. Moreover, our assumption that $C^S\leq 3$ implies that $|S|\leq (3+\frac{3}{4})\cdot {m\choose 2}^{-1}$. We will now show that for all these values of $|S|$ that ${m\choose 2}\cdot \frac{|S|}{4}-\frac{3}{16}\leq 1-\frac{3}{2m(m-1)\cdot|S|}$. By subtracting $1-\frac{3}{2m(m-1)\cdot|S|}$ form both sides and multiplying with $16|S|\cdot{m\choose 2}$, we infer that this is equivalent to
    \[4 {m\choose 2}^2 |S|^2 -19{m\choose 2}|S| + 12\leq 0.\]
    It can now be checked that $4 {m\choose 2}^2 |S|^2 -19{m\choose 2}|S| + 12=0$ if $|S|=\frac{3}{4}{m\choose 2}^{-1}$ or $|S|=4{m\choose 2}^{-1}$. As a quadratic function grows symmetrically from its minimum, this proves our inequality for $\frac{3}{4}{m\choose 2}^{-1}\leq |S|\leq (3+\frac{3}{4}){m\choose 2}^{-1}$, as required.

    As the second case, suppose that $C^S>3$ and thus $k\geq 2$. In this case, we have that 
    \begin{align*}
    &1+2+\sum_{i=m-k}^{m-3} \frac{(m-i)(m-i)(m-i+1)}{4}+\frac{(C^S-\frac{k(k+1)}{2})(m-(m-k-1))(m-(m-k-1)+1)}{4}\\
    &\leq C^S_{m-1}+C_{m-2}^S+\sum_{i=1}^{m-3} \frac{C_i^S(m-i)(m-i+1)}{4}.
    \end{align*}
    Moreover, since $\frac{2\cdot 2\cdot 3}{4}+\frac{1\cdot 1\cdot 2}{4}=\frac{7}{2}$, we can rewrite the left side of this inequality by 
    \begin{align*}
        &\sum_{i=m-k}^{m-1} \frac{(m-i)(m-i)(m-i+1)}{4}+\frac{(C^S-\frac{k(k+1)}{2})(m-(m-k-1))(m-(m-k-1)+1)}{4}-\frac{1}{2}\\
        &=\sum_{i=1}^{k} \frac{i^2(i+1)}{4}+\frac{(C^S-\frac{k(k+1)}{2})(k+1)(k+2)}{4}-\frac{1}{2}.
    \end{align*}

    Next, let $\ell=C^S-\frac{k(k+1)}{2}$. As noted in the proof of \Cref{thm:propBordaAR}, it holds that 
    \begin{align*}
        &\sum_{i=1}^{k} i^2(i+1)+\ell(k+1)(k+2)\\
        &=\frac{k^4}{4} + \frac{5k^3}{6}+\frac{3k^2}{4}+\frac{k}{6}+\ell(k+1)(k+2)\\
        &=\left(\frac{k^4}{4}+\frac{2k^3}{4} + \frac{k^2}{4} + \ell k(k+1)) + \ell(k+1)\right) + \left(\frac{k^3}{3}+\frac{k^2}{2}+\frac{k}{6} + \ell(k+1)\right)
    \end{align*}

    Now, we first note that $\ell\leq k+1$, so $\ell(k+1)\geq \ell^2$. Further, as $k\geq 2$, it holds that $\frac{k^3}{3}+\frac{k}{6}\geq \frac{k^2}{2}+1$ and $k^2\geq \frac{k(k+1)}{2}+1$. Hence, we derive that 
    \begin{align*}
    \sum_{i=1}^{k} i^2(i+1)+\ell(k+1)(k+2)&\geq \left(\frac{k^4}{4}+\frac{2k^3}{4} + \frac{k^2}{4} + \ell k(k+1)) + \ell^2\right) + \left(k^2+1+ \ell(k+1)\right)\\
    &\geq \left(\frac{k(k+1)}{2}+\ell\right)^2+\left(\frac{k(k+1)}{2}+\ell\right)+2\\
    &=C^S(C^S+1)+2.
\end{align*}

Substituting this into our original inequality shows that $\frac{C^S(C^S+1)}{4}\leq \sum_{{\succ}\in\mathcal{R}} {m\choose 2}\cdot S(\succ)\cdot u({\succ}, \rhd)$. From here on, we can complete the proof analogously to the proof of \Cref{thm:propBordaAR}.
\end{proof}

\end{document}